\pdfoutput=1 % For arXiv's pdfLaTeX compiler
\documentclass[11pt]{article}

\usepackage[utf8]{inputenc}
\usepackage{amsmath,amsfonts,amssymb,amsthm} % Math formatting
\usepackage{graphicx}                         % Figures
\usepackage{geometry}                         % Margins
\usepackage{hyperref}                         % Clickable links
\usepackage{booktabs}                         % Clean tables (Springer style)
\usepackage{authblk}
\usepackage[numbers,sort&compress]{natbib}
\usepackage{amsmath,amssymb,amsfonts}%
 \usepackage{amsthm}%
\usepackage{xpatch}

\xpatchcmd{\proof}{\itshape}{\bfseries}{}{}
\newenvironment{myproof}{
  \par\noindent\textbf{Proof.}\hspace{0.5em}
}{
  \hfill$\square$\par
}

\usepackage{mathrsfs}%
\usepackage[title]{appendix}%
\usepackage{xcolor}%
\usepackage{textcomp}%
\usepackage{manyfoot}%

\usepackage{comment}
\usepackage{xspace}
\usepackage{adjustbox}
\usepackage{multirow}
\usepackage[normalem]{ulem}
\usepackage{caption}
\usepackage{subcaption}
\usepackage{todonotes}
\usepackage{xspace} % Required for the \xspace command

\newcommand{\Leb}[1]{\mathit{L}_{#1}}
\newcommand{\Rel}{\mathbb{R}}
\newcommand{\N}{\mathbb{N}}
\newcommand{\Exp}{\mathbb{E}}
\newcommand{\xl}{\mathcal{X}_l}
\newcommand{\dl}{\mathcal{D}_l}

\newcommand{\BB}{B}
\newcommand{\Z}{\mathbf{Z}}
\newcommand{\Zh}{\mathbf{\hat{Z}}}

\newcommand\norm[1]{\left\lVert#1\right\rVert} %norm
\def\argmin{\mathop{\rm arg\,min}}

\newcommand\numberthis{\addtocounter{equation}{1}\tag{\theequation}}
\newcommand{\order}[1]{$\mathcal{O}\left(#1\right)$}

\DeclareMathOperator{\srt}{SR}
\newcommand{\sr}[3]{\srt_{#1,#2}(#3)}
\newcommand{\ubsr}{\srt_{l,\lambda}(X)}
\newcommand{\srth}[1]{\srt(\theta_{#1})}
\newcommand{\srm}[1]{\srt^{m}(#1)}
\newcommand{\srmk}[2]{\srt^{m_{#1}}(#2)}

\newcommand{\J}{J^{m_1,m_2}}
\newcommand{\jt}{\J(\mathbf{\hat{z}}_\theta,\mathbf{z}_\theta,\mathbf{v}_\theta)}
\newcommand{\Jt}{\J(\mathbf{\hat{Z}}_\theta,\mathbf{Z}_\theta,\mathbf{V}_\theta)}
\newcommand{\A}{\Exp \left[ l'(-F(\theta,\xi)-\srth{}) \nabla F(\theta,\xi) \right]}
\newcommand{\B}{\Exp \left[ l'(-F(\theta,\xi)-\srth{}) \right]}
\newcommand{\Ap}{-\frac{1}{m_2}\sum_{j=1}^{m_2} \Big[ l'(-\mathrm{Z}^j_\theta-\srth{}) \mathbf{V}^j_\theta\Big]}
\newcommand{\Bp}{\frac{1}{m_2}\sum_{j=1}^{m_2} \big[ l'(-\mathrm{Z}^j_\theta-\srth{}) \big]}
\newcommand{\Am}{-\frac{1}{m_2}\sum_{j=1}^{m_2} \Big[ l'(-\mathrm{Z}^j_\theta-\srmk{1}{\mathbf{\hat{Z}_\theta}}) \mathbf{V}^j_\theta\Big]}
\newcommand{\Bm}{\frac{1}{m_2}\sum_{j=1}^{m_2} \big[ l'(-\mathrm{Z}^j_\theta-\srmk{1}{\mathbf{\hat{Z}_\theta}}) \big]}

\DeclareMathSymbol{\shortminus}{\mathbin}{AMSa}{"39}

\usepackage{booktabs} 
\usepackage{threeparttable} 

\usepackage[noend]{algorithmic}
\usepackage{algorithm}

\usepackage{thm-restate}

\allowdisplaybreaks

\newtheorem{theorem}{Theorem}

\newtheorem{assumption}{Assumption}

\newtheorem{lemmma}[theorem]{Lemma}

\newtheorem{definition}[theorem]{Definition}
\newtheorem{proposition}[theorem]{Proposition}
\newtheorem{remark}[theorem]{Remark}

\usepackage{booktabs} % for nice tables

\usepackage{threeparttable} % for table footnotes

\allowdisplaybreaks

\usepackage{cleveref}
\crefname{enumi}{assumption}{assumptions}
\crefname{subsection}{section}{sections}
\crefname{lemma}{lemma}{lemmas}
\crefname{figure}{figure}{figures}
\crefname{assumption}{assumption}{assumptions}
\crefname{lemmma}{lemma}{lemmas}
\crefname{corollary}{corollary}{corollaries}

\begin{document}

%%=============================================================%%
%% Title & Running Header                                      %%
%%=============================================================%%
\title{Gradient-based Stochastic Optimization of Utility-based Shortfall Risk}

%%=============================================================%%
%% Authors and Affiliations (Modern Springer Format)          %%
%%=============================================================%%

\author[1]{Sumedh Gupte\thanks{\texttt{sumedh.gupte@tcs.com}}}
\author[2]{Prashanth L. A.\thanks{\texttt{prashla@cse.iitm.ac.in}}}
\author[1]{Sanjay P. Bhat\thanks{\texttt{sanjay.bhat@tcs.com}}}

\affil[1]{TCS Research, IIIT-H Research Park, Hyderabad, 500032, India}
\affil[2]{Department of Computer Science and Engineering, Indian Institute of Technology Madras, Chennai, 600036, India}
\date{\today} % Leaves the date blank, or use \date{\today}
\maketitle
%%=============================================================%%
%% Abstract & Keywords (Moved inside preamble layout)          %%
%%=============================================================%%
\abstract{
We consider the problems of estimation and optimization of utility-based shortfall risk (UBSR). We extend UBSR to cover possibly unbounded random variables. We cover prominent risk measures such as entropic risk, expectile risk, Value-at-Risk, and quadratic risk as special cases of the UBSR. In the context of estimation, we derive non-asymptotic bounds on the mean absolute error (MAE) and the mean-squared error (MSE) of the classical sample-average approximation (SAA) estimator for the UBSR. In the context of optimization, we derive an expression for the gradient of UBSR under a smooth parameterization. We propose a gradient estimator for the UBSR and derive non-asymptotic bounds on MAE and MSE for this estimator. We incorporate the aforementioned gradient estimator into a stochastic gradient (SG) optimization algorithm and derive non-asymptotic bounds on the convergence rate of our SG algorithm for optimizing UBSR under three objectives, namely, strongly convex, convex and non-convex. Finally, we conduct experiments on financial applications to demonstrate the performance of our proposed UBSR estimation and optimization algorithms.
}

%%=============================================================%%
%% Main Text                                                   %%
%%=============================================================%%

\section{Introduction}
\label{sec:intro}
Optimizing risk is important in several application domains, e.g., finance, transportation, healthcare, robotics to name a few. Financial applications rely heavily on efficient risk assessment techniques and employ a multitude of risk measures for risk estimation. Risk optimization involves risk estimation as a sub-procedure for solving decision-making problems in finance. Value-at-Risk (VaR)~\cite{jorion1997value,basak-shapiro-var}, and Conditional Value-at-Risk (CVaR)~\cite{Uryasev2001,rockafellar2000optimization} are two popular risk measures. The risk measure VaR, which is a quantile of the underlying distribution, is not the preferred choice because it is not sub-additive \cite{follmer_axiomatic_2015}. In a financial context, the sub-additivity property implies that diversification does not increase risk. CVaR as a risk measure satisfies the sub-additivity property and falls in the category of coherent risk measures~\cite{ACERBI20021487}. However, CVaR is not desirable as a risk measure because it is not invariant under randomization, and it is not sensitive to heavy-tailed losses \cite{giesecke-risk-large-losses}. Furthermore, previous works (cf. \cite{follmer_axiomatic_2015}) in the literature have questioned the relevance of the positive homogeneity property of coherent risk measures, from a financial application viewpoint. More precisely, in finance parlance, an acceptable financial position may not necessarily be acceptable after scaling by any arbitrary factor. 

A class of risk measures that subsumes coherency, and does not enforce positive homogeneity, is convex risk measures~\cite{FollmerSchied2004}. A prominent family of convex risk measures is the utility-based shortfall risk (UBSR). 
UBSR, introduced by \cite{FollmerSchied2004}, is a law-invariant \cite{Kusuoka2001}, convex risk measure that has a gained prominence lately \cite{weber-2006-distribution-invariant,dunkel_efficient_2007,dunkel2010stochastic,follmer2016stochastic,zhaolin2016ubsrest,Hegde2024,guo_distributionally_2019,gupte2023optimization}. More precisely, UBSR emerges as one among many of the families of convex risk measures that are induced by the robust Savage representation \cite[Theorem 2.78]{follmer2016stochastic}. UBSR is the only law-invariant, convex risk measure that is elicitable \cite{Bellini_2015}. It has a few advantages over the popular CVaR risk measure, namely (i) UBSR is invariant under randomization, while CVaR is not, see \cite{dunkel2010stochastic}; (ii) unlike CVaR, which only considers the values that the underlying random variable takes beyond VaR, the loss function in UBSR can be chosen to encode the risk preference for each value that the underlying random variable takes. Thus, in both risk estimation and optimization, UBSR is a more desirable alternative to the industry-standard risk measures, namely VaR and CVaR. 
\paragraph{Our contributions.}
In this paper, we consider the problems of estimation and optimization of the UBSR risk measure. 
First, we extend the UBSR measure to cover unbounded random variables that satisfy certain integrability requirements, and establish conditions under which the UBSR is a convex risk measure. Our results are stated under notably weaker assumptions on the risk parameters, namely, the loss function $l$ and the risk threshold $\lambda$. %This weakening of assumptions allows our estimation procedures to include VaR and CVaR as special cases of the UBSR and OCE respectively.
 Second, for a sample average approximation (SAA) of UBSR, which is proposed earlier in the literature, we present a novel proof under a variance assumption to obtain a mean absolute error (MAE) bound and a mean-squared error (MSE) bound of the order $\mathcal{O}(1/\sqrt{m})$ and $\mathcal{O}(1/m)$ respectively. Here, $m$ denotes the number of independent and identically distributed (i.i.d.) samples of the underlying distribution used to form the estimates.
    %\item For a SAA estimator of OCE, we obtain MAE and MSE bounds of the order \order{1/\sqrt{m}} each, for a choice of Lipschitz utility function. For the non-Lipschitz case, we obtain MAE and MSE bounds of the order $\mathcal{O}(1/m^{1/4})$ and $\mathcal{O}(1/\sqrt{m})$ respectively. These bounds are obtained under a fairly general setting, without assuming that the utility function is strongly convex or smooth. Using a new proof technique and under some mild assumptions, we obtain an MAE bound of order $\mathcal{O}(1/\sqrt{m})$. 
  Third, for the UBSR optimization problem with a vector parameter, we derive an expression for the UBSR gradient. Using these expressions, we propose an $m$-sample gradient estimator for UBSR and establish MAE and MSE bounds of \order{1/\sqrt{m}} and $\mathcal{O}(1/m)$, respectively, for the estimation error.
   Finally, we design a stochastic gradient (SG) algorithm for the aforementioned gradient estimator. For the proposed SG algorithm, we derive a non-asymptotic bound of $\mathcal{O}(1/n)$ under a strong convexity assumption on the risk objective, where $n$ denotes the number of iterations of the SG algorithm. For the convex and non-convex cases, we derive non-asymptotic bounds of $\mathcal{O}(1/\sqrt{n})$.

\paragraph{Related work. }Convex risk measures have been extensively analyzed, under the assumption that the random variables (r.v.) are bounded \cite{artzner-coherent-risk-measures,FollmerSchied2004}. The authors in \cite{kaina_convex_2009} were the first to investigate convex risk measures for the case of unbounded random variables, with a focus on continuity and representational properties. The analysis by \cite{follmer_axiomatic_2015}, on the other hand, covered other aspects of convex risk measures, such as elicitability and robustness, and extended it to unbounded random variables. The aforementioned works cover unbounded random variables and study theoretical properties of convex risk measures.
In contrast to the aforementioned works, we focus on estimating and optimizing the UBSR within a stochastic optimization framework. In other words, our emphasis is on the computational aspects of UBSR. Specifically, we propose algorithms for estimating and optimizing UBSR and present associated non-asymptotic analyses, i.e., error bounds on UBSR estimation and convergence rates for UBSR optimization.
Previous works on UBSR estimation and optimization consider the case of bounded r.v. and presume that the extension to the unbounded r.v. case is undemanding \cite{dunkel2010stochastic,zhaolin2016ubsrest}. Therefore, a rigorous extension to the unbounded r.v. is missing in the literature, and our work addresses this gap. Our results show that this extension is non-trivial.

The authors in \cite{FollmerSchied2004} introduced the UBSR risk measure for bounded random variables. Existing works in the literature have analyzed UBSR and its properties, and proposed two schemes for UBSR estimation, namely stochastic approximation (SA) and sample average approximation (SAA) (see \cite{dunkel_efficient_2007,dunkel2010stochastic,giesecke-risk-large-losses,weber-2006-distribution-invariant,zhaolin2016ubsrest} for the details). %In real-world financial markets, the financial positions continuously evolve over time, and so must their risk estimates. In~\cite{weber-2006-distribution-invariant}, the authors showed that UBSR can be used for dynamic evaluation of such financial positions. 
Following the SA approach, \cite{dunkel2010stochastic} proposed estimators based on a stochastic root finding procedure and provided asymptotic convergence guarantees. Reference \cite{zhaolin2016ubsrest} used an SAA procedure for UBSR estimation and established asymptotic convergence guarantees on the estimator. They proposed an estimator for the UBSR derivative that can be used for risk optimization under a scalar decision parameter. They show that this estimator of the UBSR derivative is asymptotically unbiased. Reference \cite{Hegde2024} gave a non-asymptotic analysis for the scalar UBSR optimization while employing a stochastic root-finding technique for UBSR estimation. In comparison to these works, we would like to note the following aspects.
\begin{enumerate}
    \item Unlike \cite{dunkel2010stochastic,zhaolin2016ubsrest}, we provide non-asymptotic bounds on the mean absolute error and the mean squared error of the UBSR estimate from a procedure that is computationally efficient.
    \item We consider UBSR optimization for a vector parameter, while earlier works (cf. \cite{zhaolin2016ubsrest,Hegde2024}) consider the scalar case.
    \item We analyze a SG-based algorithm in the non-asymptotic regime for UBSR optimization, while \cite{zhaolin2016ubsrest} provide an asymptotic guarantee for the UBSR derivative estimate.
    \item In \cite{Hegde2024}, UBSR optimization using a gradient-based algorithm has been proposed for the case of scalar parameterization. Unlike \cite{Hegde2024}, we derive a general (multivariate) expression for the UBSR gradient, leading to an estimator that is subsequently employed in a stochastic gradient algorithm mentioned above. A vector parameter makes the bias/variance analysis of UBSR gradient estimate challenging as compared to the scalar counterpart that is analyzed by \cite{Hegde2024}.
\end{enumerate}

Some prior works \cite{guo_distributionally_2019,delage_shortfall_2022,zhang_preference_2022} have also analyzed the UBSR measure within a robust optimization framework.
In \cite{guo_distributionally_2019}, the authors consider a variant of UBSR that incorporates ambiguity over the underlying distribution. They employ the SAA approach for optimizing this robust variant of UBSR, while working with a fixed dataset. They provide sample complexity bounds assuming the underlying distributions are discrete and have bounded support. 
Next, in \cite{delage_shortfall_2022}, the authors formulate a variant of UBSR that optimizes worst-case risk over an ambiguity set of loss functions. They provide a tractable formulation of the UBSR minimization problem under the assumption that the objective is convex and linear, and the underlying random variable is discrete. In \cite{zhang_preference_2022}, the authors extend the aforementioned formulation to incorporate multivariate loss functions. Both \cite{delage_shortfall_2022,zhang_preference_2022} work with an offline dataset and employ the SAA approach for finding a robust UBSR solution.
Unlike \cite{guo_distributionally_2019,delage_shortfall_2022,zhang_preference_2022}, we note the following aspects: (i) we do not require the underlying distribution to be discrete nor require bounded support; (ii) they impose convexity on the UBSR objective, while we provide results in the non-convex regime as well; (iii) we consider a more general online setting where samples from the underlying random variable are available, while they consider a setting where the exact form of the random variable, as a function of the underlying parameter, is available. In particular, UBSR gradient estimation is more challenging in our online framework than in the settings of the aforementioned references. Finally, we note that the guarantees for the SAA approach, as described in \cite[Chapter 5]{shapiro_book}, are asymptotic in nature, while we establish non-asymptotic bounds for our gradient-based UBSR optimization algorithms.  

There are several works in the literature that consider optimization of a smooth function using a stochastic gradient algorithm that is given inputs from an inexact gradient oracle, cf. \cite{bhavsar2021nonasymptotic,karimi2019non,stoc-bias-reduced-gradient-neurips,chen_closing_2021,devolder_stochastic_2011,duchi_finite_2012,hu_biased_2020,hu_bias-variance-cost_2021,pasupathy_sampling_2018}. However, the results from the aforementioned references are not directly applicable for UBSR optimization and the reader is referred to Section 5.2.3 of \cite{Hegde2024} for a detailed discussion. In contrast, the result we present for a stochastic gradient algorithm with biased gradient information is sufficiently general to apply to UBSR optimization.

A preliminary version of this manuscript was published in the IEEE Conference on Decision and Control (CDC) 2024, see \cite{gupte2023optimization}.
Compared to the conference version, this manuscript includes the following.
(i) new theoretical results and improved finite-sample bounds for  estimation of UBSR as well as its gradient;
(ii) UBSR optimization algorithms and finite-sample bounds for these algorithms in three different regimes, corresponding to the strongly convex, the convex, and the non-convex UBSR objective functions;
(iii) detailed simulation experiments; and (iv) a significantly revised presentation along with all the proofs.

\paragraph{Organization of the paper.}
 In \Cref{sec:ubsr-unbounded}, we characterize UBSR for a class of possibly unbounded random variables, derive certain useful properties, and provide popular examples for the UBSR loss functions that reduce UBSR to some well-known risk measures. In \Cref{sec:estimation}, we present SAA-based UBSR estimation techniques and derive the corresponding estimation error bounds. In \Cref{sec:ubsr-opt}, we derive the gradient expressions for UBSR, propose sample-based gradient estimators, and derive non-asymptotic bounds on their estimation errors. We then employ these estimators in a stochastic gradient (SG) scheme for risk optimization formulation and derive non-asymptotic convergence rates on the last iterate of the SG scheme. In \Cref{sec:experiments}, we present simulation experiments for estimation and optimization of UBSR risk measures. In \Cref{section:conclusions}, we provide the concluding remarks. 
 %The supplementary material \cite{prashla_web_resource} presents auxiliary theoretical results and additional simulation experiments.

\paragraph{Notation.}
We use boldface ($\mathbf{v}$), uppercase ($X$), and a combination of boldface and uppercase ($\mathbf{Z}$) to denote vectors, random variables, and random vectors, respectively. We use 'Var' as an abbreviation for variance, not to be confused with Value-at-Risk, which is abbreviated as 'VaR'. The terms $x^+$ and $x^-$ indicate $\max\left\{x,0\right\}$ and $\max\left\{-x,0\right\}$, respectively. We use $\log_b$ to denote the logarithm to the base $b$, and $\log$ to denote the natural logarithm. 

We use $\langle \cdot, \cdot \rangle$ to denote the dot product between two vectors, that is, for vectors $\mathbf{u}$ and $\mathbf{v}$ of the same dimension, $\langle \mathbf{u}, \mathbf{v} \rangle = \mathbf{u}^\mathrm{T}\mathbf{v}$. For $p\in[1,\infty)$, the $p$-norm of a vector $\mathbf{v} \in \Rel^d$ is given by $\norm{\mathbf{v}}_p \triangleq \left( \sum_{i=1}^d |\mathbf{v}_i|^p \right)^\frac{1}{p}$,
while $\norm{\mathbf{v}}_\infty$ denotes the supremum norm. Matrix norms \cite[Section 5.6]{horn_matrix_2012} induced by the vector $p$-norm are also denoted by $\norm{\cdot}_p$, where the special cases of $p=1,p=2$ and $p=\infty$ equal the maximum absolute column sum, spectral norm, and maximum absolute row sum, respectively.

Let $(\Omega, \mathcal{F}, \mathit{P})$ be a probability space. Let $\Leb{0}$ denote the space of $\mathcal{F}$-measurable, real random variables and let $\Exp(\cdot)$ denote the expectation under $\mathit{P}$. For $p \in [1,\infty)$, $\big(\Leb{p}, \norm{\cdot}_{\Leb{p}}\big)$ denotes the normed vector space of random variables $X: \Omega \to \Rel$ in $\Leb{0}$ for which 
$\norm{X}_{\Leb{p}} \triangleq \left( \Exp \big[ |X|^p \big] \right)^\frac{1}{p}$ is finite. Further, $\big(\Leb{\infty}, \norm{\cdot}_{\Leb{\infty}}\big)$ denotes the normed vector space of random variables $X: \Omega \to \Rel$ in $\Leb{0}$, for which, $\norm{X}_{\Leb{\infty}} \triangleq \inf \{k \in \Rel : |X| \leq k \text{ a.s.} \}$ is finite. Let $p\in[1,\infty)$ and let $\mathbf{Z}$ be a random vector such that each $Z_i$ is $\mathcal{F}$-measurable and has finite $p^{th}$ moment. Then the $\Leb{p}$-norm of $\Z$ is defined by $\norm{\mathbf{Z}}_{\Leb{p}} \triangleq \left( \Exp \Big[ \norm{\mathbf{Z}}_p^p\Big] \right)^\frac{1}{p}$. 

% Let $\mu_X$ and $\mu_Y$ denote the marginal distributions of random variables $X$ and $Y$ respectively. Let $\mathcal{H}(\mu_{{X}},\mu_{{Y}})$ denote the set of all joint distributions having $\mu_{{X}}$ and $\mu_{{Y}}$ as the marginals. Then, for every $p\geq 1,\mathcal{T}_p(\mu_{{X}},\mu_{{Y}}) \triangleq \inf \left\{ \int \norm{x-y}^p \eta(dx,dy) : \eta \in \mathcal{H}(\mu_{{X}},\mu_{{Y}}) \right\}$ denotes the optimal transpost cost associated with $X$ and $Y$, and $\mathcal{W}_p(\mu_{{X}},\mu_{{Y}}) = \left(\mathcal{T}_p(\mu_{{X}},\mu_{{Y}})\right)^{1/p}$ denotes the $p^\textrm{th}$ Wasserstein distance \cite{panaretos_invitation_2020}.

Given a real-valued function $l:\Rel \to \Rel$, $\mathcal{X}_l \subseteq \Leb{0}$ denotes the space of random variables $X$, for which $l(-X-t)$ is integrable for every $t \in \Rel$. The risk measures that we consider in this paper are well-defined when $X \in \mathcal{X}_l$, i.e., when $\Exp\left[l(-X-t)\right]$ is finite for all $t \in \Rel$. When the random variable $X$ is unbounded, the finiteness of the above expectation depends not only on $X$, but also on $l$. This dependency motivates the use of the $\mathcal{X}_l$ notation above. 
% The set $\mathcal{X}_f$ satisfies the following property: 
% \begin{equation}\label{eq:xl-self-membership}
%     X+c \in \mathcal{X}_f \textrm{ for every } X \in \mathcal{X}_f,\textrm{ and every } c \in \Rel.
% \end{equation}
To see how the choice of $l$ affects the class of random variables $\mathcal{X}_l$, consider $l(x) = [x^+]^2$ and suppose that $X$ has a Student's t-distribution with degrees of freedom $\nu=2$. Then, $\Exp\left[l(-X-t)\right]$ diverges to $-\infty$ for every $t \in \Rel$, and therefore $X \notin \mathcal{X}_l$.  

Throughout this paper, "increasing" means "non-decreasing". To imply strictly increasing, we explicitly use the qualifier 'strict'.

\section{UBSR for Unbounded Random Variables}
In this section, we extend UBSR to a class of possibly unbounded random variables and derive several properties of UBSR that are known to hold in the bounded case, for instance, monotonicity, cash-invariance, and convexity of UBSR, finiteness of UBSR, and existence of UBSR as a root of a decreasing function. The existing results in the literature do not readily extend to unbounded r.v. in this case, necessitating a rigorous re-derivation.    
Subsequently, we provide examples of loss functions for which the UBSR specializes to some well-known risk measures. The benefits of employing UBSR vis-à-vis CVaR/VaR are well known, and we avoid a detailed discussion. The reader is referred to \cite{follmer2016stochastic,giesecke-risk-large-losses,weber-2006-distribution-invariant,ben-tal_old-new_2007} for further reading. Our goal is to estimate and optimize UBSR for a class of unbounded random variables, while aspects of UBSR related to risk attitude are already well understood. The results in this section are presented in the context of this goal.

\textbf{Convex risk measures.} We now briefly discuss the properties that capture the features that decision makers prefer in a risk measure. Let $\mathcal{X}$ be an arbitrary set of random variables. We define the notions of monetary and convex risk measures below~\cite{FollmerSchied2004,artzner-coherent-risk-measures}.
\begin{definition}
    A mapping $\rho: \mathcal{X} \to \Rel$ is called a monetary measure of risk if it satisfies the following two conditions.
\begin{enumerate}
    \item Monotonicity: $\forall X_1, X_2 \in \mathcal{X}$ with $X_1 \leq X_2$ a.s., we have $\rho(X_1) \geq \rho(X_2)$.
    \item Cash invariance: For all $X \in \mathcal{X}$ and $m \in \Rel$, we have $\rho(X + m) = \rho(X) - m$.
\end{enumerate}
\end{definition}
\begin{definition}
A monetary risk measure $\rho$ is convex, if $\mathcal{X}$ is convex and $\forall X_1,X_2 \in \mathcal{X}$ and $\alpha \in [0,1]$, we have $ \rho(\alpha X_1 + (1-\alpha)X_2) \leq \alpha \rho(X_1) + (1-\alpha)  \rho(X_2)$.
% \begin{equation}\label{eq:rho_convex}
% \rho(\alpha X_1 + (1-\alpha)X_2) \leq \alpha \rho(X_1) + (1-\alpha)  \rho(X_2).    
% \end{equation}
\end{definition}
\label{sec:ubsr-unbounded}
Throughout this paper, $l:\Rel \to \Rel$ denotes a loss function. The loss function $l$, and a threshold $\lambda$ are chosen by the decision maker who is interested in quantifying the risk of a random variable $X \in \xl$. Here, $\lambda$ lies in the interior of the range of $l$. Throughout this paper, $X$ is assumed to model gains; therefore, a higher value for $X$ is considered to be preferable and less risky. We now formalize the notion of UBSR \cite{FollmerSchied2004,follmer_convex_2002} below.
\begin{definition}\label{eq:definition-UBSR}
The risk measure UBSR for the loss function $l$ and risk threshold $\lambda$, is given by the function $SR_{l,\lambda}:\xl \to \Rel$, defined as
\begin{equation*}
     \ubsr \triangleq \inf \{\; t \in \Rel \;|\; \Exp[l(-X-t)] \leq \lambda \} .
\end{equation*}    
\end{definition}
As an example, with $l(x)=\exp(\beta x)$ and $\lambda = 1$,  $\ubsr$ is identical to the entropic risk measure \cite{follmer2016stochastic}, which is a popular risk measure that enjoys several advantages over the standard risk measures VaR and CVaR.

Following \cite{artzner-coherent-risk-measures}, we define the \textit{acceptance set} associated with the UBSR risk measure as $\mathcal{A}_{l,\lambda} = \{ X \in \xl : \ubsr \leq 0\}$.
% \begin{equation*}\label{eq:def-acceptance_set}
%     .
% \end{equation*}
Note that the set $\mathcal{A}_{l,\lambda}$ contains all random variables $X$ whose expected loss $\Exp[l(-X)]$ does not exceed $\lambda$.

 \subsection{Characterization of UBSR}
 In this section, we discuss the problem of computing $\ubsr$ of a random variable $X \in \xl$, for a given loss function $l$ and risk threshold $\lambda$. Consider the real-valued function $g_X:~\Rel~\to~\Rel$ associated with the random variable $X$, defined by
\begin{equation}\label{eq:def-g}
    g_X(t) \triangleq \Exp\left[ l(-X-t)\right] - \lambda.
\end{equation}
The following proposition establishes that $\ubsr$ is a root of the function $g_X$ defined in \cref{eq:def-g}. 
\begin{proposition}\label{proposition:ubsr-as-a-root}
Suppose $X \in \xl$, the loss function $l$ is non-constant, increasing, and one of the following holds:\\
(A1) $l$ is continuous; or (A2) $l$ is continuous a.e., and the CDF of $X$ is continuous.\\
% \begin{enumerate}
% \renewcommand{\labelenumi}{(A\arabic{enumi})}
%     \item \label{l-continuous-everywhere} $l$ is continuous.
%     \item \label{l-continuous-ae} $l$ is continuous a.e., and the CDF of $X$ is continuous.
% \end{enumerate}
Then, $g_X$ is continuous and decreasing. In addition, if there exists $t^{\mathrm{u}}_{X}, t^{\mathrm{l}}_{X} \in \Rel$ such that $g_X(t^{\mathrm{u}}_X)\leq 0 < g_X(t^{\mathrm{l}}_X)$, then $\ubsr$ is finite and a root of $g_X(\cdot)$.
\end{proposition}
\begin{myproof}The proof is split into three parts, where we show a) $g_X$ is decreasing, b) $g_X$ is continuous, and c) $\ubsr$ is finite and a root of $g_X$. 

We begin with the monotonicity of $g_X$. Take any $t_1,t_2 \in \Rel$ such that $t_1 < t_2$. As $l$ is increasing, we have 
\begin{equation}\label{eq:monotonicity-of-l}
 l(-X-t_2) \leq l(-X-t_1), \;\; \text{w.p.} \;1   
\end{equation}
This implies that $\Exp\left[l(-X-t_2)\right] \leq \Exp\left[l(-X-t_1)\right]$, and therefore, $g_X(t_2) \leq g_X(t_1)$ holds. This concludes the proof that $g_X$ is decreasing.

Next, we show that the continuity of $g_X$ holds under either of the two assumptions A1 and A2 of the proposition.

\textit{Assumption A1: The loss function $l$ is continuous.}
Let $t_0 \in \Rel$ and define $Y \triangleq l(-X-t_0)$. Let $t_n \uparrow t_0$ be any real and increasing sequence and define $Y_n \triangleq l(-X-t_n), \forall n$. Note that the following holds: (a) $Y_n$ is integrable for all $n$, and $Y$ is integrable; and (b) $Y_n \downarrow Y$ almost surely. Verify that (a) follows because $X \in \xl$, whereas (b) follows because $l$ is continuous. Next, we define $Z_n \triangleq Y_1 - Y_n$ and claim that $Z_n \geq 0$ and $Z_n \uparrow (Y_1 - Y)$. The first part of the claim follows trivially, while the second part follows from (b). The claim implies that the conditions of the Monotone Convergence Theorem (MCT) \cite[Theorem 1.6.6]{durrett_2019} hold, and therefore, $\Exp\left[Z_n\right] \uparrow \Exp\left[Y_1 - Y\right]$, that is, $\Exp\left[Y_1 - Y_n\right] \uparrow \Exp\left[Y_1 - Y\right]$.
From the linearity of expectation and (a), it follows that $\Exp\left[Y_n\right] \downarrow \Exp\left[Y\right]$. This implies that $\lim_{t_n \uparrow t_0} g_X(t_n) = g_X(t_0)$. Using parallel arguments for a decreasing sequence $t_n \downarrow t_0$ and $Z_n \triangleq Y_n - Y_1$, we have $\lim_{t_n \downarrow t_0} g_X(t_n) = g_X(t_0)$. This concludes the proof that $g_X$ is continuous.

\textit{Assumption A2: The loss function $l$ is continuous a.e., and the CDF of $X$ is continuous.}
Let $t_0 \in \Rel$ and $\delta > 0$. Choose a real sequence $\{t_n\}_{n\geq 1}$ such that $t_n \in \left(t_0-\delta,t_0+\delta\right)$ and $t_n \to t_0$.  Define $\hat{Y} \triangleq \left|l(-X-t_0-\delta)\right| + \left|l(-X-t_0+\delta)\right|$ and $Y_n \triangleq l(-X-t_n),\forall n$. Then, $l$ being increasing implies that a) $\left|Y_n\right|\leq \hat{Y}, \forall n$ and $\Exp\left[\hat{Y}\right]< \infty$. Next, we define $Y \triangleq l(-X-t_0)$. Let $\dl$ denote the set of points at which $l$ is not continuous. Since the CDF of $X$ is continuous, $\mathit{P}\left(\{(-X-t_0) \in \dl\}\right) = 0$ and $\mathit{P}\left(\{(-X-t_n) \in \dl\}\right) = 0, \forall n$. Therefore, w.p. $1$, we have  $\lim_{n \to \infty} Y_n = \lim_{n \to \infty} l(-X-t_n) = l(-X-t_0) = Y$. In other words, we have b) $Y_n \to Y$ a.s. Then (a) and (b) satisfy the assumptions of the Dominated Convergence Theorem \cite[Theorem 1.6.7]{durrett_2019} and we have $\Exp\left[Y_n\right] \to \Exp\left[Y\right]$. This implies that $g_X(t_n) \to g_X(t_0)$, and therefore $g_X$ is continuous.

Finally, we show that $\ubsr$ is finite and is a root of $g_X$. Let $A \triangleq \left\{ t \in \Rel \left| g_X(t) \leq 0 \right. \right\}$. Then $A$ is non-empty as $t^u_X \in A$. By the definition of $\ubsr$, we have $\ubsr = \inf{A}$, and since $t> t^l_X,\forall t \in A$ we conclude that $A$ is bounded below and therefore $\ubsr$ is finite. There exists a decreasing sequence $t_n$ in $A$ such that $g_X(t_n) \leq 0$ and $t_n \to \ubsr$. Then, by continuity of $g_X, g_X(\ubsr) \leq 0$. Similarly, we define an increasing sequence $t_n \triangleq \ubsr - \frac{1}{n}$. Then by definition of $\ubsr$ as the infimum of $A$, we conclude that $t_n \notin A$, and therefore $g_X(t_n) > 0$. However, $t_n \to \ubsr$, and by continuity of $g_X$, $g_X(t_n) \to g_X(\ubsr)$, and therefore $g_X(\ubsr) \geq 0$. By combining the inequalities $g_X(\ubsr)\leq 0$ and $g_X(\ubsr)\geq 0$ we conclude that $\ubsr$ is a root of $g_X$.
\end{myproof}
% \begin{myproof}
%     See Appendix \ref{proof:proposition-ubsr-as-a-root} for the proof.
% \end{myproof}
Having established that $\ubsr$ is a root of $g_X$, we now present a proposition which, under slightly stronger assumptions, guarantees that $\ubsr$ is the unique root of $g_X$. 
\begin{proposition}\label{proposition:ubsr-unique-root}
Suppose $X \in \xl$. Suppose either of the following holds:\\
(A1') $l$ is continuous and strictly increasing; or (A2') $l$ is continuous a.e., non-constant, and increasing, and the CDF of $X$ is continuous and strictly-increasing.\\
% \begin{enumerate}
% \renewcommand{\labelenumi}{(A\arabic{enumi}')}
%     \item \label{l-strictly-increasing} $l$ is continuous and strictly increasing.
%     \item \label{F-strictly-increasing} $l$ is continuous a.e., non-constant, and increasing, and the CDF of $X$ is continuous and strictly-increasing.
% \end{enumerate}
Then, $g_X$ is continuous and strictly decreasing. In addition, if there exists $t^{\mathrm{u}}_{X}, t^{\mathrm{l}}_{X} \in \Rel$ such that $g_X(t^{\mathrm{u}}_X)\leq 0 < g_X(t^{\mathrm{l}}_X)$, then $\ubsr$ is finite and coincides with the unique root of $g_X(\cdot)$.
\end{proposition}
\begin{myproof}
We conclude from \Cref{proposition:ubsr-as-a-root} that $g_X$ is continuous under either of the two assumptions $A1'$ and $A2'$ of \Cref{proposition:ubsr-unique-root}. It is easy to see that if $l$ is strictly increasing then (\ref{eq:monotonicity-of-l}) holds with strict inequality and therefore $g_X$ is strictly decreasing under $A1'$. We now show that $g_X$ is strictly decreasing under $A2'$.

Let $t_1 < t_2$. Since $l$ is non-constant and increasing, there exists $x_1 < x_2$ such that $l(x_1) < l(x_2)$ and $x_2 \in \left(x_1,x_1 + (t_2-t_1)\right)$. Let $c \triangleq l(x_2)-l(x_1)$. Then, invoking Theorem 2.2.13 of \cite{durrett_2019} with $Y = l(-X-t_1) - l(-X-t_2)$ and $p=1$, we have
\begin{align*}
    &\Exp\left[Y\right] = \int_0^\infty \mathit{P}\left(\{Y > y\}\right) dy 
    \geq \int_0^{c} \mathit{P}\left(\{Y > y\}\right) dy 
    \geq \int\limits_{0}^{c} \mathit{P}\left(\{Y \geq c\}\right) dy  \\
    &= c \mathit{P}\left(\{Y \geq c\}\right) \geq c \mathit{P}\left(\{\left(l(-X-t_1) \geq l(x_2)\right) \cap  \left(l(-X-t_2) \leq l(x_1)\right)\}\right) \\
    &\geq c \mathit{P}\left(\{\left(-X-t_1 \geq x_2\right) \cap  \left(-X-t_2 \leq x_1\right)\}\right)  \\
    &= c \mathit{P}\left(\{-x_1-t_2 \leq X \leq -x_2-t_1\}\right) = c\left(F_X(-x_2-t_1)-F_X(-x_1-t_2)\right) > 0.
\end{align*}

The first two inequalities follow trivially. The third inequality follows because $l$ is increasing. The final inequality follows because $l(x_2)>l(x_1)$ and the CDF $F_X(\cdot)$ is strictly increasing. Indeed, $x_2 \in \left(x_1,x_1 + (t_2-t_1)\right)$, which implies that $-x_1-t_2 < -x_2-t_1$. Since $t_1<t_2$ were chosen arbitrarily, the result : $\Exp\left[Y\right]=\Exp\left[l(-X-t_1) - l(-X-t_2)\right]>0$ implies that $g_X(t_1)>g_X(t_2)$ holds for every $t_1<t_2$ and we conclude that $g_X$ is strictly decreasing.

From \Cref{proposition:ubsr-as-a-root} we know that $\ubsr$ is a root of $g_X$. Since we have established that $g_X$ is continuous and strictly decreasing, it must have exactly one root, and therefore $\ubsr$ is the unique root of $g_X$. This concludes the proof.
\end{myproof}
% \begin{myproof}
%     See \Cref{proof:proposition-ubsr-unique-root} for the proof.
% \end{myproof}
\begin{remark}
Compared to \Cref{proposition:ubsr-unique-root}, \Cref{proposition:ubsr-as-a-root} contains weaker assumptions on the loss function $l$ and the random variable $X$, and these assumptions suffice for showing that UBSR is a convex risk measure.  \Cref{proposition:ubsr-unique-root} on the other hand, is useful in the estimation of UBSR measure. %In particular, we present theoretical results for the case when the loss function satisfies \Cref{l-strictly-increasing}, and our experiments on risk estimation cover loss functions which satisfy \Cref{F-strictly-increasing}.
%\Cref{proposition:ubsr-as-a-root} on the other hand gives a weaker assertion, but under weaker assumptions, and therefore applies to a much broader class of loss functions. 
To the best of our knowledge, assertions similar to those made in \Cref{proposition:ubsr-as-a-root} are unavailable in the existing literature.
\end{remark}
% \begin{remark}\label{remark:ubsr-unique-root-conditions}
%     In the \Cref{proposition:ubsr-unique-root}, the r.v. $X \in \xl$ is chosen arbitrarily, and hence, the proposition applies to every r.v. in $\xl$, if \Cref{l-continuous-everywhere} holds. Instead, if \Cref{l-continuous-ae} holds, then the proposition applies to every r.v. in $\hat{\xl} \subset \xl,$ where $\hat{\xl} \triangleq \{X \in \xl \left| \textrm{The CDF of X is continuous and strictly increasing}\right.\}.$ For intelligibility reasons, we avoid stating two versions for each result (one for each condition). Instead of explicitly stating whether '\Cref{l-continuous-everywhere} and $X \in \xl$' holds or '\Cref{l-continuous-ae} and $X \in \hat{\xl}$' holds, we frame the results under the assertion that the assumptions of \Cref{proposition:ubsr-unique-root} hold. For notational convenience, we use the set $\xl$ throughout the paper. The set $\xl$ is to be construed as $\hat{\xl}$ whenever \Cref{l-continuous-ae} holds and \Cref{l-continuous-everywhere} fails.
% \end{remark}
\begin{remark}
Proposition 4.104 of \cite{FollmerSchied2004} shows that UBSR is the unique root for the case when the underlying random variables are bounded, and the loss function $l$ is convex and strictly increasing. In \Cref{proposition:ubsr-unique-root}, we generalize this result to unbounded random variables and, unlike \cite{FollmerSchied2004}, our proof does not require a convexity assumption on $l$. 
\end{remark}

% \begin{assumption}\label{as:l-convex}
%     The loss function $l$ is convex.
% \end{assumption}
% \begin{assumption}\label{as:loss-fn-polynomial-bound}
%      There exists $a_1\geq 0, q_1 \geq 1/2$ and $M_1 \geq 0$ such that the sub-derivative $ \partial l(x)$ satisfies: $ \partial l(x) \leq a_1 |x|^{q_1} + M_1, \forall x \in \Rel$.
% \end{assumption}

%The UBSR is a convex risk measure under \Cref{as:l-convex}. 
\subsection{Convexity of UBSR}
In \cite{FollmerSchied2004}, the authors have shown that UBSR is convex for the case of bounded random variables ($\mathcal{X} \subset \Leb{\infty}$). Using a novel proof technique in the following proposition, we extend the convexity of UBSR to $\xl$, a class of possibly unbounded random variables.
 
\begin{proposition}\label{proposition:sr-convex_risk_measure}
If $l$ is increasing, then $\xl$ is a convex set. Suppose the assumptions of \Cref{proposition:ubsr-as-a-root} hold. Then $SR_{l,\lambda}:\xl \to \Rel$ is a monetary risk measure. In addition, if $l$ is convex, then $\mathcal{A}_{l,\lambda}$ is a convex set and $\sr{l}{\lambda}{\cdot}$ is a convex risk measure.
\end{proposition}  
\begin{myproof}
We first establish that $\xl$ is a convex set. Recall that $\xl$ denotes the space of random variables $ X \in \Leb{0}$ for which the random variable $l(-X-t)$ is integrable for each $t \in \Rel$. Fix $\alpha \in \left[0,1\right]$. Since $l$ is increasing, for every $x,y \in \Rel$, we have $\min\left\{l(x), l(y) \right\} \leq l\left(\alpha x + (1-\alpha) y\right) \leq \max\left\{l(x), l(y) \right\}$ which implies that $\left|l\left(\alpha x + (1-\alpha) y\right)\right| \leq \left|l(x)\right| + \left|l(y)\right|$.
Let $X_1,X_2 \in \xl$. Fix $t \in \Rel$ and $\omega \in \Omega$. Substituting $x = -X_1(\omega)-t$ and $y = -X_2(\omega)-t$ in the last inequality above, we have   
\begin{align*}
    &\left|l\left(\alpha (-X_1(\omega)-t) + (1-\alpha) (-X_2(\omega)-t)\right)\right| \leq \left|l(-X_1(\omega)-t)\right| + \left|l(-X_2(\omega)-t)\right|.
\end{align*}
Since $X_1,X_2 \in \xl$, the r.h.s. in the above equation is integrable, and therefore the l.h.s is also integrable. Precisely, $l\left(-(\alpha X_1 + (1-\alpha) X_2)-t\right)$ is integrable for every $t \in \Rel$, which implies that $\alpha X_1 + (1-\alpha) X_2 \in \xl$. Since $\alpha \in [0,1], t \in \Rel,X_1$ and $X_2$ were chosen arbitrarily, we conclude that $\xl$ is convex. This proves the first assertion of the proposition. 

We now turn to proving that UBSR is a monetary risk measure, i.e., it satisfies the monotonicity and cash-invariance properties.
%let the assumptions of \Cref{proposition:ubsr-as-a-root} hold. 

We begin by showing that the monotonicity property is satisfied. Let $X_1,X_2\in \xl$ be such that $X_1 \leq X_2$ holds almost surely. Let $t_1 \triangleq SR_{l,\lambda}(X_1)$ and $t_2 \triangleq SR_{l,\lambda}(X_2)$. Since $l$ is increasing, we have $l(-X_1 - t_1) \geq l(-X_2 - t_1)$ almost surely. Taking expectation and subtracting $\lambda$ yields (a) $g_{X_1}(t_1) \geq g_{X_2}(t_1)$. By \Cref{proposition:ubsr-as-a-root}, we have $g_{X_1}(t_1) = g_{X_2}(t_2) = 0$, which combined with (a), yields $0 = g_{X_2}(t_2) \geq g_{X_2}(t_1)$. By definition of $t_2$ as $SR_{l,\lambda}(X_2)$, i.e., $t_2 = \inf\{t \in \Rel | g_{X_2}(t) \leq 0\}$, we conclude that $t_2\leq t_1$ this proves the monotonicity of $\sr{l}{\lambda}{\cdot}$.

We now show that the cash-invariance property is satisfied. Let $A \triangleq \left\{\hat{t} \in \Rel \left| g_X(\hat{t}) \leq 0 \right. \right\}$. Fix $m \in \Rel$. Let $X \in \xl$. Then $X+m \in \xl$ and by Definition \ref{eq:definition-UBSR}, we have
\begin{align*}
    &\sr{l}{\lambda}{X+m} = \inf \{\; t \in \Rel \;|\; \Exp[l(-(X+m)-t)] \leq \lambda \} \\
    % &= \inf \{\; t \in \Rel \;|\; \Exp[l(-X-(t+m))] \leq \lambda \} \\
    &= \inf \{\; t \in \Rel \;|\; t+m \in A \} 
    = \inf \{A\} - m = \ubsr - m.
\end{align*}
The first and fourth equalities follow from the definition of $\sr{l}{\lambda}{\cdot}$ in (\ref{eq:definition-UBSR}), while the second equality follows from the definition of $A$. The third equality is an identity on infimum that holds for every $m \in \Rel$ and every non-empty $A$. Thus we conclude that $\sr{l}{\lambda}{X+m} = \ubsr-m$.

We now prove the final assertion, i.e., UBSR is convex under the added assumption that $l$ is convex. To show the convexity of UBSR, we start by showing that $\mathcal{A}_{l,\lambda}$ is convex. Let $Y_1,Y_2 \in \mathcal{A}_{l,\lambda}$. Then by definition of $\mathcal{A}_{l,\lambda}$, we have (a) $\sr{l}{\lambda}{Y_1}\leq 0$ and $\sr{l}{\lambda}{Y_2}\leq 0$. From the fact that both $g_{Y_1}$ and $g_{Y_2}$ are decreasing functions, (a) implies that (b) $g_{Y_1}(0) \leq 0$ and $g_{Y_2}(0) \leq 0$.

Fix $\alpha \in \left[0,1\right]$ and denote $Y_\alpha \triangleq \alpha Y_1+ (1-\alpha) Y_2$ and $t_\alpha \triangleq SR_{l,\lambda}(Y_\alpha)$. Convexity of $\xl$ implies that $Y_\alpha \in \xl$. Then, we have
\begin{align*}
    &g_{Y_\alpha}(0) = \Exp\left[ l(\alpha (-Y_1) + (1-\alpha) (-Y_2))\right] - \lambda \\
    &\leq \alpha (\Exp\left[ l(-Y_1)\right]\lambda) + (1-\alpha)(\Exp\left[ l((-Y_2)) - \lambda \right]) = \alpha g_{Y_1}(0) + (1-\alpha)g_{Y_2}(0) \leq 0. 
\end{align*}
The first inequality follows from the convexity of $l$, while the last inequality follows from (b). Now consider the set $A \triangleq \left\{ t \in \Rel \left| g_{Y_\alpha}(t) \leq 0\right. \right\}$. The above claim $g_{Y_\alpha}(0) \leq 0$ implies that $0 \in A$, and by definition of $t_\alpha$ as the infimum of $A$, we conclude that $t_\alpha \leq 0$. This implies that $Y_\alpha \in \mathcal{A}_{l,\lambda}$. Since $\alpha$ was arbitrary, we conclude that $\mathcal{A}_{l,\lambda}$ is convex. 

Finally, we show that $SR_{l,\lambda}(\cdot)$ is convex. Let $X_1,X_2 \in \xl$. Then by definition of $\xl$, we have $X_i + \sr{l}{\lambda}{X_i} \in \xl; i \in \{1,2\}$. Moreover, cash invariance of  $SR_{l,\lambda}(\cdot)$ implies that $X_i + \sr{l}{\lambda}{X_i} \in \mathcal{A}_{l,\lambda}$ for each $i \in \{1,2\}$. 
 Then, by the convexity of $\mathcal{A}_{l,\lambda}$, we have 
\begin{align*}
    0 &\geq SR_{l,\lambda}\left( \alpha (X_1 + SR_{l,\lambda}(X_1)) + (1-\alpha )(X_2 + SR_{l,\lambda}(X_2))\right) \\
    %&= SR_{l,\lambda}\big(\alpha (X_1)+(1-\alpha)(X_2) + \alpha SR_{l,\lambda}(X_1) + (1-\alpha)SR_{l,\lambda}(X_2)\big)\\
    &= SR_{l,\lambda}\big(\alpha (X_1)+(1-\alpha)(X_2)\big) - \alpha SR_{l,\lambda}(X_1) - (1-\alpha)SR_{l,\lambda}(X_2),
\end{align*}
where the last equality follows from cash invariance. Since $\alpha, X_1$ and $X_2$ were chosen arbitrarily, it follows that $\sr{l}{\lambda}{\cdot}$ is convex.
\end{myproof}
\subsection{Popular instances of UBSR}
We next present a few popular choices for the UBSR loss function.
\begin{enumerate}
\item Value-at-Risk (VaR). Let $\lambda \in \Rel$, and choose the loss function to be the Heaviside function, i.e., $l(x) = \mathbf{1}_{\{x>0\}}, \forall x \in \Rel$. Then $\ubsr$, with the choice of $\lambda = \alpha$, coincides with $\textrm{VaR}_\alpha(X)$. The Value-at-Risk (VaR) at level $\alpha \in (0,1)$ for a random variable $X$ is given by \cite[Definition 4.45]{follmer2016stochastic}
\begin{equation*}
    \textrm{VaR}_\alpha(X) \triangleq  \inf\left\{ t \in \Rel \left| \textrm{Pr}\left(X+t < 0\right) \leq \alpha \right.\right\}.
\end{equation*}
See \cite{follmer2016stochastic,giesecke-risk-large-losses,dunkel_efficient_2007,zhaolin2016ubsrest} for more details.

% \todo{is this needed ?}
% Note that the loss function in \cref{eq:sr-var} is not strictly increasing and fails to satisfy \Cref{as:l-strictly-increasing}. However, \Cref{proposition:ubsr-unique-root} can be shown to hold if the CDF of $X$ is strictly increasing. See remarks \ref{remark:weak-probability-condition} and \ref{remark:strictly-increasing-cdf}.
\item  Entropic risk. Let $\lambda =1$ and $\beta > 0$, and define the loss function as $l(x) = e^{\beta x},\forall x \in \Rel$. Then, $\ubsr$ coincides with the entropic risk measure ($\rho_{e}$) \cite[Example 4.114]{follmer2016stochastic}), defined as $\rho_\textrm{e}(X) \triangleq \beta^{-1}\left[\log\left(\Exp[e^{-\beta x}]\right)\right]$. See \cite{giesecke-risk-large-losses,dunkel_efficient_2007,dunkel2010stochastic} for more details.

\item  Expectile risk. Given $a \geq b\geq 0$ and $c \in \Rel$, define the loss function as $l(x) = c + ax^+ - bx^-, \forall x \in \Rel$. The piecewise linear function above is a simple yet popular choice of loss function that scales losses and gains differently. $\sr{l}{\lambda}{\cdot}$ is coherent if and only if the loss function is of the form above \cite{giesecke-risk-large-losses}. For the special case of $a \in \left[1/2,1\right), b = 1 - a, c = 0$ and $\lambda = 0$, $-\ubsr$ coincides with the expectile risk \cite{delbaen_risk_2016}, which is defined as the solution to
\begin{equation*}
    \argmin_{t \in \Rel} \left\{a \Exp\left[\left([X-t]^+\right)^2\right] + (1-a)\Exp\left[\left([X-t]^-\right)^2 \right]\right\}.
\end{equation*}
Expectiles were introduced by \cite{expectile_newey_asymmetric_1987} as the minimizer to an asymmetric least-square criterion for solving a regression problem. An expectile can be interpreted in multiple ways \cite{philipps_interpreting}, however, we focus on the negative expectile, which is a coherent risk measure. In fact, it is the only coherent risk measure that is elicitable \cite{Bellini_2015,ziegel_coherence_2016}. See \cite{follmer2016stochastic,philipps_interpreting,daouia_expectile_2024} for detailed discussions on the expectile risk.\\[0.5ex]

\item  Polynomial Risk. Given $a > 1$, define the loss function as $l(x) = a^{-1}[x^+]^a,$ $\forall x \in \Rel$. Polynomial loss functions have been previously analyzed by \cite{giesecke-risk-large-losses} and \cite{dunkel_efficient_2007} for UBSR estimation with a bounded random variable.
\end{enumerate}
% \paragraph{5. Expectile risk:}Given $\alpha \in \left[1/2,1\right)$, define the loss function as
% \begin{equation*}
%     l(x) = \alpha [x]^+ - (1 - \alpha)[x]^- , \forall x \in \Rel.
% \end{equation*}
% Then, $-\ubsr$ coincides with the expectile risk $\rho_x$,  defined below.
% \begin{equation*}
%     \rho_{x}(X) \triangleq \left\{ t \in \Rel \left| \alpha \Exp\left[\left([X-t]^+\right)^2\right] + (1-\alpha)\Exp\left[\left([X-t]^-\right)^2 \right]\right.\right\}.
% \end{equation*}
% Expectile risk is the only coherent risk measure that is elicitable. See \todo{add references} for further details on the expectile risk.

Observe that the loss function underlying VaR satisfies the assumption A$2$ of \Cref{proposition:ubsr-as-a-root} and the assumption A$2$' of \Cref{proposition:ubsr-unique-root}, and therefore both propositions apply to VaR. Furthermore, by \Cref{proposition:sr-convex_risk_measure} we conclude that VaR is a monetary risk measure. Similarly, the loss functions underlying entropic risk, expectile, and polynomial risk measures  satisfy assumption A$1$ of \Cref{proposition:ubsr-as-a-root} and the assumption A$1$' of \Cref{proposition:ubsr-unique-root}, and therefore both propositions apply to each of these instances of UBSR. In addition, the last three risk measures  mentioned above have convex loss functions, and by \Cref{proposition:sr-convex_risk_measure} we conclude that the corresponding UBSR measures are convex. 
%\cite{dunkel2010stochastic} consider this loss function under two assumptions: Gaussian and heavy-tailed distributions. \cite{giesecke-risk-large-losses} consider a mixture of student and Gaussian distributions to test sensitivity of the risk measure, and observed that $a=1.5$ to be a good choice.

% \subsection{OCE for Unbounded Random Variables}
% \label{section:oce-unbounded}
% \input{inputs/oce-unbounded}

\section{Estimation of UBSR.}
\label{sec:estimation}
In this section, we discuss techniques to estimate the UBSR of a given random variable $X$. In practice, the true distribution of $X$ is unavailable, and one relies on the samples of $X$ to estimate the UBSR. We use the sample average approximation (SAA) technique \cite{saa-intro-shapiro-alexander,robust-stoc-appx-nemirovski-shapiro,shapiro_book} for UBSR estimation. Such a scheme for UBSR estimation was proposed and analyzed by \cite{zhaolin2016ubsrest}.

Consider the following optimization problem:
\begin{align}\label{eq:saa_stochastic}
    \text{minimize} \,\,\, t, &\quad
    \text{subject to} \,\,\,\, \Exp[l(-X-t)] \leq \lambda.
\end{align}
It is easy to see from Definition \ref{eq:definition-UBSR} that $\ubsr$ is the solution to the above problem. In the SAA scheme, we solve an alternate optimization problem obtained by replacing the expectation in \cref{eq:saa_stochastic} with an $m$-sample estimate. Using i.i.d. samples $\{{Z}_i\}_{i=1}^m$ from $X$, collectively denoted by a random vector $\Z$, we frame the following optimization problem, whose solution is an estimator of $\ubsr$:
\begin{align}\label{eq:saa_deterministic}
    \text{minimize} \,\,\, t, &\quad
    \text{subject to} \,\,\,\,\, \dfrac{1}{m} \sum_{i=1}^{m}l(-{Z}_i-t) \leq \lambda.
\end{align}
\subsection{Error bounds for UBSR estimation.}
Each realization of $\Z$ yields a possibly different solution to \cref{eq:saa_deterministic}. Therefore, we first quantify the estimation error for a single realization. In fact, we make the analysis even more general and view the SAA estimator given by \cref{eq:saa_deterministic} as a function of $m$ real numbers, denoted as $\mathbf{z} \in \Rel^m$. Thus, we do not restrict $\mathbf{z}$ to be a realization of $\Z$ and instead, $\mathbf{z}$ is any vector in $\Rel^m$. Formally, let $m\geq 1$, and let the function $\srt^m:\Rel^m \to \Rel$ be defined as 
\begin{equation}\label{eq:sr-m-definition}
    \srm{\mathbf{z}} \triangleq \min \left\{ t \in \Rel \left| \dfrac{1}{m} \sum_{j=1}^m l(-\textrm{z}_j-t) \leq \lambda \right. \right\}.
\end{equation}
Since $\lambda$ lies in the interior of the range of $l$, it is easy to verify that $\srm{\mathbf{z}}$ satisfies the constraint in (\ref{eq:sr-m-definition}) with equality for any continuous loss function $l$. This implies that for all $\mathbf{z} \in \Rel^d$, 
\begin{equation}\label{eq:saa-ubsr-solution-characterization}
    \dfrac{1}{m} \sum_{j=1}^m l(-\textrm{z}_j-\srm{\mathbf{z}}) = \lambda.
\end{equation}
The quality of solution given by the SAA scheme in \cref{eq:saa_deterministic} is determined by the estimation error averaged across all the realizations of $\Z$, i.e., by the mean absolute error $\Exp\left[|\srm{\mathbf{Z}} - \ubsr |\right]$ or the mean squared-error $\Exp\left[|\srm{\mathbf{Z}} - \ubsr |^2\right]$. We now introduce an assumption on the loss function that is used to derive non-asymptotic error bounds on the estimator $\srm{\mathbf{Z}}$. \begin{assumption}\label{assumption:l-b1-increasing}
    The loss function $l$ is strictly increasing, and there exists $b>0$ such that $l(y)-l(x)>b(y-x)$ for every $y>x, \text{ and } x,y \in \Rel$. 
\end{assumption}
Assumptions similar to \Cref{assumption:l-b1-increasing} are commonly made in the estimation of UBSR, cf. \cite{dunkel2010stochastic,zhaolin2016ubsrest}.
Recall that $\Z$ is a random vector of i.i.d. samples of $X$. Next, we derive MAE and MSE bounds on the UBSR estimator $\srm{\mathbf{Z}}$. 
\begin{lemmma}\label{lemma:ubsr-saa-bounds}
     Suppose the assumptions of \Cref{proposition:ubsr-unique-root} hold and suppose $l$ satisfies \Cref{assumption:l-b1-increasing}. Let $X \in \xl$ be such that there exists $\sigma>0$ satisfying $\mathrm{Var}\left(l\left(-X-\ubsr\right)\right) \leq \sigma^2$. Then,
    \begin{align*}%\label{eq:ubsr-saa-bounds}
        \Exp\left[|\srm{\mathbf{Z}} - \ubsr |\right] &\leq \dfrac{\sigma}{b \sqrt{m}}, \;\; \textrm{ and } 
        \Exp\left[|\srm{\mathbf{Z}}- \ubsr |^2\right] \leq \dfrac{\sigma^2}{b^2 m},
    \end{align*}
    where $b$ is as given in \Cref{assumption:l-b1-increasing}.
\end{lemmma}
\begin{myproof}
    Let $\mathbf{z} \in \Rel^m$ and define $g_\mathbf{z}^m(t) \triangleq \dfrac{1}{m}\sum_{i=1}^m l(-\textrm{z}_i-t) - \lambda $. By \Cref{assumption:l-b1-increasing}, $ g_\mathbf{z}^m(t_1) - g_\mathbf{z}^m(t_2) \geq b(t_2-t_1), \forall t_1<t_2$. Recall the definition of $\srm{\cdot}$ in (\ref{eq:sr-m-definition}). If $\ubsr \geq \srm{\mathbf{z}}$, then by substituting $t_1=\srm{\mathbf{z}}$ and $t_2=\ubsr$, we have
    \begin{align}\label{eq:ubsr-saa-sr-inequality}
        \left| \ubsr - \srm{\mathbf{z}} \right| \leq \dfrac{\left| g_\mathbf{z}^m(\ubsr) - g_\mathbf{z}^m(\srm{\mathbf{z}})\right|}{b}.
    \end{align}
    If $\ubsr < \srm{\mathbf{z}}$ holds, then $t_1=\ubsr$ and $t_2=\srm{\mathbf{z}}$ results in same bound as in \cref{eq:ubsr-saa-sr-inequality}. 
     By (\ref{eq:saa-ubsr-solution-characterization}) and \Cref{proposition:ubsr-unique-root}, we have $g_\mathbf{z}^m(\srm{\mathbf{z}}) = 0 = g_X(\ubsr)$. Then by substituting $g_\mathbf{z}^m(\srm{\mathbf{z}})$ in \cref{eq:ubsr-saa-sr-inequality} with $g_X(\ubsr)$, we have
    \begin{align*}
        &\left| \ubsr - \srm{\mathbf{z}} \right| \leq \dfrac{\left| g_\mathbf{z}^m(\ubsr) - g_X(\ubsr)\right|}{b} \\
        &=\dfrac{\left| \dfrac{1}{m}\sum_{i=1}^ml(-\textrm{z}_i-\ubsr) - \Exp\left[l(-X-\ubsr)\right]\right|}{b}
    \end{align*}
    The above inequality holds for any $\mathbf{z} \in \Rel^m$. Consequently, for any $m$-dimensional random vector $\mathbf{Z}$, the following holds w.p. $1$:
    \begin{align*}
        &\left| \ubsr - \srm{\mathbf{Z}} \right| \leq \dfrac{\left| \dfrac{1}{m}\sum_{i=1}^ml(-Z_i-\ubsr) - \Exp\left[l(-X-\ubsr)\right]\right|}{b}.
    \end{align*}
    Taking expectation on both sides, and replacing $\Exp\left[l(-X-\ubsr)\right]$ with \\$\dfrac{1}{m}\Exp\left[\sum_{i=1}^m l(-Z_i-\ubsr)\right]$, we have 
    \begin{align*}
        &\Exp\left[\left| \ubsr - \srm{\mathbf{Z}}\right] \right| \\
        &\leq \dfrac{\Exp\left[\left|\sum_{i=1}^m l(-Z_i-\ubsr)\right.\right.}{b m} - \dfrac{\left.\left.\Exp\left[\sum_{i=1}^m l(-Z_i-\ubsr)\right]\right|\right]}{b m} \\
        &\leq \dfrac{\sqrt{\mathrm{Var}\left(\sum_{i=1}^m l(-Z_i-\ubsr)\right)}}{b m} 
        %&= \dfrac{\sqrt{\sum_{i=1}^m \left(\mathrm{Var}\left(l(-Z_i-\ubsr)\right)\right)}}{b_1 m} \\
        = \dfrac{\sqrt{\mathrm{Var}\left(l(-X-\ubsr)\right)}}{b \sqrt{m}} \leq \dfrac{\sigma}{b \sqrt{m}}.
    \end{align*}
    The second inequality is the Holder's inequality. The first equality follows because $Z_i$'s are independent, while the second equality follows because $Z_i$'s are identical. The last inequality follows from the variance assumption of the lemma and we conclude the proof of the first claim of the lemma. The second claim follows by first squaring both sides before taking the expectation in the above proof, and then applying completely parallel arguments as above. 
\end{myproof}
% \begin{myproof}
%     See \Cref{proof:lemma-ubsr-saa-bounds}
% \end{myproof}
%Compared to the conference version \cite{gupte2023optimization} of this manuscript where MAE and MSE bounds of the order \order{1/m^{1/4}} and \order{1/\sqrt{m}} were obtained, not only the bounds in \Cref{lemma:ubsr-saa-bounds} are significantly tighter, but also apply to a wider class of loss functions as we no longer require the loss function $l$ to be smooth. In addition, the result in the conference version was obtained using a Wasserstein distance bound, which in turn required certain higher moment bounds to hold, whereas \Cref{lemma:ubsr-saa-bounds} guarantees MAE and MSE bounds of the order \order{1/\sqrt{m}} and \order{1/m}, respectively using only a variance bound. 
In comparison to \cite{JMLR-LAP-SPB}, where an MAE bound of order \order{1/\sqrt{m}} was obtained under the assumption of convexity and Lipschitz continuity on $l$, \Cref{lemma:ubsr-saa-bounds} covers possibly non-Lipschitz, non-smooth and non-convex loss functions, with an MAE bound of the same order and an MSE bound of order \order{1/m}. 

\subsection{Algorithm for UBSR estimation.}
To obtain $\srm{\mathbf{Z}}$ one needs to solve an optimization problem, as defined in \cref{eq:sr-m-definition}. However, a closed form expression of the solution to the problem is not available. Instead, it is possible to obtain a solution within $\delta$-neighborhood of $\srm{\mathbf{Z}}$. We describe below a variant of the bisection method that gives ${t}_m(\Z)$, a $\delta$-approximate solution to the optimization problem in \cref{eq:sr-m-definition}. 
%Algorithm \ref{alg:saa_bisect} presents the pseudocode to get  that satisfies Proposition \ref{proposition:sr-estimator-bounds}. 

\begin{figure}[!t]
\centering
\begin{algorithm}[H]
\caption{UBSR-SB (Search and Bisect)}\label{alg:saa_bisect}
\begin{algorithmic}[1] 
    \REQUIRE threshold $\delta > 0, \text{ i.i.d. samples } \{Z_i\}_{i=1}^m$
    \item[]\hspace*{-\leftmargin}\textbf{Define:} $\hat{g}(t) = \frac{1}{m}\sum_{j=1}^m l(-Z_i-t) - \lambda$.
    % \IF{$\hat{g}(0)>0$}
    %         \STATE $low, high \leftarrow -1, 0 $
    %     \ELSE
    %         \STATE $low, high \leftarrow 0, 1$
    % \ENDIF
    \STATE $low, high \leftarrow 0, 1$
    \STATE \textbf{if} $\hat{g}(0)>0$ \textbf{then} $low, high \leftarrow -1, 0 $
    % \WHILE {$\hat{g}(high) > 0$} \STATE {$high \leftarrow 2 * high$} \ENDWHILE
    % \WHILE{$\hat{g}(low) < 0$} \STATE $low \leftarrow 2 * low$ \ENDWHILE
    \STATE \textbf{while} $\hat{g}(high) > 0$ \textbf{do} $high \leftarrow 2 * high$
    \STATE \textbf{while} $\hat{g}(low) < 0$ \textbf{do} $low \leftarrow 2 * low$
    \STATE $T \leftarrow high - low$ 
    \WHILE{$T > 2\delta$}
        % \IF{$\hat{g}(t_m )>0$}
        %     \STATE $low \leftarrow t_m$
        % \ELSE
        %     \STATE $high \leftarrow t_m$
        % \ENDIF
        \STATE \textbf{if} $\hat{g}(t_m )>0$ \textbf{then} $low \leftarrow t_m$ \textbf{else} $high \leftarrow t_m$ 
        \STATE $T \leftarrow high$ - $low, t_m \leftarrow (low+high)/2$
    \ENDWHILE
    
    \RETURN $t_m$
\end{algorithmic}
\end{algorithm}
\end{figure}

The reference \cite{zhaolin2016ubsrest} assumes knowledge of $t_X^{\mathrm{l}},t_X^{\mathrm{u}}$ (defined in \Cref{proposition:ubsr-as-a-root}) for a bisection method to solve (\ref{eq:sr-m-definition}).  However, these values are seldom known in practice. Our algorithm does not require the knowledge of $t_X^{\mathrm{l}},t_X^{\mathrm{u}}$ beforehand. Instead, the algorithm works by first finding the search interval $\left[t_X^{\mathrm{l}},t_X^{\mathrm{u}}\right]$, and then performing a bisection search. The variables $low,high$ of the algorithm are proxies for $t^{\mathrm{l}}_{X},t^{\mathrm{u}}_{X}$ and the first two loops of the algorithm find $low$ and $high$ respectively such that $\srm{\mathbf{Z}} \in [low, high]$. The final loop in the algorithm performs bisection search to return a value in the $\delta$-neighborhood of $\srm{\mathbf{Z}}$. The total iteration complexity of the algorithm is at most $ 2\left(1+ \log\left( \dfrac{\max(|t^{\mathrm{u}}_{X}|, |t^{\mathrm{l}}_{X}|)}{\delta} \right)\right)$. 
The error bounds from \Cref{lemma:ubsr-saa-bounds} easily extend to the solution given by \Cref{alg:saa_bisect}. The following proposition extends the bounds from \Cref{lemma:ubsr-saa-bounds} to the solution given by Algorithm \ref{alg:saa_bisect}.

\begin{proposition}\label{proposition:sr-estimator-bounds} 
Suppose the UBSR risk parameters $l$ and $\lambda$, and $X \in \xl$ are chosen such that the assumptions of \Cref{lemma:ubsr-saa-bounds} hold. Let $\mathbf{Z} = \{{Z}_i\}_{i=1}^m$ be $m$ i.i.d. samples of $X$. Let $t_m(\Z)$ be an approximate solution to (\ref{eq:saa_deterministic}) given by the Algorithm \ref{alg:saa_bisect} for the inputs, $\mathbf{Z}$ and $\delta = \dfrac{d_1}{\sqrt{m}}$, for some $d_1>0$. Then,
    \begin{equation*}
        \Exp[|t_m(\Z) - \ubsr|] \leq \frac{d_1+\frac{\sigma}{b}}{\sqrt{m}}, \,\, \textrm{ and } \;\; 
        \Exp[\left(t_m(\Z) - \ubsr \right)^2] \leq \frac{2\left(d_1^2+\frac{\sigma^2}{b^2}\right)}{m},
    \end{equation*}
    where $b$ and $\sigma_1$ are as given in \Cref{lemma:ubsr-saa-bounds}.
\end{proposition}
\begin{myproof}
    Using the triangle inequality, we have
    \begin{align*}
        \Exp[|{t}_m(\Z) - \ubsr|] &\leq \Exp[|{t}_m(\Z) - \srm{\mathbf{Z}}| + |\srm{\mathbf{Z}} - \ubsr|] \\
        &\leq \delta + \Exp[|\srm{\mathbf{Z}} - \ubsr|] \leq \dfrac{d_1+\dfrac{\sigma}{b}}{\sqrt{m}}.
    \end{align*}
    where the last inequality follows from \Cref{lemma:ubsr-saa-bounds}.
    Similarly, using the identity: $(a+b)^2\le 2a^2+2b^2$, and invoking \Cref{lemma:ubsr-saa-bounds} we obtain
    \begin{align*}
    \Exp[\left({t}_m(\Z) - \ubsr\right)^2] &\leq 2\Exp\left[\left({t}_m(\Z) - \srm{\mathbf{Z}}\right)^2\right] 
        + \Exp\left[\left(\srm{\mathbf{Z}} - \ubsr\right)^2\right]\\
        &\leq \dfrac{2 \left(d_1^2 + \dfrac{\sigma^2}{b^2}\right)}{m}.
    \end{align*}
\end{myproof}

We now analyze the iteration complexity of Algorithm \ref{alg:saa_bisect}. Suppose the first and second loops run for $n_1,n_2$ iterations. It is trivial to see that $n_1 < 1 + \log_2(|t^{\mathrm{u}}_{X}|)$ and $n_2 < 1 + \log_2(|t^{\mathrm{l}}_{X}|)$. Due to the carefully chosen initial values of variables $low$ and $high$, at least one among $n_1$ or $n_2$ will always be $0$. Then $T \leq 2^n$ holds, where $n \triangleq \max(n_1, n_2)$. Suppose the final loop terminates after $k$ iterations. Then at $k-1$, we have $\dfrac{T}{2^{k-1}} > 2\delta$ which implies that $k < 1 + \log_2\left(\dfrac{\max(|t^{\mathrm{u}}_{X}|,|t^{\mathrm{l}}_{X}|)}{\delta}\right)$.
Thus the total iteration complexity of the algorithm is at most $\max(n_1,n_2) + k $ which is upper-bounded by $ 2\left(1+ \log\left( \dfrac{\max(|t^{\mathrm{u}}_{X}|, |t^{\mathrm{l}}_{X}|)}{\delta} \right)\right)$.

% A detailed discussion on UBSR estimation using samples, including an estimation algorithm and its analysis is provided in \Cref{supp:ubsr-estimation-analysis} of the supplementary material \cite{prashla_web_resource}.

% \section{Optimization of Convex Risk Measures}
% \label{sec:optimization}
% \input{inputs/biased-stoc-optimization}

\section{UBSR Optimization}
\label{sec:ubsr-opt}
%\subsection{A note on stochastic optimization with biased gradients.}
Let $\BB$ be an open subset of $\Rel^d$. Consider an objective function $F: \BB \times \Rel^k \to \Rel$ and a $k$-dimensional random vector $\xi$. A standard stochastic optimization problem involves maximizing (or minimizing) the function $f:\Theta \subseteq \BB \to \Rel$, where $f(\theta) \triangleq \Exp[F(\theta, \xi)]$, and $\Theta$ is a compact and convex set. The difficulty in solving this problem arises because, in reality, the distribution of $\xi$ is unknown, and one may only have access to samples from $F(\theta, \xi)$, $\nabla_\theta F(\theta, \xi)$, or both. A popular approach to solve this problem is to form a sample-based gradient estimator of $\nabla f$ and employ it in a stochastic gradient (SG) based algorithm. The convergence rates for such SG algorithms are well known in the literature \cite{shapiro_book,moulines2011non,robust-stoc-appx-nemirovski-shapiro}.

A major drawback in the above problem formulation is that the uncertainty or risk of the random variable $F(\theta,\xi)$ is ignored. Risk-sensitive optimization, on the other hand, accounts for uncertainty arising from the underlying random variable by incorporating a risk measure $\rho$ into the objective. The optimization problem is given by
\begin{equation*}
    \centerline{$\theta_* \in \argmin_{\theta \in \Theta} 
\rho(F(\theta, \xi)).$}
\end{equation*}
We now state a result that applies to all convex risk measures. The following lemma shows that the above risk-sensitive objective inherits the convexity of the objective function $F$ when the risk measure $\rho$ is convex. 

\begin{lemmma}\label{lemma:h-convex-general}
    Let $F(\cdot,\xi)$ be continuously differentiable and $\mu$-strongly concave w.p. $1$ for some $\mu\geq 0$, and let $\rho$ be convex. Then $h:\Theta \to \Rel$ given by $h(\theta) = \rho(F(\theta,\xi))$ is $\mu$-strongly convex.
\end{lemmma}
% \begin{proof}
%     See \Cref{proof:lemma_h_cvx_gen} for the proof.
% \end{proof}
\begin{myproof}Since $F$ is continuously differentiable and $\mu$-strongly concave w.p. 1. Then $-F$ is $\mu$-strongly convex and by Theorem 2.1.8 of \cite{nesterov_introductory_2004}, we have the following for every $\theta_1,\theta_2 \in \Theta$ and every $\alpha \in [0,1]$ . 
\begin{equation}\label{eq:f-sc}
    F(\alpha \theta_1 +(1-\alpha)\theta_2,\xi) \geq \alpha F(\theta_1,\xi) + (1-\alpha) F(\theta_2,\xi) + \frac{\alpha(1-\alpha)\mu}{2}\norm{\theta_1-\theta_2}^2, \qquad \text{w.p. 1}.
\end{equation} 
Then we have,
\begin{align*}
    &h(\alpha \theta_1 + (1-\alpha) \theta_2) \\
    &= \rho\left(F\left(\alpha \theta_1 + (1-\alpha)\theta_2, \xi\right)\right) \leq \rho\left(\alpha F\left(\theta_1,\xi\right) + (1-\alpha) F\left(\theta_2,\xi\right) + \frac{\alpha(1-\alpha)\mu}{2}\norm{\theta_1-\theta_2}^2\right) \\
    &= \rho\left(\alpha F\left(\theta_1,\xi\right) + (1-\alpha) F\left(\theta_2,\xi\right)\right) - \frac{\alpha(1-\alpha)\mu}{2}\norm{\theta_1-\theta_2}^2 \\
    &\leq \alpha \rho\left(F\left(\theta_1,\xi\right)\right) + (1-\alpha) \rho\left(F\left(\theta_2,\xi\right)\right) - \frac{\alpha(1-\alpha)\mu}{2}\norm{\theta_1-\theta_2}^2 \\
    &= \alpha h(\theta_1) + (1-\alpha) h(\theta_2) - \frac{\alpha(1-\alpha)\mu}{2}\norm{\theta_1-\theta_2}^2.
\end{align*}
Here, the first inequality follows from \cref{eq:f-sc} and the monotonicity of $\rho$ given by \Cref{proposition:sr-convex_risk_measure}, while the second equality and the second inequality follow from cash-invariance and convexity of $\rho$ respectively given by \Cref{proposition:sr-convex_risk_measure}. Then by Theorem 2.1.8 of \cite{nesterov_introductory_2004}, we conclude that $h$ is $\mu$-strongly convex. 
\end{myproof}

In this paper, we choose UBSR as the risk criterion. Precisely, we are interested in solving the problem of minimizing the UBSR of $F(\theta,\xi)$ over the decision space $\Theta$, i.e., to find \\[1ex]
\begin{equation}
\centerline{$\theta_* \in \argmin_{\theta \in \Theta} 
SR_{l,\lambda}(F(\theta, \xi)).$}
\end{equation}
Previous works \cite{guo_distributionally_2019,delage_shortfall_2022,zhang_preference_2022}, with the exception of \cite{Hegde2024}, consider UBSR optimization in an offline setting where a fixed set of $M$ samples of $\xi$ are available. In this offline setting, $F$ is assumed to be a known function such that both $F(\theta,\xi)$ and $\nabla F(\theta,\xi)$ can be evaluated for every $\theta$, using the same $M$ samples of $\xi$. A popular example of this setting is portfolio optimization, where $F(\theta,\xi) = \theta^T \xi$. The convexity and linearity of $F$ are then used to solve an SAA variant of UBSR optimization where $\xi$ is replaced with the $M$-sample empirical r.v. Following \cite{Hegde2024}, we consider an online setting where we assume that given a parameter choice $\theta$, samples of $F(\theta, \xi)$ and $\nabla_\theta F(\theta, \xi)$ can be obtained. We do not impose assumptions of convexity or linearity on $F$. Our online setting addresses risk-sensitive optimization problems in the domains of online learning and multi-armed bandits. Moreover, our setting also captures the learning constraints of horizontal federated learning \cite{federated_learning}, where the dataset ($\xi$) may not be publicly available, but samples of $F(\theta,\xi)$ and $\nabla_\theta F(\theta, \xi)$ may be obtained by querying the individual nodes by passing the parameter $\theta$. Lastly, our setting subsumes the offline setting considered in previous works, and our gradient-based approach can be easily specialized to it; see \Cref{section:offline} for further details. 
%We next derive an expression for the UBSR gradient in Section \ref{sec:ubsr-grad-expr}. Subsequently, Sections \ref{ss:ubsr_grad_estimator} and \ref{ss:sg-online} handle the online UBSR optimization problem, while Section \ref{section:offline} covers the extension to the offline setting.
\subsection{UBSR gradient expression}
\label{sec:ubsr-grad-expr}
From \Cref{eq:definition-UBSR}, the optimization objective $\srt:\BB \to \Rel$ is given as\footnote{For notational convenience, we suppress the dependency of $l,\lambda$ on $\srt$.}: 
\begin{equation}\label{eq:sr-theta}
     \srth{} = \inf \left\{ t \in \Rel \left| \Exp[l(-F(\theta,\xi)-t)] \leq \lambda \right.\right\} .
\end{equation}    
Next, we define $g:\Rel \times \BB \to \Rel$ as follows:
\begin{equation}\label{eq:g_t_theta}
    g(t,\theta) \triangleq \Exp[ l(-F(\theta,\xi)-t)] - \lambda.
\end{equation}
For the sake of readability, we shall restate \Cref{proposition:ubsr-unique-root} in terms of $\srt$ and $\theta$. We first restate an assumption made in \Cref{proposition:ubsr-unique-root}.
\begin{assumption}\label{assumption:sr-base-assumption}
    For every $\theta \in \BB$, $ F(\theta, \xi) \in \xl$ and there exist $t_\textrm{u}(\theta), t_\textrm{l}(\theta) \in \Rel$ such that $g(t_\textrm{u}(\theta), \theta) \leq 0 < g(t_\textrm{l}(\theta),\theta)$.
\end{assumption}

\begin{proposition}\label{prop:proposition-g-root}
Suppose that $l$ is continuous and increasing. Suppose that either $l$ is strictly increasing or for every $\theta \in \BB$, $F(\theta,\xi)$ has a strictly increasing CDF. Suppose \Cref{assumption:sr-base-assumption} holds. Then, for every $\theta \in \BB$, $g(\cdot, \theta)$ is continuous, strictly decreasing, and has a unique root that coincides with $\srth{}$, i.e., $g(\srth{}, \theta) = 0$ for every $\theta \in \BB$.
\end{proposition}

We shall use the result above to invoke the implicit function theorem \cite{real-analysis-rudin} to derive an expression for the gradient of UBSR. %, and use this expression to arrive at a biased gradient estimator for $\nabla\srt$. %Next, we incorporate this gradient estimator into an SG scheme for solving (\ref{eq:hthetaopt}). We provide non-asymptotic bounds for this SG algorithm for the case when $\srt(\cdot)$ is either strongly convex or convex.
%A standard assumption of smoothness of the random function $F$ is made below.
For this derivation, we require the following assumptions on the loss function $l$ and the objective function $F(\cdot,\xi)$.
% \begin{assumption} \label{as:l_cont_diff} $l$ is continuously differentiable. 
% \end{assumption}
% \begin{assumption}\label{as:loss-fn-polynomial-bound}
%      The loss function $l$ is continuously differentiable and $l'$ is $S_1$-Lipschitz. \todo{Need $S_1$ here ?}
% \end{assumption}
% \begin{assumption}\label{as:F_cont_diff}
%     $F(\cdot, \xi)$ is continuously differentiable almost surely. 
% \end{assumption}
\begin{assumption}\label{assumption:loss-fn-smooth}
    The loss function $l$ is continuously differentiable, and there exists $S>0$ such that $l$ is $S$-smooth.
\end{assumption}
\begin{assumption}\label{as:lp_moments_F_dF}
    The objective function $F(\cdot,\xi)$ is continuously differentiable a.s. and there exists $L>0$ such that, for every $\theta_1,\theta_2 \in \Theta$, we have
    \begin{align*}
        \norm{\nabla F(\theta_1, \xi)- \nabla F(\theta_2,\xi)}_{2} \leq L \norm{\theta_1-\theta_2}_2 \;\; \text{a.s.} 
    \end{align*}
\end{assumption}
\begin{assumption}\label{assumption:F-bounded-higher-moment}
    There exists $\sigma_0 > 0$ and $M>0$ satisfying Var$\left(F(\theta,\xi)\right) \leq \sigma_0^2$ and $\norm{\nabla F(\theta, \xi)}_{\Leb{2}} \leq M$, for all $\theta \in \Theta$.
\end{assumption}
\begin{assumption}\label{as:variance-bound-l-F} 
There exist $T>0$ such that
\begin{align*}
\norm{\mathbf{Y}_\theta - \Exp\left[\mathbf{Y}_\theta\right]}_{\Leb{2}} \leq T, \;\; \text{for every } \theta \in \Theta,
\end{align*}
where $\mathbf{Y}_\theta \triangleq l'(-F(\theta,\xi)-\srth{})\nabla F(\theta,\xi), \forall \theta \in \Theta$.
\end{assumption}

In the literature on UBSR optimization \cite{zhaolin2016ubsrest,Hegde2024}, the authors assume that the underlying random variables have bounded support. In contrast, our work extends the analysis to possibly unbounded random variables. In that spirit, we make \cref{assumption:F-bounded-higher-moment,as:lp_moments_F_dF,as:variance-bound-l-F} and use them to derive MAE and MSE bounds on our proposed UBSR gradient estimator. Assumptions similar to \cref{as:lp_moments_F_dF,as:variance-bound-l-F} have been made before in \cite{zhaolin2016ubsrest} and \cite{Hegde2024} for the asymptotic and non-asymptotic analysis of the UBSR optimization scheme, respectively. Using these assumptions, we establish that the objective function $\srt(\cdot)$ is smooth in the next section. Akin to \Cref{as:variance-bound-l-F}, an assumption that bounds the variance of the stochastic gradient is commonly made in the non-asymptotic analysis of stochastic gradient algorithms \cite{moulines2011non,bottou2018optimization,bhavsar2021nonasymptotic}. %The assumption above is made in a similar spirit, and the result below establishes that the mean-squared error of the UBSR gradient estimate \cref{eq:J-definition}  vanishes asymptotically at a $O(1/m)$ rate. 

We now present a lemma that we use to derive the gradient expression of UBSR. This lemma uses the dominated convergence theorem to derive expressions for the partial derivatives of the function $g$ defined in \cref{eq:g_t_theta}.
\begin{lemmma}\label{lemma:g-dct-lemma}
    Suppose the loss function $l$ is continuously differentiable. Suppose $F(\theta, \xi) \in \xl, \forall \theta \in \BB$ and $F(\cdot, \xi)$ is continuously differentiable almost surely. Then $g$ is continuously differentiable on $\Rel \times \BB$, and for all $(t,\theta) \in \Rel \times \BB$, the partial derivatives are given by
    % \begin{align}\label{part_der}
    % \frac{\partial g(t,\theta)}{\partial \theta_i} &= - \Exp\left[ l'(-F(\theta,\xi)-t) \frac{\partial F(\theta, \xi)}{\partial \theta_i} \right],  \\ \label{part_der2}
    % \frac{\partial g(t,\theta)}{\partial t} &= - \Exp\big[ l'(-F(\theta,\xi)-t) \big],
    % \end{align}
    \begin{align}\label{part_der}
    \frac{\partial g(t,\theta)}{\partial \theta_i} &= - \Exp\left[ l'(-F(\theta,\xi)-t) \frac{\partial F(\theta, \xi)}{\partial \theta_i} \right], \frac{\partial g(t,\theta)}{\partial t} &= - \Exp\big[ l'(-F(\theta,\xi)-t) \big],
    \end{align}
    where $i \in \{1,2,\ldots,d\}$.
\end{lemmma}
\begin{myproof}
    Let $\mu_\xi$ be the probability measure on $\Rel$ induced by the random variable $\xi$. Then for every $t \in \Rel$ and $\theta \in \BB, g$ can be written as 
    \begin{equation*}
        g(t,\theta) = \int_{z} \left[l\left(-F\left(\theta,z\right) - t\right) - \lambda\right] \mu_\xi(dz).
    \end{equation*}
    Let $t_0 \in \Rel$ and $\theta_0 \in \BB$ be chosen arbitrarily.
    
    We now obtain an expression for the partial derivative of $g(t,\theta)$ w.r.t. $t$.
    Suppressing the dependency on $\lambda$ and $\theta_0$, define $f:\Rel \times \Rel \to \Rel$ by $f(t,z) = l(-F(\theta_0,z) - t) - \lambda$. Then, it is clear that  $g(t,\theta_0) =  \int_{z} f(t,z)\mu_\xi(dz)$. Let $\delta > 0$, and suppose that $t \in (t_0 - \delta, t_0+\delta)$, then we claim that\\
    (i) $\int_{z} \left| f(t,z)\right|\mu_\xi(dz) < \infty$;
        (ii) For a fixed $z, \frac{\partial f}{\partial t}(t,z)$ exists and $\frac{\partial f}{\partial t}(\cdot , z)$ is continuous; and\\
        (iii) $\int_z {\sup_{t \in \left[t_0-\delta, t_0+\delta \right]} \left| \frac{ \partial f_{\theta_0}}{\partial t} (t,z)\right|  \mu_\xi(dz)} < \infty.$
    
    The first claim holds because $F(\theta_0,\xi) \in \xl$, which follows from the assumptions of the lemma. The second claim follows because $l$ and $F$ are continuously differentiable. Note that a continuous function on a compact set is bounded. The set $\left[t_0 - \delta, t_0 + \delta\right]$ is compact, and the partial derivative w.r.t. $t$ is continuous, and therefore bounded. Combining this with the fact that the probability measure $\mu_\xi$ is finite, the third claim follows. Precisely, the following inequality is satisfied. 
    \begin{align*}
            &\int_z \sup_{t \in \left[t_0-\delta, t_0+\delta \right]} \left| \frac{ \partial f_{\theta_0}}{\partial t} (t,z)\right|  \mu_\xi(dz) \\
            &= \int_z \sup_{t \in \left[t_0-\delta, t_0+\delta \right]} \left| \frac{\partial }{\partial t} \left[l(-F(\theta_0,z) - t) - \lambda\right] \right|  \mu_\xi(dz) < \infty. 
    \end{align*}
    Then the claims $1$ to $3$ show that the conditions (i),(ii), and (iii') respectively of Theorem A.5.3 of \cite{durrett_2019} hold, and the theorem allows the interchange of the derivative and the integral to yield 
    \begin{align*}
        \frac{\partial g(t,\theta)}{\partial t}\Bigr|_{\substack{t=t_0}} = -\int_z l'(-F(\theta,z) - t_0) \mu_\xi(dz)  
        = -\Exp\left[l'(-F(\theta,\xi) - t_0)\right].    
    \end{align*}
    Since $t_0$ and $\theta_0$ were chosen arbitrarily, the partial derivative expression holds for every $t \in \Rel$ and every $\theta \in \BB$. Furthermore, since the conditions of Theorem A.5.3 of \cite{durrett_2019} hold and Theorem A.5.1 of \cite{durrett_2019} also applies, it is implied that $g(\cdot,\theta)$ is continuously differentiable for every $\theta \in \BB$.
    
     We now obtain an expression for the partial derivative of $g(t,\theta)$ w.r.t. $\theta$. For each $i \in [1,2,\ldots, d]$, let $\mathcal{N}_i(\delta)$ be a neighborhood of $\theta_0$ defined as 
    \begin{equation*}    
    \mathcal{N}_i(\delta) \triangleq \left\{\theta \in \Rel^d \left| \forall j, |\theta^j - \theta_0^j| \leq \delta \;\textrm{ if } \; j = i \; \textrm{ and } \theta^j = \theta_0^j \; \textrm{o.w. }\; \right.\right\}. 
    \end{equation*}
    Since $\theta_0 \in \BB$, and $\BB$ is open, there exists $\hat{\delta}>0$ such that $\mathcal{N}_i(\hat{\delta}) \subseteq \BB, \forall i \in [1,2,\ldots, d]$. Suppressing the dependency on $F$ and $t_0$, define $f:\BB \times \Rel \to \Rel$ by  $f(\theta,z) = l(-F(\theta,z) - t_0) - \lambda$. Then, it is clear that $g(t_0,\theta) =  \int_{z} f(\theta, z)\mu_\xi(dz)$. Choose $i \in [1,2,\ldots, d]$ arbitrarily. Let $\partial_i$ denote $\frac{\partial}{\partial \theta^i}$. Then, for every $\theta \in \mathcal{N}_i(\hat{\delta})$, \\
    1) $\int_{z} \left| f(\theta, z)\right|\mu_\xi(dz) < \infty$.\\
    2) For every $z, \partial_i f(\cdot,z)$ exists and $\partial_i f(\cdot , z) $ is continuous.\\
    3) $\int_z \sup_{\theta \in \mathcal{N}_i(\hat{\delta})} \left| \partial_i f (\theta,z)\right|  \mu(dz) < \infty. $\\
        The first claim holds because $F(\theta_0,\xi) \in \xl$, which follows from the assumptions of the lemma.  Claim 2 follows from the fact that $l$ and $F$ are continuously differentiable. In Claim 3, we apply the fact that the partial derivatives are locally bounded, which follows from the continuous differentiability of $l$ and $F$ and because the set $\mathcal{N}_i(\hat{\delta})$ is compact. Precisely, we have
        \begin{align*}
           &\int_z \sup_{\theta \in \mathcal{N}_i(\hat{\delta})} \left| \partial_i f(\theta,z)\right|  \mu(dz) = \int_z \sup_{\theta \in \mathcal{N}_i(\hat{\delta})} \left| \partial_i \left[l(-F(\theta, z) - t_0) - \lambda\right]\right|  \mu(dz) \\
           & \leq \int_z \sup_{\theta \in \mathcal{N}_i(\hat{\delta})} \left| \partial_i \left[l'(-F(\theta, z) - t_0)\right]\right| \left| \partial_i F(\theta, z) \right| \mu(dz) < \infty.
        \end{align*}
        
        The claims (1),(2), and (3) show that the conditions of Theorem A.5.3 of \cite{durrett_2019} hold, and the theorem yields 
    \begin{align*}
        \partial_i g(t,\theta) &= \int_z l'(-F(\theta,z) - t)\partial_i F(\theta) \mu_\xi(dz) = -\Exp\left[l'(-F(\theta,\xi) - t)\partial_iF(\theta)\right].    
    \end{align*}
    Since $i,t_0$ and $\theta_0$ are chosen arbitrarily, the proof applies to all partial derivatives of $g$ w.r.t $\theta$. Furthermore, since the conditions of Theorem A.5.3 of \cite{durrett_2019} hold, Theorem A.5.1 of \cite{durrett_2019} also applies, which implies that all partial derivatives of $g$ w.r.t. $\theta$ are continuous. 
    
    Combining the two steps, we conclude that all partial derivatives of $g$ are continuous on $\Rel \times \BB$. Then, we recursively apply Theorem 12.11 of \cite{Apostol} to conclude that $g$ is continuously differentiable on $\Rel \times \BB$.
\end{myproof}
% \begin{myproof}
%     See \Cref{proof:lemma_g_dct} for the proof.
% \end{myproof}
%The assumptions of continuous differentiability on $l$ and $F$ are common for interchanging the derivative and the expectation, to obtain the partial derivatives of $g$, defined in (\ref{eq:g_t_theta}). Similar assumptions have been made in the context of UBSR estimation and optimization in previous works, cf. \cite{zhaolin2016ubsrest,Hegde2024}. 
% The next lemma \todo{add this to Thm}shows that the partial derivatives of $\srt$ can be expressed as a ratio of partial derivatives of $g$, using the implicit function theorem in conjunction with the results derived in \Cref{lemma:g-dct-lemma} and \Cref{prop:proposition-g-root}. 
% \begin{lemmma}\label{lemma:g-ift-lemma}Then, $\srt$ is a continuously differentiable mapping on $\B$ and its partial derivatives can be expressed as:
% \begin{equation}\label{eq:implicit}
%     \frac{\partial \srth{}}{\partial
%     \theta_i}\Big|_{\hat{\theta}} =  \frac{-\partial g(t,\theta) / \partial \theta_i}{\partial g(t,\theta) / \partial t}\Big|_{h(\hat{\theta}), \hat{\theta}} \qquad \forall i \in \{1,2,\ldots,d\} \;\text{ and }\;\forall \hat{\theta} \in \B.
% \end{equation}
% \end{lemmma}
% \begin{myproof}
%     See \Cref{proof:lemma-g-ift} for the proof.
% \end{myproof}
We now present the main result, which provides an expression for the gradient of UBSR with respect to the vector parameter $\theta$. 
\begin{theorem}[Gradient of UBSR]\label{theorem:ubsr-gradient}
Suppose the assumptions of \Cref{prop:proposition-g-root} and \Cref{lemma:g-dct-lemma} hold. Then the function $\srt$ given by \cref{eq:sr-theta} is continuously differentiable and the gradient of $\srt$ can be expressed as follows: 
\begin{equation}\label{h_grad_0}
\nabla \srth{} =  \frac{-\Exp \Big[ l'(-F(\theta,\xi)-\srth{}) \nabla F(\theta,\xi) \Big]}{\Exp \big[ l'(-F(\theta,\xi)-\srth{}) \big]}, \forall \theta \in \BB. 
\end{equation}
\end{theorem}
\begin{myproof}
From \Cref{lemma:g-dct-lemma}, we have that (a) $g$ is continuously differentiable on $\Rel \times \BB$. By \Cref{prop:proposition-g-root}, we have (b) $\frac{\partial g}{\partial t}  = -\Exp\left[l'\left(-F\left(\theta,\xi\right)-t\right) \right] < 0$, and c) $g(\srth{}, \theta) = 0, \forall \theta \in \BB$. We now invoke the implicit function theorem. Using (a), (b), and (c), we invoke Theorem 9.28 of \cite{real-analysis-rudin} to infer that, for every $\hat{\theta} \in \BB$, we have
\begin{equation*}
    \frac{\partial \srth{}}{\partial \theta_i}\Big|_{\hat{\theta}} = -\frac{\partial g(t,\theta) / \partial \theta_i}{\partial g(t,\theta) / \partial t}\Big|_{(h(\hat{\theta}), \hat{\theta})} \quad \forall i \in \{1,2,\ldots,d\}.
\end{equation*} 
By \Cref{lemma:g-dct-lemma}, the partial derivatives above are expressed as
\begin{equation*}%\label{h_partial}
\frac{\partial \srth{}}{\partial \theta_i} = \frac{-\Exp\Big[ l'(-F(\theta,\xi)-\srth{}) \frac{\partial F(\theta,\xi)}{\partial \theta_i} \Big]}{\Exp\big[ l'(-F(\theta,\xi)-\srth{}) \big]}, \forall i.%=1\text{ to } d.%\in \{1,2,\ldots,d\}.
\end{equation*}
and the claim of the theorem follows.
\end{myproof}
A scalar version of the above theorem has been proved in earlier works \cite[Theorem 3]{zhaolin2016ubsrest}. The expression for the derivative there is a special case of \cref{h_grad_0} with $d=1$, and is obtained for the weaker case of bounded random variables, whereas we cover possibly unbounded random variables.  

We obtained the expression for the gradient of $\srt$, and showed that the gradient is continuous on $\BB$. It is natural to examine properties such as Lipschitz continuity and smoothness of $\srt$, considering the fact that in the risk-neutral setting, these properties have been extensively used to establish convergence guarantees for gradient-based algorithms. In the following, we show that these properties also hold for the UBSR objective.  
\begin{lemmma}\label{lemma:h-Lipschitz}
    Suppose the assumptions of \Cref{theorem:ubsr-gradient}, and \cref{assumption:l-b1-increasing,assumption:F-bounded-higher-moment,assumption:loss-fn-smooth} hold. Then $\srt$ is $K_0$-Lipschitz, i.e., 
    \begin{align*}
    |\srth{1} - \srth{2}| \leq K_0\norm{\theta_1 - \theta_2}_2, \forall \theta_1,\theta_2 \in \BB,
\end{align*}
where $K_0=M \sqrt{\frac{S^2 \sigma_0^2}{b^2}+1}$, $\sigma_0$ and $M$ are as in \Cref{assumption:F-bounded-higher-moment}, and $b$ and $S$ are as in \cref{assumption:l-b1-increasing,assumption:loss-fn-smooth} respectively.
\end{lemmma}
\begin{myproof}The assumptions of \Cref{theorem:ubsr-gradient} are satisfied. Then by \Cref{theorem:ubsr-gradient}, $\srt$ is continuously differentiable, and for every $\theta \in \BB$, the gradient is bounded as follows.
    \begin{align*}
        \norm{\nabla \srth{}}_2 &=  \norm{\frac{-\Exp \Big[ l'(-F(\theta,\xi)-\srth{}) \nabla F(\theta,\xi) \Big]}{\Exp \big[ l'(-F(\theta,\xi)-\srth{}) \big]}}_2 \\
        &\leq \sqrt{\frac{\Exp\left[l'\left(-F(\theta,\xi)-\srth{}^2\right)\right]}{(\Exp\left[l'\left(-F(\theta,\xi)-\srth{}\right)\right])^2}} M_0.
    \end{align*}
    The above inequality follows from \Cref{lemma:uv-2-norm} in the supplementary with $\mathbf{V} = \nabla F(\theta,\xi)$ and $U=\frac{l'\left( -F(\theta,\xi)-\srth{}\right)}{\Exp\left[l'\left(-F(\theta,\xi)-\srth{}\right)\right]}$, and subsequently applying \Cref{assumption:F-bounded-higher-moment}. Then we have
    \begin{align*}
        \norm{\nabla \srth{}}_2 &\leq M_0 \sqrt{\frac{\textrm{Var}(l'(-F(\theta,\xi)-\srth{})}{(\Exp\left[l'\left(-F(\theta,\xi)-\srth{}\right)\right])^2}+1} \\
        &\leq M_0 \sqrt{\frac{S^2 \sigma_0^2}{b^2}+1} \leq M_0 \left(\frac{S \sigma_0}{b}+1\right),
    \end{align*}
    where the last inequality follows from \Cref{lemma:lip-function-variance-bound} in the supplementary and \cref{assumption:l-b1-increasing,assumption:F-bounded-higher-moment,assumption:loss-fn-smooth}. Since $\nabla \srt: \Theta \to \Rel^d$ is continuous and bounded above by $K_0 \triangleq M_0 \sqrt{\frac{S^2 \sigma_0^2}{b^2}+1}$, and the set $\Theta$ is compact, we conclude that $\srt$ is $K_0$-Lipschitz.
\end{myproof}

Next, we show that $\srt$ is smooth when the objective function $F$ satisfies \Cref{as:lp_moments_F_dF} and the loss function $l$ satisfies the smoothness assumption in \Cref{assumption:loss-fn-smooth}. 
\begin{lemmma}\label{lemma:h-smooth}
    Suppose the loss function $l$ is convex and continuously differentiable, and \cref{assumption:sr-base-assumption,as:lp_moments_F_dF} are satisfied. Then $\srt$ is $L$-smooth, where $L$ is as defined in \Cref{as:lp_moments_F_dF}.
\end{lemmma}
\begin{myproof}
Let $\theta_1,\theta_2 \in \Theta$. By \Cref{prop:proposition-g-root}, we have $g(\srth{1}, \theta_1) = g(\srth{2}, \theta_2) = 0$. Therefore, we have
\begin{align*}
    0 &= \Exp\left[l\left(-F\left(\theta_1,\xi\right)) - \srth{1}\right) - l\left(-F\left(\theta_2,\xi\right)) - \srth{2}\right) \right] \\
    \numberthis \label{eq:sr-smoothness-t1}
    &\leq \Exp\left[l'\left(-F\left(\theta_1,\xi\right)) - \srth{1}\right) \left(F\left(\theta_2,\xi\right)) -F\left(\theta_1,\xi\right) + \srth{2} - \srth{1}\right)\right].
\end{align*}
The above inequality follows because $l$ is convex and continuously differentiable. Then, by the smoothness assumption on $F$, we have $F\left(\theta_2,\xi\right)) -F\left(\theta_1,\xi\right) \leq \nabla F(\theta_1,\xi)^T(\theta_2-\theta_1) + \frac{L}{2}\norm{\theta_1 - \theta_2}^2$. Applying this inequality to \cref{eq:sr-smoothness-t1}, we have
\begin{align*}
    0 &\leq \Exp\left[l'\left(-F\left(\theta_1,\xi\right)) - \srth{1}\right)\right]\left(\srth{2} - \srth{1} + \frac{L}{2}\norm{\theta_1 - \theta_2}^2\right) \\
    &+\Exp\left[l'\left(-F\left(\theta_1,\xi\right)) - \srth{1}\right)F(\theta_1,\xi)\right]^T(\theta_2-\theta_1) \\
    \implies& -\srth{1} - (-\srth{2}) \leq -\nabla \srth{1}^T(\theta_2-\theta_1) + \frac{L}{2}\norm{\theta_1 - \theta_2}^2.
\end{align*}
The above inequality implies that $\srth{\cdot}$ is $L$-smooth, and the claim of the lemma follows. 
\end{myproof}

For cases when the loss function is not convex, one can still establish a smoothness result on $\srth{}$ if \cref{assumption:sr-base-assumption,assumption:loss-fn-smooth,as:lp_moments_F_dF,assumption:F-bounded-higher-moment} are satisfied. We make this claim precise in the following result. 

\begin{lemmma}\label{lemma:sr-smooth-non-convex-loss-fn}
    Suppose the loss function is $S$-smooth, and the objective function is $L$-smooth and $M$-Lipschitz. Suppose the assumptions of \Cref{theorem:ubsr-gradient} are satisfied. Then $\srth{\cdot}$ is $K$-smooth, where K=$LSM$.
\end{lemmma}
\begin{myproof}
    Let $\theta_1,\theta_2 \in \Theta$. By \Cref{prop:proposition-g-root}, we have $g(\srth{1}, \theta_1) = g(\srth{2}, \theta_2) = 0$. Therefore, we have
\begin{align*}
    &0 = \Exp\left[l\left(-F\left(\theta_1,\xi\right)) - \srth{1}\right) - l\left(-F\left(\theta_2,\xi\right)) - \srth{2}\right) \right] \\
    &\leq \Exp\left[l'\left(-F\left(\theta_2,\xi\right)) - \srth{2}\right)\left(F\left(\theta_2,\xi\right)) -F\left(\theta_1,\xi\right) + \srth{2} - \srth{1}\right) \right] \\
    &+ \Exp\left[\frac{S}{2} \norm{F\left(\theta_2,\xi\right)) -F\left(\theta_1,\xi\right) + \srth{2} - \srth{1}}_2^2\right] \\
    \numberthis \label{eq:sr-smoothness-t2}
    &\leq \Exp\left[l'\left(-F\left(\theta_2,\xi\right)) - \srth{2}\right)\left(F\left(\theta_2,\xi\right)) -F\left(\theta_1,\xi\right) + \srth{2} - \srth{1}\right)\right] \\
    &+ \frac{S(M+K_0)}{2}\norm{\theta_1 - \theta_2}_2^2,
\end{align*}
where the first inequality follows because $l$ is $S$-smooth, while the second inequality follows because $F(\cdot,\xi)$ and and $\srt$ are $M$ and $K_0$ Lipschitz respectively. Then, by the smoothness assumption on $-F$, we have $ -F\left(\theta_1,\xi\right) - (-F\left(\theta_2,\xi\right)) \leq - \nabla F(\theta_2,\xi)^T(\theta_1-\theta_2) + \frac{L}{2}\norm{\theta_1 - \theta_2}^2$. Applying this inequality to \cref{eq:sr-smoothness-t2}, we have
\begin{align*}
    &0 \leq \\
    &\Exp\left[l'\left(-F\left(\theta_2,\xi\right)) - \srth{2}\right)\left(- \nabla F(\theta_2,\xi)^T(\theta_1-\theta_2) + \frac{L}{2}\norm{\theta_1 - \theta_2}^2 + \srth{2} - \srth{1}\right)\right] \\
    &+ \frac{S(M+K_0)}{2}\norm{\theta_1 - \theta_2}_2^2 \\
    & = \frac{S(M+K_0)}{2}\norm{\theta_1 - \theta_2}_2^2 -\Exp\left[l'\left(-F\left(\theta_2,\xi\right)) - \srth{2}\right)\left(\nabla F(\theta_2,\xi)\right)\right]^T\left(\theta_1 - \theta_2\right) \\
    &+ \Exp\left[l'\left(-F\left(\theta_2,\xi\right)- \srth{2}\right)\right]\left(\frac{L}{2}\norm{\theta_1 - \theta_2}^2 + \srth{2} - \srth{1}\right).
\end{align*}
Dividing by $\Exp\left[l'\left(-F\left(\theta_2,\xi\right)- \srth{2}\right)\right]$ and rearranging, we have
\begin{align*}
    \srth{1} - \srth{2} \leq \nabla \srth{2}^T(\theta_1-\theta_2) + \frac{L}{2}\norm{\theta_1 - \theta_2}^2 + \frac{S(M+K_0)}{2 b}\norm{\theta_1 - \theta_2}_2^2.
\end{align*}
Comparing the above inequality with Theorem 2.1.5 of \cite{nesterov_introductory_2004}, we conclude that $\srt$ is $L+\frac{S(M+K_0)}{b}$ smooth and the claim of the lemma follows. 
\end{myproof}

\subsection{UBSR gradient estimation.}\label{ss:ubsr_grad_estimator}
Obtaining unbiased estimates of the gradient in (\ref{h_grad_0}) is difficult, not only because it is a ratio of expectations, but also because the gradient expression in (\ref{h_grad_0}) features $\srth{}$ and the latter quantity needs to be estimated from samples.
Thus, one can only obtain a biased UBSR gradient estimator by replacing the expectations in (\ref{h_grad_0}) with their sample averages, and using UBSR estimates in place of the true UBSR. We present such an estimator for the UBSR gradient and subsequently show that the estimator's MSE is \order{1/m}, where \order{m} samples are used to form the estimator.

Fix a $\theta \in \Theta$. Our gradient estimator for $\nabla \srth{}$ is constructed from two batches of samples. The first batch contains $m_1$ i.i.d. samples of $F(\theta,\xi)$. The second batch is sampled independently of the first, and contains i.i.d. joint samples from $F(\theta,\xi)$ and $\nabla F(\theta,\xi)$. The double sampling from $F(\theta,\xi)$ to estimate $\nabla \srth{}$ is necessary to avoid cross terms and has been used previously by \cite{Hegde2024}.  From the gradient expression in (\ref{h_grad_0}), it is apparent that an estimate of $\srth{}$ is required to form an estimate for $\nabla \srth{}$. Using the first batch of $m_1$ samples, we estimate $\srth{}$ using the scheme presented in \Cref{sec:estimation}. Precisely, using the following shorthand notation for $m_1$ samples from $F(\theta,\xi):\mathbf{\hat{Z}}_\theta = \left[ \mathrm{\hat{Z}}_\theta^1, \mathrm{\hat{Z}}_\theta^2,\ldots, \mathrm{\hat{Z}}_\theta^{m_1} \right]$, we form the UBSR estimator $\srmk{1}{\mathbf{\hat{Z}}_\theta}$, where $\srmk{1}{\cdot}$ is as defined in \cref{eq:sr-m-definition}. Next, we define the following shorthand notation for the $m_2$ joint-samples from $F(\theta,\xi)$ and $\nabla F(\theta,\xi):\mathbf{Z}_\theta = \left[ \mathrm{Z}_\theta^1, \mathrm{Z}_\theta^2,\ldots, \mathrm{Z}_\theta^{m_2} \right]$ and $\mathbf{V}_\theta = \left[ \mathrm{V}_\theta^1, \mathrm{V}_\theta^2,\ldots, \mathrm{V}_\theta^{m_2} \right]$, respectively. Finally, we define the function $\J:\Rel^{m_1} \times \Rel^{m_2} \times \Rel^{m_2 \times d} \to \Rel^d$ as follows:
\begin{align}\label{eq:J-definition}
     \jt &\triangleq \frac{-\sum_{j=1}^{m_2} \Big[ l'(-\mathrm{z}^j_\theta-\srmk{1}{\mathbf{\hat{z}_\theta}})\cdot \mathbf{v}^j_\theta\Big]}{\sum_{j=1}^{m_2} \big[ l'(-\mathrm{z}^j_\theta-\srmk{1}{\mathbf{\hat{z}_\theta}}) \big]}.
\end{align}
%Given a $\theta \in \Theta$, the functions $\srmk{1}{\cdot},J^{m_1,m_2}_\theta(\cdot,\cdot,\cdot)$ are used to estimate $\srth{},\nabla \srth{}$ respectively. 
% To summarize, $\mathbf{\hat{Z}}_\theta$ is an $m_1$-dimensional random vector such that each $\hat{Z}^j_\theta$ is an i.i.d. copy of $F(\theta,\xi)$. $\mathbf{Z}_\theta$ and $\mathbf{V}_\theta$ are $m_2$-dimensional random vector and $m_2 \times d$ dimensional random matrix respectively, such that each pair $Z^j_\theta,{V}^j_\theta$ is jointly-sampled, and i.i.d. copy of $F(\theta,\xi),\nabla F(\theta,\xi)$, respectively. Furthermore, $\mathbf{Z}_\theta$ and $\mathbf{V}_\theta$ are sampled independently from $\mathbf{\hat{Z}}_\theta$. 
Then $\srmk{1}{\mathbf{\hat{Z}_\theta}}$ and $\Jt$ are our proposed estimators for $\srth{}$ and $\nabla \srth{}$ respectively. The sample sizes $m_1$ and $m_2$ control the bias of the gradient estimator, and we now present the MAE and MSE bounds on the gradient estimator $\J$, as a function of $m_1$ and $m_2$. To derive these bounds on $\J$, we utilize the bounds on the UBSR estimate $\srmk{1}{\cdot}$ given in \Cref{lemma:ubsr-saa-bounds}. For the sake of readability, we first rephrase a variance assumption made in \Cref{lemma:ubsr-saa-bounds}.%, and subsequently, we restate \Cref{lemma:ubsr-saa-bounds} as \Cref{lemma:ubsr-saa-bounds-opt} below.
\begin{assumption}\label{as:variance-bound-only-l} 
There exist $\sigma_1>0$ such that Var$\left(l(-F(\theta,\xi)-\srth{})\right) \leq \sigma_1^2$, for all $\theta \in \Theta$.
\end{assumption}

We now present error bounds on the gradient estimator $\J$.

\begin{lemmma}\label{lemma:ubsr-grad-est-bound}
    Suppose \cref{assumption:l-b1-increasing,as:variance-bound-only-l,assumption:sr-base-assumption,assumption:F-bounded-higher-moment,as:variance-bound-l-F,assumption:loss-fn-smooth} are satisfied. Then, for every $\theta \in \Theta$ and every $m_1,m_2 \in \N$, the gradient estimator $\hat{J} \triangleq \Jt$ defined in \cref{eq:J-definition} satisfies 
    \begin{align*}
        \Exp \left[ \norm{\hat{J} - \nabla \srth{}}_2 \right] \leq \frac{D_1}{\sqrt{m^*}}, 
        \;\;\textrm{ and } \;\;
        \Exp \left [\norm{\hat{J} - \nabla \srth{}}_2^2 \right] \leq \frac{D_2}{m^*},
    \end{align*}
    where $D_1 = \frac{1}{b_1} \left[ M_0S_1\sigma_0 \left(\frac{S_1\sigma_0}{b_1} + 1\right) + 11 \log(e^2(d+1))T_1 \right] $ $+ \frac{M_0S_1\sigma_1}{b_1^2}\left(\frac{S_1\sigma_0}{b_1}+2\right)$, $D_2 =\frac{4M_0^2S_1^2\sigma_1^2}{b_1^4}\left(\left(\frac{S_1\sigma_0^2}{b_1}+2\right)^2 \right)$ $+\frac{4}{b_1^2} \left[M_0^2S_1^2\sigma_0^2\left(\frac{S_1\sigma_0}{b_1}+1\right)^2 + 72 \log^2(e^2(d+1))T_1^2 \right]$, and $m^* \triangleq \min\{m_1,m_2\}$.
\end{lemmma}
\begin{myproof}
Let $q \in \{1,2\}$, and let $A,A',A_m$ and $B,B',B_m$ be defined as follows:
    \begin{align*}
        &A \triangleq \A,\,\, 
        B \triangleq \B,\\
        &A' \triangleq \Ap, \,\,
        B' \triangleq \Bp,\\
        &A_m \triangleq \Am, \,\,
        B_m \triangleq \Bm. 
    \end{align*}
    Note that by \Cref{assumption:l-b1-increasing}, $B_m,B' \geq b$ a.s. and $B\geq b$.
    
    We divide the proof into three steps. We derive intermediate results in the first two steps, and use them in the third step to derive bound on the gradient estimator $\Jt$. 
    
    \text{Step 1.} (A bound on the $q^{th}$ moment of the $2$-norm of the random vectors $A_m-A'$ and $B_m-B'$) 
    \begin{align}
        \nonumber &\Exp \left[ \norm{A_m - A'}_2^q \right] =  \Exp \left[ \left\lVert \sum_{j=1}^{m_2} \left[l'(-\mathrm{Z}^j_\theta-\srmk{1}{\mathbf{\hat{Z}_\theta}}) - l'(-\mathrm{Z}^j_\theta-\srth{}) \right] \mathbf{V}^j_\theta\right]\Bigg\rVert_2^q\right] \cdot \frac{1}{m_2^q}\\
        \label{eq:mq-lemma}
        &\leq \Exp \left[ \sum_{j=1}^{m_2} \bigg\lVert \left[ l'(-\mathrm{Z}^j_\theta-\srmk{1}{\mathbf{\hat{Z}_\theta}}) - l'(-\mathrm{Z}^j_\theta-\srth{})\right] \mathbf{V}^j_\theta\bigg\rVert_2^q \right] \frac{m_2^{q-1}}{m_2^q} \\
        \nonumber &= \frac{1}{m_2}\Exp \left[ \sum_{j=1}^{m_2} \left|l'(-\mathrm{Z}^j_\theta-\srmk{1}{\mathbf{\hat{Z}_\theta}})- l'(-\mathrm{Z}^j_\theta-\srth{}) \right|^q \norm{\mathbf{V}^j_\theta}_2^q\right] \\
        \nonumber &\leq \frac{S^q}{m_2}\Exp \left[\sum_{j=1}^{m_2}  \left|\srmk{1}{\mathbf{\hat{Z}_\theta}}) - \srth{} \right|^q \norm{\mathbf{V}^j_\theta}_2^q\right] \\
        \label{eq:AmA-prime}
        &= \frac{S^q}{m_2} \left[\sum_{j=1}^{m_2}  \Exp \left[ \left|\srmk{1}{\mathbf{\hat{Z}_\theta}}) -\srth{} \right|^q\right] \Exp \left[\norm{\mathbf{V}^j_\theta}_2^q \right] \right] \leq  M^qS^q \Exp \left[ \left| \srmk{1}{\mathbf{\hat{Z}_\theta}})-\srth{} \right|^q \right],
    \end{align}
    where the inequality in \cref{eq:mq-lemma} follows from \Cref{lemma-normed-sum} in the supplementary material. The second inequality follows from \Cref{assumption:loss-fn-smooth}, which implies that $l'$ is $S$-Lipschitz. The equality in \cref{eq:AmA-prime} uses the independence of $\Zh_\theta$ and $\mathbf{V}_\theta$, while the final inequality follows from \Cref{assumption:F-bounded-higher-moment} after utilizing the assumption that each $\mathbf{V}_\theta^j$ is an identical copy $\nabla F(\theta,\xi)$. Using similar arguments, we have
    \begin{align} \label{eq:BmB-prime}
        \Exp &\left[ \left| B_m - B' \right|^q \right] \leq S^q \Exp \big[ \big| \srmk{1}{\mathbf{\hat{Z}_\theta}})-\srth{} \big|^q \big].
    \end{align}
    \text{Step 2.} (Bounds on $A'-A$ and $B'-B$) \\
    Let $\mathbf{Y}_j \triangleq l'(-\mathrm{Z}^j_\theta-\srth{})\mathbf{V}^j_\theta - \Exp\left[l'(-F(\theta,\xi)-\srth{})\nabla F(\theta,\xi)\right]$ for every $j \in \{1,2,\ldots,d\}$. Letting $\mathbf{W}=\sum_{j=1}^{m_2} \mathbf{Y}_j$, we have 
    \begin{align*}
         \Exp &\left[ \norm{A' - A}_2^q \right] 
         = \frac{1}{m_2^q} \Exp \left[ \norm{\mathbf{W}}_2^q\right] \leq \frac{1}{m_2^q} \left[\Exp \left[ \norm{\mathbf{W}}_2^2\right]\right]^\frac{q}{2},
    \end{align*}
    where the last inequality is Lyapunov's inequality. Since each jointly-sampled pair $\langle \mathrm{Z}^j_\theta,\mathbf{V}^j_\theta\rangle$ is i.i.d., we note that $\mathbf{Y}_j$ is a zero-mean random vector and the assumptions of Theorem 1 of \cite{tropp-springer-2016} are satisfied. By Theorem 1 of \cite{tropp-springer-2016}, we have the following for $q=1$.
    \begin{align*}
        &\Exp \left[ \norm{A' - A}_2 \right] \\
        &\leq \frac{\sqrt{8\log(e^2(d+1)) \max\left\{ U, V \right\}}}{m_2}  + \frac{1}{m_2}\left[8\log(e^2(d+1))\sqrt{\Exp\left[\max_j \norm{Y_j}_2^2\right]}\right], 
    \end{align*}
    where $U= \norm{\sum_j \Exp\left[Y_j Y_j^{\mathrm{T}}\right]}_2$ and $V=\sum_j \Exp\left[\norm{Y_j}^2_2\right]$. Next, we note that $ U, V$, and the term inside the second square-root are all bounded above by $\sum_j \Exp\left[\norm{{\mathbf{Y}_j}}_2^2\right]$. We conclude this using simple matrix algebra and the independence of $\mathbf{Y}_j$'s. Then, we have 
    \begin{align*}
       \Exp\left[ \norm{A' - A}_2 \right] &\leq \frac{1}{m_2}\left( \sqrt{8\log(e^2(d+1))}+8\log(e^2(d+1))\right) \sqrt{\sum_j \Exp\left[\norm{{Y_j}}_2^2\right]} \\
       \numberthis \label{eq:AA-prime-1}
       &\leq \frac{11 \log(e^2(d+1))\sqrt{m_2 \Exp\left[\norm{{Y_1}}_2^2\right]}}{m_2} \leq  \frac{11 \log(e^2(d+1)) T}{\sqrt{m_2}}.
    \end{align*}
    The second inequality above follows because ${Y_j}$'s are identical, while the last inequality follows from \Cref{as:variance-bound-l-F}. Similarly, by Theorem 1 of \cite{tropp-springer-2016} and for $q=2$, we have
    \begin{equation}\label{eq:AA-prime-2}
        \Exp \left[ \norm{A' - A}_2^2 \right] \leq \frac{72 \log^2(e^2(d+1)) T^2}{m_2}.
    \end{equation}
    Let $Y_j \triangleq l'(-\mathrm{Z}^j-\srth{})$. Then $B = \Exp\left[Y_j\right],\forall j=1$ to $m_2$. Then, we have
    \begin{align*}
        \Exp \left[ \left| B' - B \right|^q \right] &\!=\! \frac{\Exp \left[ \left| \sum\limits_{j=1}^{m_2} Y_j - \Exp\left[\sum\limits_{j=1}^{m_2} Y_j\right]\right|^q\right]}{m_2} 
        \leq \frac{\left[\Exp \left[ \left| \sum\limits_{j=1}^{m_2} Y_j - \Exp\left[\sum\limits_{j=1}^{m_2} Y_j\right]\right|^2\right]\right]^{\frac{q}{2}}}{m_2^q}  \\
        \numberthis \label{eq:BB-prime}
        &=\frac{1}{m_2^q}\left[\text{Var}\left(\sum_{j=1}^{m_2} Y_j \right)\right]^{q/2} =\frac{1}{m_2^q}\left[m_2 \text{Var}\left(Y_1 \right)\right]^{q/2} \leq \frac{S^q\sigma_0^q}{m_2^{q/2}}.
    \end{align*}
    The first inequality is Lyapunov's inequality. The last equality follows by combining two facts: a) variance of the sum of independent random variables equals the sum of their variances, and b) each $Y_j$ is i.i.d. The last inequality follows from \Cref{lemma:lip-function-variance-bound} in the supplementary material.
    
\textbf{Step 3.} (Main bound on the error term: $\Jt - \nabla \srth{}$) \\
Since the assumptions of \Cref{theorem:ubsr-gradient} hold, $\nabla h$ exists and is well-defined. Then, we have
    \begin{align*}
        &\norm{ \Jt - \nabla \srth{}}_2^q \\
        &= \norm{ B_m^{-1}A_m - B^{-1}A}_2^q = \norm{(BB_m)^{-1}(BA_m - AB_m)}_2^q \\
        &\leq (BB_m)^{-q}\norm{(BA_m - AB_m)}_2^q = (BB_m)^{-q} \norm{B(A_m - A) + A(B - B_m)}^q_2 \\
        &\leq \frac{q}{B_m^{q}}\left(\norm{A_m - A}_2^q + B^{-q} \norm{A(B - B_m)}^q_2\right) \leq \frac{q}{b^q}\left[\norm{A_m - A}_2^q + |B - B_m|^q B^{-q} \norm{A}_2^q\right] \\
        \numberthis \label{eq:tmppp}
        &\leq \frac{q}{b^q}\left[\norm{A_m - A}_2^q + M_0^q \left(\frac{S \sigma_0}{b}+1\right)^q |B - B_m|^q\right]
    \end{align*}
    For the second inequality above, we first use the Minkowski's inequality, followed by the 'power of sum' inequality: $(a_1+a_2)^q \leq q(a_1^q + a_2^q)$, which holds for $a_1 \geq 0,a_2\geq 0$ and $q \in \{1,2\}$. The third inequality follows from \Cref{assumption:l-b1-increasing} and the Cauchy-Schwartz inequality for the Euclidean norm on $\Rel$. For the fourth inequality, we use the following bound.
    \begin{align*}
        B^{-q} &\norm{A}_2^q = B^{-q} \norm{\Exp \left[ l'(-F(\theta,\xi)-\srth{}) \nabla F(\theta,\xi) \right]}_2^q \\
        &\leq B^{-q}\left(\Exp \left[ l'(-F(\theta,\xi)-\srth{})^2\right]\right)^{q/2}\norm{ \nabla F(\theta,\xi)}_{\Leb{2}}^q \\
        &=  B^{-q} \left(\mathrm{Var}\left(l'(-F(\theta,\xi)-\srth{})\right)+B^2\right)^{q/2}\norm{ \nabla F(\theta,\xi)}_{\Leb{2}}^q \\
        &= \left(\frac{\mathrm{Var}\left(l'(-F(\theta,\xi)-\srth{})\right)}{B^2}+1\right)^{q/2}\norm{ \nabla F(\theta,\xi)}_{\Leb{2}}^q \leq M_0^q \left(\frac{S \sigma_0}{b}+1\right)^q,
    \end{align*}
    where the first inequality follows from \Cref{lemma:uv-2-norm}. The last inequality follows from \cref{assumption:l-b1-increasing,assumption:F-bounded-higher-moment,lemma:lip-function-variance-bound}. Precisely, for \Cref{lemma:lip-function-variance-bound}, we use \Cref{assumption:F-bounded-higher-moment} and the fact that $l'$ is $S$-Lipschitz. Then, from (\ref{eq:tmppp}), we have
    \begin{align*}
          &\norm{ \Jt - \nabla \srth{}}_2^q \leq \frac{q^2}{b^q} \left[ \norm{A_m - A'}_2^q + \norm{A' - A}_2^q \right] \\
          \numberthis \label{eq_pq}
          &+ \frac{q^2M^q}{b^q}\left(\frac{S_1 \sigma_0}{b}+1\right)^q\left[|B - B'|^q + |B' - B_m|^q\right],
    \end{align*}
    where we add and subtract $A'$ and $B'$ in the two terms of (\ref{eq:tmppp}) respectively, and apply the 'power of sum' inequality given earlier in the proof. Taking expectation on both sides of (\ref{eq_pq}), we have
    \begin{align*}
        &\Exp \left[ \norm{\Jt - \nabla \srth{}}_2^q \right] \leq \frac{q^2}{b^q} \left[ \Exp \norm{A_m - A'}^q_2 +  \Exp \left[\norm{A' - A}_2^q \right] \right] \\
        &+ \frac{q^2  M^q}{b^{q}}\left(\frac{S \sigma_0}{b}+1\right)^q \left[ \Exp\left[|B_m - B'|^q \right] + \Exp\left[|B' - B|^q \right]  \right].
    \end{align*}
    The \cref{assumption:l-b1-increasing,as:variance-bound-only-l} hold, which implies that the assumptions of \Cref{lemma:ubsr-saa-bounds} are satisfied. Using \cref{eq:AmA-prime,eq:BmB-prime,eq:AA-prime-1,eq:BB-prime} with $q=1$, and invoking \Cref{lemma:ubsr-saa-bounds}, we have
    \begin{align*}
       \Exp &\left[ \norm{\Jt - \nabla \srth{}}_2 \right] \\
       &\leq  \frac{MS\sigma_1}{b^2\sqrt{m^*}}\left[\frac{S\sigma_0}{b}+2\right] + \frac{1}{b\sqrt{m^*}} \left[ MS\sigma_0 \left(\frac{S\sigma_0}{b} + 1\right) + 11 \log(e^2(d+1))T \right],
    \end{align*}
    where $m^* \triangleq \min\left(m_1, m_2\right)$.
    Using \cref{eq:AmA-prime,eq:BmB-prime,eq:AA-prime-2,eq:BB-prime} with $q=2$, and invoking \Cref{lemma:ubsr-saa-bounds}, we have
    \begin{align*}
        \Exp &\left[ \norm{ \Jt - \nabla \srth{}}_2^2 \right] \leq \frac{4M^2S^2\sigma_1^2}{b^4 m^*}\left[1 + \left(\frac{S\sigma_0^2}{b}+1\right)^2 \right] \\
        &+ \frac{4}{b^2 m^*} \left[M^2S^2\sigma_0^2\left(\frac{S\sigma_0}{b}+1\right)^2 + 72 \log^2(e^2(d+1))T^2 \right].
    \end{align*} 
\end{myproof}
% \begin{myproof}
%     See \Cref{proof:ubsr_grad_est_bound} for the proof.
% \end{myproof}
\begin{remark}
    In \cite{zhaolin2016ubsrest,Hegde2024}, the authors consider the scalar case for UBSR optimization.   In \cite{zhaolin2016ubsrest}, the authors establish asymptotic consistency of the UBSR derivative estimator, while \cite{Hegde2024} establishes non-asymptotic error bounds for this estimator. In contrast, our result above applies to the multivariate case, and we provide non-asymptotic error bounds for the UBSR gradient estimator. These bounds can be used to infer asymptotic consistency.
\end{remark}

\paragraph{Practical technique for computing gradient of UBSR}
The gradient estimator in \cref{eq:J-definition} contains the UBSR estimate: $\srmk{1}{\mathbf{\hat{z}_\theta}}$ in the numerator as well as the denominator. As discussed in \Cref{sec:estimation}, $\srmk{1}{\mathbf{\hat{z}_\theta}}$ may not be available in closed form and in practice, one may use a modified gradient estimator where $\srmk{1}{\mathbf{\hat{z}_\theta}}$ in \cref{eq:J-definition} is replaced with its approximation $t_{m_1}$, given by the Algorithm \ref{alg:saa_bisect}. The result below shows that the bounds given by \Cref{lemma:ubsr-grad-est-bound} can be readily extended to this modified gradient estimator.
\begin{proposition}\label{prop:grad-est-bound-numerical-I}
    Let $\hat{J} \triangleq \hat{J}^{m,m}(\mathbf{\hat{Z}}_\theta,\mathbf{Z}_\theta,\mathbf{V}_\theta)$ be the gradient estimator constructed by replacing $\srmk{1}{\mathbf{\hat{Z}_\theta}}$ in (\ref{eq:J-definition}) with its approximation $t_m(\Zh_\theta)$ obtained from Algorithm \ref{alg:saa_bisect} using $\delta=\frac{d_1}{\sqrt{m}}$ for some $d_1>0$. Suppose the assumptions of \Cref{lemma:ubsr-grad-est-bound} hold. Then, for every $\theta \in \Theta$ and every $m \in \N$, we have the following error bounds on this gradient estimator:
    \begin{align*}
        \Exp \left[ \norm{\hat{J} - \nabla \srth{}}_2 \right] \leq \frac{\hat{D}_1}{\sqrt{m}}, 
        \;\; \textrm{ and } \;\; 
        \Exp \left [\norm{\hat{J} - \nabla \srth{}}_2^2 \right] \leq \frac{\hat{D}_2}{m},
    \end{align*}
    where $\hat{D}_1 = \frac{1}{b} \left[ M_0S\sigma_0 \left(\frac{S\sigma_0}{b} + 1\right) + 11 \log(e^2(d+1))T \right] + \frac{M_0S(\sigma_1+d_1)}{b^2}\left[\frac{S\sigma_0}{b}+2\right]$  and $\hat{D}_2 =  \frac{4}{b^2} \left[M_0^2S^2\sigma_0^2\left(\frac{S\sigma_0}{b}+1\right)^2 + 72 \log^2(e^2(d+1))T^2 \right] + \frac{8M_0^2S^2(\sigma_1^2+d_1^2)}{b^4}\left[1 + \left(\frac{S\sigma_0^2}{b}+1\right)^2 \right]$.
\end{proposition}
We omit the detailed proof of the proposition as the line of proof mirrors that of \Cref{lemma:ubsr-grad-est-bound}. The modified gradient estimator given by the proposition is constructed by replacing $\srmk{1}{\mathbf{\hat{Z}_\theta}}$ from \cref{eq:J-definition} with $t_m(\Zh_\theta)$ given by \cref{alg:saa_bisect}. One can conclude from the proof of \Cref{lemma:ubsr-grad-est-bound} that because of this replacement, the resulting proof of the proposition differs from that of \Cref{lemma:ubsr-grad-est-bound} by exactly one term. Precisely, instead of bounding the term: $\Exp \left[ \left| \srmk{1}{\mathbf{\hat{Z}_\theta}}-\srth{} \right|^q \right]$ as done in the Step 1 of the proof of \Cref{lemma:ubsr-grad-est-bound}, we now bound the term : $\Exp \left[ \left| t_m(\Zh_\theta) - \srth{} \right|^q \right]$. We bound this term by breaking $t_m(\Zh_\theta) - \srth{}$ into two separate terms: $t_m(\Zh_\theta) - \srmk{1}{\mathbf{\hat{Z}_\theta}}$ and $\srmk{1}{\mathbf{\hat{Z}_\theta}}-\srth{}$. The \Cref{alg:saa_bisect} guarantees that the first term is bounded by $\delta$, while the bound on the second term is obtained by \Cref{lemma:ubsr-saa-bounds}. The claim of the proposition now follows from the proof of \Cref{lemma:ubsr-grad-est-bound}.

\subsection{SG Algorithm for UBSR optimization}
\label{ss:sg-online}
In this section, we propose and analyze a SG algorithm that uses the gradient estimates described in the previous section.
\begin{figure}[!t]
\centering
\begin{algorithm}[H]
\caption{UBSR-SG}\label{alg:ubsr-minimization}
\begin{algorithmic}[1] 
    \REQUIRE $\theta_0 \sim \Theta$, batch sizes $\{m_1^k,m_2^k\}_{k \ge 1}$, step sizes $\{\alpha_k\}_{k \ge 1}$
    \FOR{$k= 1, 2, \ldots, n$}
        \STATE $\mathbf{\hat{Z}}^k_{\theta_{k-1}} \leftarrow m_1^k$ samples from $\sim F(\theta_{k-1},\xi)$.
        \STATE compute $\srt^{m_1^k}(\mathbf{\hat{Z}}^k_{\theta_{k-1}})$ using \cref{eq:sr-m-definition}.
        \STATE $\left( \mathbf{{Z}}^k_{\theta_{k-1}}, \mathbf{{V}}^k_{\theta_{k-1}}\right) \leftarrow m_2^k$ paired samples from $\left( F(\theta_{k-1},\xi), \nabla F(\theta_{k-1},\xi) \right)$.
        \STATE compute $\hat{J}^k = J^{m_1^k,m_2^k}(\mathbf{\hat{Z}}^k_{\theta_{k-1}},\mathbf{Z}^k_{\theta_{k-1}}, \mathbf{V}^k_{\theta_{k-1}})$ using \cref{eq:J-definition} and $\srt^{m_1^k}$.
        \STATE update $\theta_k \leftarrow \Pi_\Theta\left(\theta_{k-1} - \alpha_k \hat{J}_{k}\right)$
    \ENDFOR
    
    \RETURN $\theta_n$
\end{algorithmic}
\end{algorithm}
\end{figure}
% \begin{algorithm}
% \SetKwInOut{Input}{Input}\SetKwInOut{Output}{Output}
% \SetAlgoLined
% \Input{$\theta_0 \sim \Theta$, thresholds $\{\delta_k\}_{k \ge 1}$, batch sizes $\{m_k\}_{k \ge 1}$ and step sizes $\{\alpha_k\}_{k \ge 1}$.}
% \For{$k= 1, 2, \ldots, n$}{
%     sample $\mathbf{\hat{Z}}^k = [\hat{Z}^k_1, \hat{Z}^k_2, \ldots, \hat{Z}^k_{m_k} ]$\; 
%     compute $t^k$ with inputs $\delta_k,\mathbf{\hat{Z}}^k$ to the Algorithm \ref{alg:saa_bisect}\;
%     sample $\mathbf{Z}^k = [Z^k_1, Z^k_2, \ldots, Z^k_{m_k} ]$\; 
%     compute $J^k = \hat{J}_{\theta_{k-1}}^{m_k}(\mathbf{Z}^k, \mathbf{\hat{Z}}^k)$ using $t^k$\;
%     update $\theta_k \leftarrow \Pi_\Theta\left(\theta_{k-1} - \alpha_k J^k\right)$\;
% }
% \caption{UBSR-SG}\label{alg:ubsr-minimization}
% \Output{$\theta_n$}
% \end{algorithm}
\Cref{alg:ubsr-minimization} presents the pseudocode for the SG algorithm for UBSR minimization. 
In this algorithm, $\Pi_{\Theta}(x) \triangleq \arg \min_{\theta \in \Theta} 
 \norm{x - \theta}_2$ denotes the operator that projects onto the convex and compact set $\Theta$.  In each iteration $k$ of this algorithm, we use the iterate $\theta_{k-1}$ and obtain $m_1^k$ samples from $F(\theta_{k-1},\xi)$ to estimate $\srth{k-1}$. Then, we obtain $m_2^k$ joint-samples from $\left(F\left(\theta_{k-1},\xi\right), \nabla F\left(\theta_{k-1},\xi\right)\right)$, that are independent of the previous samples and use them to estimate $\nabla \srth{k-1}$. Then, we perform the update given in the \Cref{alg:ubsr-minimization} to obtain the next iterate. 

In the following, we present three results on the convergence of the iterates of \Cref{alg:ubsr-minimization} under different classes of the UBSR objective $\srt$, namely, strongly convex, convex, and non-convex. All three results are stated under the assumption that $\srt$ is smooth. The sufficient conditions to ensure convexity and strong convexity of $\srt$ are given in \Cref{lemma:h-convex-general}, while the sufficient conditions for smoothness of $\srt$ are provided in \cref{lemma:h-smooth,lemma:sr-smooth-non-convex-loss-fn}.
\subsubsection{The strongly-convex case.}
 For Algorithm \ref{alg:ubsr-minimization}, we derive non-asymptotic bounds for the choice of the increasing batch sizes, and under the following strong-convexity assumption on the objective.
 %Further, non-asymptotic analysis of stochastic optimization algorithms under a strong-convexity assumption has received a lot of research attention, cf.  \cite{ghadimi-lan-strongly-convex-stoc-opt,RakhlinSS12-strongly-cvx-sco,stoc-bias-reduced-gradient-neurips}. 
\begin{assumption}\label{as:strong-convexity}
    There exists $\mu>0$ such that $\srt$ is $\mu$-strongly convex, i.e.,  $\nabla^2 \srth{} - \mu I \succeq 0, \forall \theta \in \Theta$.
\end{assumption}
While \Cref{lemma:h-convex-general} provides sufficient conditions to infer that \Cref{as:strong-convexity} holds, the strong convexity requirement above can also be shown to hold under different hypotheses on the loss function $l$, the objective function $F$, and the noise $\xi$. For instance, the strong convexity assumption above holds in a portfolio optimization example with Gaussian noise. Thus, instead of imposing constraints on $l,F$, and $\xi$ that could vary depending on the problem instance, \Cref{as:strong-convexity} subsumes all these variations into one single problem-independent assumption. We now present a result that bounds the error on the last iterate of \Cref{alg:ubsr-minimization} under \Cref{as:strong-convexity}.  
\begin{theorem}\label{theorem:sgd_convergence}
    Suppose $\srt$ is $K$-smooth and the assumptions of \Cref{lemma:ubsr-grad-est-bound} are satisfied. Suppose \Cref{as:strong-convexity} is satisfied for some $\mu>0$. Let $\theta_*$ denote the minimizer of $\srth{}$ defined in \cref{eq:sr-theta} and we assume that $\nabla \srth{*} = 0$. Let $c\geq \frac{3}{2\mu}$, and Let $\left\{\theta_k\right\}_{k=1}^n$ be the iterates obtained by running \Cref{alg:ubsr-minimization} with $\alpha_k = \frac{c}{k^a}$, and $m_1^k=m_2^k=k,\,\forall k$. Then for all $n \in \N$, we have
    \begin{align*}
        % &\Exp \left[ \srth{n} - \srth{*}\right] \leq \frac{K}{2} \Exp \left[ \norm{\theta_n - \theta_*}_2^2 \right], \;\; \text{ and}\\
        &\Exp\left[\norm{\theta_n - \theta_*}_2^2\right] \leq \frac{\Exp\left[\norm{\theta_0 - \theta_*}_2^2\right]}{(n+1)^{3}} + \frac{2^{2 \mu c }c^2{D}_2}{\left(n+1\right)^{2}} 
        + C \left(\frac{2^{\mu c}c D_1}{n+1}+\frac{2 \Exp\left[ \norm{\theta_0 - \theta_*}_2\right]}{\left(n+1\right)^{2}}\right),
    \end{align*}
    where $C=\exp{\left(c^2K^2\right)}\left(2^{2 \mu c}\left(1+c K\right)c{D}_1\right)$ and $D_1,D_2$ are as given in \Cref{lemma:ubsr-grad-est-bound}.
\end{theorem}
\begin{myproof}
%We split the proof into three parts. In Part I, we derive some intermediate results that are applied in the later parts of the proof. In part II, we derive an MAE bound on the last iterate of the SG algorithm, whereas in part III, we derive an MSE bound on the last iterate of the SG algorithm.
Recall that the objective function $\srt$ is $\mu$-strongly convex and $K$-smooth. $\theta_0$ is chosen arbitrarily and $\left\{\theta_1,\theta_2,\ldots,\theta_n\right\}$ are the random iterates of the SG algorithm. Using the notation $z_k \triangleq \theta_k - \theta^*$, we have the following w.p. $1$.
\begin{align*}
    \norm{z_{n-1} - \alpha_n \nabla \srth{n-1}}_2^2  &= \norm{z_{n-1}}_2^2 + \alpha_n^2 \norm{\nabla \srth{n-1}}_2^2 - 2 \alpha_n \left\langle z_{n-1}, \nabla \srth{n-1}\right\rangle \\
    \numberthis \label{eq:local-reference-1a}
    &\leq (1 + \alpha_n^2K^2)\norm{z_{n-1}}_2^2 - 2 \alpha_n \left\langle z_{n-1}, \nabla \srth{n-1}\right\rangle.
\end{align*}
The above inequality follows because $\nabla \srt$ is $K$-Lipschitz and $\nabla \srth{*} = 0$. Since $\srt$ is a differentiable and $\mu$-strongly convex function, by Definition 2.1.3 of \cite{nesterov_introductory_2004}, we have $\srth{1} \geq \srth{2} + \left\langle \nabla \srth{2}, \theta_2 - \theta_1 \right\rangle + \frac{\mu}{2}\norm{\theta_1 - \theta_2}_2,$ for every $\theta_1,\theta_2 \in \Theta$. Putting $\theta_1 = \theta_{n-1},\theta_2 = \theta^*$ in the identity and using the condition: $\nabla \srth{*} = 0$, we have $\srth{*} - \srth{n-1}) \leq \frac{-\mu}{2}\norm{z_{n-1}}_2^2$. Furthermore, by putting $\theta_1 = \theta^*$ and $\theta_2 = \theta_{n-1}$ in the identity, we have
\begin{align*}
    -\left\langle z_{n-1}, \nabla \srth{n-1}\right\rangle \leq \srth{*} - \srth{n-1}  - \frac{\mu}{2}\norm{z_{n-1}}_2^2 \leq -\mu \norm{z_{n-1}}_2^2.
\end{align*}
Substituting the above result back in (\ref{eq:local-reference-1a}), we have
\begin{equation}\label{eq:local-reference-1b}
    \norm{z_{n-1} - \alpha_n \nabla \srth{n-1}}_2^2 \leq (1 -2\alpha_n \mu + \alpha_n^2K^2)\norm{z_{n-1}}_2^2.
\end{equation}
Next, recall that for a fixed $\theta \in \Theta$ and batch sizes $m_1,m_2$, and by \Cref{prop:grad-est-bound-numerical-I}, the gradient estimator of $\Jt$ satisfies
\begin{align}\label{eq:gradient-bounds}
        \Exp \left[ \norm{\Jt - \nabla \srth{}}_2 \right] \leq \frac{{D}_1}{\sqrt{\max\{m_1,m_2\}}}, 
        \text{ and } \\
        \Exp \left [\norm{\Jt - \nabla \srth{}}_2^2 \right] \leq \frac{{D}_2}{\max\{m_1,m_2\}},
\end{align}
%\Cref{assumption:general-gradient-estimator}. %For each iteration $k$ of the SG algorithm, we have batch size $m_k$ and a random iterate $\theta_{k-1}$ that we use to form the gradient estimator $J_{m_k}(\theta_{k-1}, \mathbf{Z})$. 
Note that the above bounds apply to a fixed $\theta \in \Theta$. We now extend these bounds to the gradient estimators $\{J^{m_1^{k},m_2^{k}}(\mathbf{\hat{Z}}_{\theta_{k-1}},\mathbf{Z}_{\theta_{k-1}},\mathbf{V}_{\theta_{k-1}})\}_{k \in \mathcal{N}}$ that correspond to the random iterates $\{\theta_{k-1}\}_{k \in \mathcal{N}}$ given by \Cref{alg:ubsr-minimization}. We define the following shorthand notation $\zeta_k \triangleq J^{m_1^{k},m_2^{k}}(\mathbf{\hat{Z}}_{\theta_{k-1}},\mathbf{Z}_{\theta_{k-1}},\mathbf{V}_{\theta_{k-1}}) - \nabla \srth{k-1}$. Define filtration $\mathcal{F}_0 = \sigma(\theta_{0})$ and $\mathcal{F}_{k} = \sigma\left(\theta_0, \mathbf{\hat{Z}}_{\theta_{0}},\mathbf{Z}_{\theta_{0}},\mathbf{V}_{\theta_{0}}, \mathbf{\hat{Z}}_{\theta_{1}},\mathbf{Z}_{\theta_{1}},\mathbf{V}_{\theta_{1}},\ldots,\mathbf{\hat{Z}}_{\theta_{k-1}},\mathbf{Z}_{\theta_{k-1}},\mathbf{V}_{\theta_{k-1}}\right), \forall k \in \mathcal{N}$. By the SGD update in \Cref{alg:ubsr-minimization}, $\theta_{k-1}$ is $\mathcal{F}_{k-1}$ measurable, and by the sampling procedure to construct $J^{m_1^k,m_2^k}$, we have $\mathbf{\hat{Z}}_{\theta_{k-1}} \perp \mathcal{F}_{k-1}$, $\mathbf{{Z}}_{\theta_{k-1}} \perp \mathcal{F}_{k-1}$ and $\mathbf{V}_{\theta_{k-1}} \perp \mathcal{F}_{k-1}$. Then by the Lemma 2.3.4 (Independence Lemma) of \cite{shreve_stochastic_2004}, following holds for every $k \in \mathbb{N}$:
\begin{align}\label{eq:sg-independence-lemma}
    \Exp\left[ \left. \norm{\zeta_k}_2  \right| \mathcal{F}_{k-1} \right] \leq \frac{{D}_1}{\sqrt{\max\{m_1^k,m_2^k\}}}, \;\; \textrm{ and } \;\;
    \Exp\left[ \left. \norm{\zeta_k}_2^2  \right| \mathcal{F}_{k-1} \right] \leq \frac{{D}_2}{\max\{m_1^k,m_2^k\}}.
\end{align}

\paragraph{MAE bound proof:}
We now derive MAE bounds on the last iterate of the SG algorithm. For each iteration $n\ \in \textrm{N}$ of the SG update, we have $z_n = \Pi_\Theta\left(\theta_{n-1} - \alpha_n \left(\nabla \srth{n-1}+\zeta_n\right)\right) - \theta^*$. Note that $\theta^* \in \Theta$ holds, and therefore, $\theta^* = \Pi_\Theta(\theta^*)$. Using this identity along with the non-expansive property of the projection operator, we have the following w.p. $1$.
\begin{align*}
    \norm{z_n}_2 &\leq \norm{z_{n-1} - \alpha_n \left(\nabla \srth{n-1}+\zeta_n\right)}_2 \\
    &\leq \norm{z_{n-1} - \alpha_n \nabla \srth{n-1}}_2 + \alpha_n \norm{\zeta_n}_2 \leq \sqrt{1 - 2\alpha_n \mu + \alpha_n^2K^2}\norm{z_{n-1}}_2 + \alpha_n \norm{\zeta_n}_2,
\end{align*}
where the last inequality follows from \cref{eq:local-reference-1b}. Here, the square-root is well-defined because $1 - 2\alpha_k \mu + \alpha_k^2K^2$ is non-negative for every $k$. Indeed, $(1 - 2\alpha_k \mu + \alpha_k^2K^2) \geq (1 - 2\alpha_k \mu + \alpha_k^2\mu^2) \geq 0$, where we used the fact that for a $K$-smooth and $\mu$-strongly convex function, $K \geq \mu$ holds. After unrolling the above inequality, following holds w.p. $1$:
\begin{align*}
    \norm{z_n}_2 &\leq \norm{z_0}_2 \left(\prod_{k=1}^n \sqrt{1 - 2\alpha_k \mu + \alpha_k^2K^2}\right) + \sum_{k=1}^n \left[\left(\alpha_k \norm{\zeta_k}_2\right) \right] \left(\prod_{j=k+1}^n \sqrt{1 - 2\alpha_j \mu + \alpha_j^2K^2}\right) \\
    \numberthis \label{eq:local-reference-4}
    &= \norm{z_0}_2  \sqrt{\prod_{k=1}^n \left( 1 - 2\alpha_k \mu + \alpha_k^2K^2\right)} + \sum_{k=1}^n \left[\left(\alpha_k \norm{\zeta_k}_2\right) \right] \sqrt{\prod_{j=k+1}^n \left( 1 - 2\alpha_j \mu + \alpha_j^2K^2\right)}.
\end{align*}
Note that if $0\leq a_j \leq b_j,\forall j$ then $\Pi_j a_j \leq \Pi_j b_j$. Let $a_j = (1 - 2\alpha_j \mu + \alpha_j^2K^2)$ and $b_j = \exp{\left(2\alpha_j \mu + \alpha_j^2K^2\right)}$. Then, we apply the identity: $1+x \leq e^x,\forall x \in \Rel$ to infer that $a_j\leq b_j,\forall j$. Then, we have
\begin{equation}\label{eq:product-1-plus-x}
    \prod_{j=k+1}^n \left( 1 - 2\alpha_j \mu + \alpha_j^2K^2\right) \leq \sum_{j=k+1}^n \exp{\left(- 2\alpha_j \mu + \alpha_j^2K^2\right)}.
\end{equation}
Substituting the above result in (\ref{eq:local-reference-4}), we have
\begin{align*}
    \norm{z_n}_2 & \leq \norm{z_{0}}_2 \exp{\left(\sum_{k=1}^n\left(- \alpha_k \mu + \frac{\alpha_k^2K^2}{2}\right)\right)} + \sum_{k=1}^n\exp{\left(\sum_{j=k+1}^n\left(- \alpha_j \mu + \frac{\alpha_j^2K^2}{2}\right)\right)}\alpha_k \norm{\zeta_k}_2 \\
    \numberthis \label{eq:local-reference-2}
    &\leq \exp{\left(\sum_{j=1}^n \frac{\alpha_j^2K^2}{2}\right)} \left[\norm{z_{0}}_2  \exp{\left(\sum_{k=1}^n - \alpha_k \mu\right)} + \sum_{k=1}^n\exp{\left(\sum_{j=k+1}^n - \alpha_j \mu\right)}\alpha_k \norm{\zeta_k}_2\right].
\end{align*}
For the term: $\exp{\left(\sum_{j=1}^n \frac{\alpha_j^2K^2}{2}\right)}$, we use the condition: $a \in \left(\frac{1}{2},1\right]$ and apply simple calculus to have the following bound:
\begin{align*}
    \exp{\left(\sum_{j=1}^n \frac{\alpha_j^2K^2}{2}\right)} &= \exp{\left(\frac{c^2K^2}{2}\left(1+\sum_{j=2}^n \frac{1}{j^{2a}}\right)\right)} \leq \exp{\left(\frac{c^2K^2}{2}\left(1+\int_{j=1}^n \frac{1}{j^{2a}}dj\right)\right)}  \\
    &\leq \exp{\left(\frac{c^2K^2}{2}\left(1+\frac{1}{2a-1}\right)\right)} \leq \exp{\left(\frac{c^2K^2}{2a-1}\right)}.
\end{align*}
In a similar manner, for the other term: $\exp{\left(\sum_{j=k+1}^n - \alpha_j \mu\right)}$, we have
\begin{align*}
    &\exp{\left(\sum_{j=k+1}^n - \alpha_j \mu\right)} = \exp{\left( \mu c \sum_{j=k+1}^n \frac{-1}{j^a}\right)} \leq \exp{\left( \mu c \sum_{j=k+1}^n \frac{-1}{j}\right)} \\
    &\leq \exp{\left( \mu c \int_{j=k+1}^{n+1} \frac{-1}{j}dj\right)} = \exp{\left( \mu c \left[-\log(x)\right]^{n+1}_{k+1}\right)} = \left(\frac{k+1}{n+1}\right)^{\mu c} \leq 2^{\mu c}\left(\frac{k}{n+1}\right)^{\mu c}.
\end{align*}
Substituting the bounds for the above two terms back in (\ref{eq:local-reference-2}), we have
\begin{equation*}
    \norm{z_n}_2  \leq \exp{\left(\frac{c^2K^2}{2a-1}\right)} \left[\frac{\norm{z_{0}}_2}{(n+1)^{\mu c}} + \frac{2^{\mu c}c}{(n+1)^{\mu c}} \sum_{k=1}^n k^{\mu c - a} \norm{\zeta_k}_2\right].
\end{equation*}
Taking expectation on both sides, we have
\begin{align*}
    \Exp\left[\norm{z_n}_2\right] 
    \leq \exp{\left(\frac{c^2K^2}{2a-1}\right)} \left[\frac{\Exp\left[\norm{z_{0}}_2\right]}{(n+1)^{\mu c}} + \frac{2^{\mu c}c{D}_1}{(n+1)^{\mu c}} \sum_{k=1}^n k^{\mu c - a - 0.5}\right],
\end{align*}
% where the above inequality follows from \cref{eq:sg-independence-lemma} after applying the law of total expectation. With $m_k=k$, we have
% \begin{align*}
%     \Exp\left[\norm{z_n}_2\right] 
%     &\leq \exp{\left(\frac{c^2S^2}{2a-1}\right)} \left[\frac{\Exp\left[\norm{z_{0}}_2\right]}{(n+1)^{\mu c}} + \frac{2^{\mu c}cC_1}{(n+1)^{\mu c}} \sum_{k=1}^n k^{\mu c - a - e_1}\right].
% \end{align*}
The theorem condition : $\mu c-a-0.5 > -1$ implies that the summation above is bounded by a finite integral given below.
\begin{align*}
    \sum_{k=1}^n k^{\mu c - a - 0.5} \leq 
        \int_{k=1}^{n+1} k^{\mu c - a - 0.5}dk \leq \frac{(n+1)^{1+\mu c - a - 0.5}}{0.5+\mu c - a}.
\end{align*}
Then,
\begin{equation}\label{eq:sg-theorem-mae-bound}
    \Exp\left[\norm{z_n}_2\right] \leq \exp{\left(\frac{c^2K^2}{2a-1}\right)} \left[\frac{\Exp\left[\norm{z_{0}}_2\right]}{(n+1)^{\mu c}} + \frac{2^{\mu c}c{D}_1}{\left(0.5+\mu c - a\right)(n+1)^{a - 0.5}}\right].
\end{equation}
This concludes the MAE bound on the last iterate of SG algorithm. 
\paragraph{MSE bound proof:} We now provide the bound on parameter error on the last iterate: $z_n \triangleq \theta_n - \theta^*$.
\begin{align*}
    \norm{z_n}_2^2 &\leq \norm{z_{n-1} - \alpha_n \nabla \srth{n-1} - \alpha_n \zeta_n }_2^2 \\
    &= \norm{z_{n-1} - \alpha_n \nabla \srth{n-1}}_2^2 - 2 \alpha_n \langle z_{n-1} - \alpha_n \nabla \srth{n-1}, \zeta_n\rangle + \alpha_n^2 \norm{\zeta_n}_2^2 \\
    &\leq (1 - 2\alpha_n \mu + \alpha_n^2K^2)\norm{z_{n-1}}_2^2 + 2\alpha_n \norm{z_{n-1} - \alpha_n \nabla \srth{n-1}}_2 \norm{\zeta_n}_2 + \alpha_n^2 \norm{\zeta_n}_2^2 \\
    &\leq (1 - 2\alpha_n \mu + \alpha_n^2K^2)\norm{z_{n-1}}_2^2 + 2\alpha_n \sqrt{1 - 2\alpha_n \mu + \alpha_n^2K^2} \norm{z_{n-1}}_2 \norm{\zeta_n}_2 + \alpha_n^2 \norm{\zeta_n}_2^2 \\
    &\leq (1 - 2\alpha_n \mu + \alpha_n^2K^2)\norm{z_{n-1}}_2^2 + 2\alpha_n  \left(1 + c K\right) \norm{z_{n-1}}_2 \norm{\zeta_n}_2 + \alpha_n^2 \norm{\zeta_n}_2^2,
\end{align*}
where the second and third inequalities follow from (\ref{eq:local-reference-1b}), while the last inequality follows from: $\sqrt{1 - 2\alpha_n \mu + \alpha_n^2K^2} < 1+\alpha_n K \leq 1 + \alpha_1 K = 1+cK$. Unrolling the above equation, we have
\begin{align*}
    \norm{z_n}_2^2 &\leq \norm{z_0}_2^2 \prod_{k=1}^n \left( 1 - 2\alpha_k \mu + \alpha_k^2K^2\right) \\
    &+ \sum_{k=1}^n \left[\left(2 \left(1+c K\right)\alpha_k \norm{z_{k-1}}_2 \norm{\zeta_k}_2 + \alpha_k^2 \norm{\zeta_k}_2^2 \right) \prod_{j=k+1}^n \left( 1 - 2\alpha_j \mu + \alpha_j^2K^2\right) \right] \\
    &\leq \frac{\norm{z_0}_2^2}{(n+1)^{2\mu c}} + \sum_{k=1}^n \left[\left(2 \left(1+c K\right)\alpha_k  \norm{z_{k-1}}_2 \norm{\zeta_k}_2 + \alpha_k^2 \norm{\zeta_k}_2^2 \right)\left(\frac{k+1}{n+1}\right)^{2 \mu c} \right].
\end{align*}
The last inequality follows from (\ref{eq:product-1-plus-x}). Taking expectations on both sides, we have
\begin{align}
    \nonumber 
    \Exp&\left[\norm{z_n}_2^2\right] \\
    \label{eq:local-reference-3}
    &\leq \frac{\Exp\left[\norm{z_0}_2^2\right] + 2^{2 \mu c}\sum_{k=1}^n \left[\left(2\left(1+c K\right) \alpha_k  \Exp\left[ \norm{z_{k-1}}_2 \norm{\zeta_k}_2\right] + \alpha_k^2 \Exp\left[\norm{\zeta_k}_2^2\right] \right) k^{2 \mu c} \right]}{\left(n+1\right)^{2 \mu c}} 
\end{align}
Next, we have for all $k \in \{1,2,\ldots,n\}$,
\begin{align*}
    \Exp&\left[ \norm{z_{k-1}}_2 \norm{\zeta_k}_2\right] \\
    &= \Exp\left[ \left.\Exp\left[ \norm{z_{k-1}}_2 \norm{\zeta_k}_2 \right| \mathcal{F}_{k-1}\right]\right] =\Exp\left[ \norm{z_{k-1}}_2 \left.\Exp\left[  \norm{\zeta_k}_2\right| \mathcal{F}_{k-1}\right]\right] \leq \frac{{D}_1}{\sqrt{k}}\Exp\left[ \norm{z_{k-1}}_2\right].
\end{align*}
The first equality is the law of total expectation, while the second equality follows because $\theta_{k-1}$ is $\mathcal{F}_{k-1}$-measurable. The last inequality follows from the first bound in \cref{eq:sg-independence-lemma} with $m_1^k=m_2^k=k$. Substituting this back into \cref{eq:local-reference-3} and applying the second bound from \cref{eq:sg-independence-lemma}, we have
\begin{align*}
    \Exp&\left[\norm{z_n}_2^2\right] \\
    &\leq \frac{\Exp\left[\norm{z_0}_2^2\right]}{(n+1)^{2\mu c}} + \frac{2^{2 \mu c}}{\left(n+1\right)^{2 \mu c}}\sum_{k=1}^n \left[ \frac{2\left(1+c K\right){D}_1}{k^{0.5}} \alpha_k k^{2 \mu c} \Exp\left[ \norm{z_{k-1}}_2\right]  + \alpha_k^2 \frac{{D}_2}{k}k^{2 \mu c} \right].
\end{align*}
Substituting $\alpha_k=c/{k^a}$, we have
\begin{align}
    \nonumber
    \Exp\left[\norm{z_n}_2^2\right] &\leq \frac{\Exp\left[\norm{z_0}_2^2\right]}{(n+1)^{2\mu c}}+ \frac{2^{2 \mu c }c^2{D}_2}{\left(n+1\right)^{2 \mu c}}\sum_{k=1}^n k^{2 \mu c - 2a - 1} \\
    \label{eq:local-reference-5}
    &+ \frac{2^{2 \mu c+1}\left(1+c K\right)c{D}_1}{\left(n+1\right)^{2 \mu c}}\sum_{k=1}^n k^{2\mu c-a-0.5}\Exp\left[ \norm{z_{k-1}}_2\right].
\end{align}
For bounding the last summation in \cref{eq:local-reference-5}, we make the following claim: for every $k \geq 1$, 
\begin{align*}
    \Exp\left[ \norm{z_{k-1}}_2\right] \leq \exp{\left(\frac{c^2K^2}{2a-1}\right)} \left[\frac{\Exp\left[\norm{z_{0}}_2\right]}{k^{\mu c}} + \frac{2^{\mu c}c{D}_1}{\left(1+\mu c - a - 0.5\right)k^{a -0.5}}\right].
\end{align*}
For $k>1$, the claim follows from the MAE bound in \cref{eq:sg-theorem-mae-bound}, while for the case $k=1$ it holds trivially. Substituting the above inequality back in \cref{eq:local-reference-5}, we have
\begin{align*}
    &\Exp\left[\norm{z_n}_2^2\right] \leq \frac{\Exp\left[\norm{z_0}_2^2\right]}{(n+1)^{2\mu c}} + \frac{2^{2 \mu c }c^2{D}_2}{\left(n+1\right)^{2 \mu c}}\sum_{k=1}^n k^{2 \mu c - 2a - 1} \\
    &+ \frac{2^{2 \mu c+1}\left(1+c K\right)c{D}_1}{\left(n+1\right)^{2 \mu c}}\exp{\left(\frac{c^2K^2}{2a-1}\right)} \sum_{k=1}^n \left(k^{\mu c-a-e_1}\Exp\left[ \norm{z_{0}}_2\right] + \frac{2^{\mu c}c{D}_1k^{2\mu c -2a}}{\left(0.5+\mu c - a\right)}\right) \\
    &\leq \frac{\Exp\left[\norm{z_0}_2^2\right]}{(n+1)^{2\mu c}} + \frac{2^{2 \mu c }c^2{D}_2}{\left(2 \mu c - 2a\right)\left(n+1\right)^{2a}} + \exp{\left(\frac{c^2K^2}{2a-1}\right)} \\
    & \times \left[\left(\frac{2^{2 \mu c+1}\left(1+c K\right)c{D}_1}{0.5+\mu c - a}\right) \frac{\Exp\left[ \norm{z_{0}}_2\right]}{\left(n+1\right)^{\mu c+a-0.5}} + \left(\frac{2^{3 \mu c}\left(1+c K\right)c^2{D}_1^2}{\left(0.5+\mu c - a\right)^2\left(n+1\right)^{2a - 1}}\right)\right].
\end{align*}
The last inequality follows by bounding each of the three summations with a finite integral, by applying the inequalities: $2\mu c - 2a > 0$, $\mu c - a > -0.5$, and $2\mu c - 2a  > -1$ respectively. Each of these inequalities follow from the theorem condition : $\mu c - a > -0.5$. This completes the proof for the bound on convergence in parameter, i.e., $\theta_n \to \theta^*$. 

For the bound on convergence in value ($\srth{n} \to \srth{*}$), we use the assumption that $\srt$ is $K$-smooth. Therefore, $\srt$ satisfies : $\srt(x) - \srt(y) \leq \nabla \srt(y)^\mathrm{T}(x-y) + \frac{K}{2}\norm{x-y}_2^2$ for every $x,y$ in $\BB$. Using the theorem condition $\nabla \srt(\theta^*)=0$, we substitute $x = \theta_n$ and $y=\theta^*$ in above inequality, the claim follows.
\end{myproof}

\subsubsection{The convex case.}
We now present a result that provides bounds for  \Cref{alg:ubsr-minimization} under the assumption that $\srt$ is smooth and convex.
\begin{theorem}\label{theorem:sgd_convergence_convex}
    Suppose $\srt$ is convex, $K$-smooth, and the assumptions of \Cref{lemma:ubsr-grad-est-bound} are satisfied. Let $\theta_*$ denote the minimizer of $\srth{}$ defined in \cref{eq:sr-theta} and we assume that $\nabla \srth{*} = 0$. Let $\left\{\theta_k\right\}_{k=1}^n$ be the iterates obtained by running \Cref{alg:ubsr-minimization} with $\alpha_k = \frac{1}{K\sqrt{k}}$, and $m_1^k=m_2^k=k,\,\forall k$. Define $\overline{\theta}_n\triangleq \sum_{k=1}^n\frac{\theta_k}{n}$. Then, for all $n\geq 1$, we have
    \begin{align*}
        &\Exp \left[ \srt(\overline{\theta}_n) - \srth{*}\right] \leq \frac{D}{\sqrt{n}},
    \end{align*}
    where $D=K\Exp\left[\norm{z_{0}}^2\right]+\frac{D_2 \pi^2}{6K}+2\Exp\left[\norm{z_0}\right]D_1(1+\log{n})+\frac{2D_1^2 \log^2{n}}{K}$. Here $D_1,D_2$ are as given in \Cref{lemma:ubsr-grad-est-bound}. 
\end{theorem}
\begin{myproof}
    Note that $\srt$ is $K$-smooth, and $\Theta$ is convex. Then by Theorem 2.1.5 of \cite{nesterov_introductory_2004}, for every $\theta_2,\theta_2 \in \Theta$, we have
    \begin{align}\label{eq:grad-product-bound}
        \left<\nabla \srth{2} - \nabla \srth{1},\theta_2-\theta_1\right> \geq \frac{1}{K}\norm{\nabla \srth{2} - \nabla \srth{1}}^2.
    \end{align} 
    Since $\srt(\cdot)$ is continuously differentiable and convex, by Definition 2.1.2 of \cite{nesterov_introductory_2004}, we have
    \begin{equation}\label{eq:sr-convexity}
        \srth{2} \geq \srth{1} + \nabla \srth{1}^T (\srth{2}-\srth{1}), \;\; \forall \theta_1,\theta_2 \in \B.
    \end{equation}
    Putting $\theta_2=\theta_*$ and $\theta_1=\theta_{n-1}$ in \cref{eq:grad-product-bound}, we have
    $-\left\langle z_{n-1}, \nabla \srth{n-1} \right\rangle \leq \frac{-1}{K}\norm{\nabla \srth{n-1}}^2$. Likewise, putting $\theta_2=\theta_*$ and $\theta_1=\theta_{n-1}$ in \cref{eq:sr-convexity}, we have $-\left\langle z_{n-1}, \nabla \srth{n-1} \right\rangle \leq \srth{*} - \srth{n-1}$. Combining both, we have the following inequality for all $n\geq 1$.
    \begin{equation}\label{eq:cvx-smth-inequality}
        - 2 \alpha_n \left\langle z_{n-1}, \nabla \srth{n-1} \right\rangle \leq \frac{-\alpha_n}{K}\norm{\nabla \srth{n-1}}^2 - \alpha_n \left(\srth{n-1} - \srth{*}\right).
    \end{equation}
    Then, we have
    \begin{align*}
        &\norm{z_{n-1} - \alpha_n \nabla \srth{n-1}}^2 = \norm{z_{n-1}}^2 + \alpha_n^2 \norm{\nabla \srth{n-1}}^2 - 2 \alpha_n \left\langle z_{n-1}, \nabla \srth{n-1} \right\rangle\\
        \numberthis \label{eq:smth-parameter-int-0}
        &\leq \norm{z_{n-1}}^2 + \left(\alpha_n^2 - \frac{\alpha_n}{K} \right) \norm{\nabla \srth{n-1}}^2 - \alpha_n \left(\srth{n-1} - \srth{*}\right) \leq \norm{z_{n-1}}^2,
    \end{align*}
    where the first inequality follows from \cref{eq:cvx-smth-inequality} and the second inequality follows from the step size assumption: $\alpha_n \leq 1/K$.
     Next, we have
    \begin{align*}
        \norm{z_n}^2 &= \norm{z_{n-1} - \alpha_n J^{m_1^{k},m_2^{k}}(\mathbf{\hat{Z}}_{\theta_{k-1}},\mathbf{Z}_{\theta_{k-1}},\mathbf{V}_{\theta_{k-1}})}^2 = \norm{\left(z_{n-1} - \alpha_n \nabla \srth{n-1}\right) - \alpha_n \zeta_n}^2 \\
        &= \norm{z_{n-1} - \alpha_n \nabla \srth{n-1}}^2 - 2 \alpha_n \left\langle z_{n-1} - \alpha_n \nabla \srth{n-1}, \zeta_n \right\rangle + \alpha_n^2 \norm{\zeta_n^2} \\
        &\leq \numberthis \label{eq:smth-parameter-int-1} \norm{z_{n-1}}^2 - \alpha_n \left(\srth{n-1} - \srth{*}\right) + 2 \alpha_n \norm{z_{n-1}}\norm{\zeta_n} + \alpha_n^2 \norm{\zeta_n^2},
    \end{align*}
    where the first inequality above follows from the first inequality in \cref{eq:smth-parameter-int-0} while the second inequality above follows from the Cauchy-Schwartz inequality and the final inequality in \cref{eq:smth-parameter-int-0}.The inequality in \cref{eq:smth-parameter-int-1} also implies that $\norm{z_n} \leq \norm{z_{n-1}}+\alpha_n \norm{\zeta_n}$. Unrolling this, we have the following inequality.
    \begin{equation}\label{eq:z-unrolled}
        \norm{z_n} \leq \norm{z_0}+\sum_{k=1}^n \alpha_k \norm{\zeta_k}
    \end{equation} 
    
    Taking expectation on both sides of \cref{eq:smth-parameter-int-1}, and rearranging the terms, we have
    \begin{align*}
        &\alpha_n \Exp\left[\srth{n-1} - \srth{*}\right] \leq \Exp\left[\norm{z_{n-1}}^2\right] - \Exp\left[\norm{z_n}^2\right] + \alpha_n^2 \Exp\left[\norm{\zeta_n^2}\right] + 2\alpha_n \Exp\left[\norm{z_{n-1}}\norm{\zeta_n}\right]\\
        &\leq \Exp\left[\norm{z_{n-1}}^2\right] - \Exp\left[\norm{z_n}^2\right] + \alpha_n^2 \frac{{D}_2}{n} + 2\alpha_n \frac{{D}_1}{\sqrt{n}}\Exp\left[\norm{z_{n-1}}\right]\\
        &\numberthis \label{eq:smth-biased-avg-weighted}\leq \Exp\left[\norm{z_{n-1}}^2\right] - \Exp\left[\norm{z_n}^2\right] + \alpha_n^2 \frac{{D}_2}{n} + 2\alpha_n \frac{{D}_1}{\sqrt{n}}\left[\Exp\left[\norm{z_{0}}\right] + \sum_{j=1}^{n-1} \alpha_j \frac{{D}_1}{\sqrt{j}}\right],
    \end{align*}
    where the last inequality follows from \cref{eq:z-unrolled}. The bounds on $\zeta_n$ follow from the gradient estimation bounds in \Cref{prop:grad-est-bound-numerical-I}. Interested readers are referred to the proof of \Cref{theorem:sgd_convergence} for further technical details. Taking summation from $k=1$ to $n$, we have 
    \begin{align*}    
    &\sum_{k=1}^n \alpha_k\Exp\left[\srth{k-1} - \srth{*}\right] \\
    &\leq \Exp\left[\norm{z_{0}}^2\right] + \sum_{k=1}^n \frac{\alpha_k^2 {D}_2}{k} + 2\Exp\left[\norm{z_0}\right]\sum_{k=1}^n \frac{\alpha_k {D}_1}{\sqrt{k}} + 2 {D}_1^2 \sum_{k=1}^n \frac{\alpha_k}{\sqrt{k}} \sum_{j=1}^{k-1}  \frac{\alpha_j}{\sqrt{j}}.
    \end{align*}
    Then, we have 
    \begin{align*}
    &\sum_{k=1}^n \alpha_k\Exp\left[\srth{k-1} - \srth{*}\right] \geq \alpha_n \sum_{k=1}^n \Exp\left[\srth{k-1} - \srth{*}\right] \\
    &= n \alpha_n \sum_{k=1}^n \frac{1}{n}\Exp\left[\srth{k-1} - \srth{*}\right] \geq \frac{\sqrt{n}}{K} \Exp\left[\srt(\overline{\theta_{n-1}})-\srth{*}\right]
    \end{align*}
    ,where the last inequality follows from the convexity of $\srt$. Then, we have
    \begin{align*}
    &\frac{\sqrt{n}}{K} \Exp\left[\srt(\overline{\theta_{n-1}})-\srth{*}\right] \\
    &\leq \Exp\left[\norm{z_{0}}^2\right] + \sum_{k=1}^n \frac{\alpha_k^2 {D}_2}{k} + 2\Exp\left[\norm{z_0}\right]\sum_{k=1}^n \frac{\alpha_k {D}_1}{\sqrt{k}} + 2 {D}_1^2 \sum_{k=1}^n \frac{\alpha_k}{\sqrt{k}} \sum_{j=1}^{k-1}  \frac{\alpha_j}{\sqrt{j}}.
    \end{align*}
    Putting $\alpha_k = \frac{1}{K\sqrt{k}}$ in the above inequality, we have
    \begin{align*}
    &\Exp\left[\srt(\overline{\theta_{n-1}})-\srth{*}\right] \\
    &\leq \frac{K\Exp\left[\norm{z_{0}}^2\right]+\frac{2 {D}_1^2}{K}}{\sqrt{n}} + \frac{ {D}_2 \pi^2}{6K\sqrt{n}} + \frac{2{D}_1\Exp\left[\norm{z_0}\right](1+\log{n})}{\sqrt{n}} + \frac{{D}_1^2(\log{n})^2}{K\sqrt{n}},
    \end{align*}
    where we use the following summation inequalities. 
    \begin{align*}
        \sum_{k=1}^n \frac{1}{\sqrt{k}} \geq \sum_{k=1}^n \frac{1}{\sqrt{n}} = \sqrt{n}, \; \sum_{k=1}^n \frac{1}{k^2} &\leq \frac{\pi^2}{6}, \; \sum_{k=1}^n \frac{1}{k} \leq 1+ \log{n}, \; \sum_{k=1}^n \frac{\log{k}}{k} \leq 1 + \frac{(\log{n})^2}{2}.
    \end{align*}
\end{myproof}

The $\log{n}$ and $\log^2{n}$ terms can be avoided by replacing the step size choice $\frac{1}{K\sqrt{k}}$ with $\frac{1}{K\sqrt{n}}$, and we avoid a separate proof.  

\subsubsection{The non-convex case.}
Next, we present a result that provides bounds on the norm of the gradient of $\srt$ at a random iterate $\theta_R$ of the \Cref{alg:ubsr-minimization}, distributed uniformly over the $n$ iterates. Precisely, let $\left\{\theta_k\right\}_{k=1}^n$ be the iterates given by \Cref{alg:ubsr-minimization}. Then $\theta_R$ takes value $\theta_k$ w.p. $1/n$, for every $k \in \{0,1,\ldots,n-1\}$. 
\begin{theorem}\label{theorem:sgd_convergence_non_convex}
 Suppose $\srt$ is $K$-smooth and the assumptions of \Cref{lemma:ubsr-grad-est-bound} are satisfied. Let $\theta_*$ denote the minimizer of $\srth{}$ defined in \cref{eq:sr-theta} and we assume that $\nabla \srth{*} = 0$. Let $\left\{\theta_k\right\}_{k=1}^n$ be the iterates obtained by running \Cref{alg:ubsr-minimization} with $\alpha_k = \frac{1}{K\sqrt{k}}$, and $m_1^k=m_2^k=k,\,\forall k$. Then, we have
\begin{equation*}
    \Exp\left[\norm{ \nabla \srth{R}}^2_2\right] \leq \frac{KD}{\sqrt{n}},
\end{equation*}
where  $K$ and $D$ are as given in \Cref{lemma:h-smooth} and \Cref{theorem:sgd_convergence_convex} respectively. 
\end{theorem}
\begin{myproof}
    Note that $\srt$ is $K$-smooth and $\Theta$ is convex. Then from the proof of \Cref{theorem:sgd_convergence_convex}, we know that
    $-\left\langle z_{n-1}, \nabla \srth{n-1} \right\rangle \leq \frac{-1}{K}\norm{\nabla \srth{n-1}}^2$. Therefore, we have
    \begin{align*}
        \norm{z_n}^2 &= \norm{z_{n-1} - \alpha_n \hat{J}^m_\theta(\mathbf{Z,\hat{Z}})}^2 = \norm{\left(z_{n-1} - \alpha_n \nabla \srth{n-1}\right) - \alpha_n \zeta_n}^2 \\
        &= \norm{z_{n-1} - \alpha_n \nabla \srth{n-1}}^2 - 2 \alpha_n \left\langle z_{n-1} - \alpha_n \nabla \srth{n-1}, \zeta_n \right\rangle + \alpha_n^2 \norm{\zeta_n^2} \\
        &\leq  \norm{z_{n-1}}^2 - \left(\frac{2\alpha_n}{K} - \alpha_n^2 \right) \norm{\nabla \srth{n-1}}^2  + 2 \alpha_n \norm{z_{n-1}}\norm{\zeta_n} + \alpha_n^2 \norm{\zeta_n^2} \\
        &\leq \norm{z_{n-1}}^2 - \frac{\alpha_n}{K} \norm{\nabla \srth{n-1}}^2  + 2 \alpha_n \norm{z_{n-1}}\norm{\zeta_n} + \alpha_n^2 \norm{\zeta_n^2}.
    \end{align*}
    The last inequality follows from the step size assumption: $\alpha_n \leq 1/K$. Rearranging and applying \cref{eq:z-unrolled}, we have
    \begin{align*}
        \frac{\alpha_n}{K}\norm{\nabla \srth{n-1}}^2 &\leq \norm{z_{n-1}}^2 - \norm{z_n}^2  + 2 \alpha_n \norm{\zeta_n} \left(\norm{z_0} + \sum_{j=1}^{n-1} \alpha_j \norm{\zeta_j}\right) + \alpha_n^2 \norm{\zeta_n^2}.
    \end{align*}
    Taking summation from $k=1$ to $n$, we have
    \begin{align*}
        &\frac{1}{K}\sum_{k=1}^n \alpha_k \norm{\nabla \srth{k-1}}^2 \leq \norm{z_{0}}^2 + 2 \sum_{k=1}^n \alpha_k \norm{\zeta_k} \left(\norm{z_0} + \sum_{j=1}^{k-1} \alpha_j \norm{\zeta_j}\right) + \sum_{k=1}^n \alpha_k \norm{\zeta_k^2} \\
& \frac{n \alpha_n}{K}\sum_{k=1}^n \frac{1}{n} \norm{\nabla \srth{k-1}}^2 \leq \norm{z_{0}}^2 + 2 \sum_{k=1}^n \alpha_k \norm{\zeta_k} \left(\norm{z_0} + \sum_{j=1}^{k-1} \alpha_j \norm{\zeta_j}\right) + \sum_{k=1}^n \alpha_k \norm{\zeta_k^2} \\
& \frac{n \alpha_n}{K} \Exp_R\left[\norm{\nabla \srth{R}}^2\right] \leq \norm{z_{0}}^2 + 2 \sum_{k=1}^n \alpha_k \norm{\zeta_k} \left(\norm{z_0} + \sum_{j=1}^{k-1} \alpha_j \norm{\zeta_j}\right) + \sum_{k=1}^n \alpha_k \norm{\zeta_k^2}.
    \end{align*}
    Taking expectation on both sides, we get the same bound as in \Cref{theorem:sgd_convergence_convex} with an added factor of $K$. The remainder of the proof is identical to the proof of \Cref{theorem:sgd_convergence_convex}, and we avoid a separate proof.
\end{myproof}
A few remarks are in order. 

From \Cref{theorem:sgd_convergence}, we infer that for the strongly convex case, the iteration complexity $N$, i.e., the number of iterations of the SG algorithm required to attain the bound $\Exp\left[\srth{N} - \srth{*}\right] \leq \epsilon$, is of the order \order{1/\epsilon}. With the batch sizes $m_1^k=m_2^k=k$, the total number of samples required after $N$ iterations is $\frac{N(N+1)}{2} \sim$ \order{N^2}, which translates to the sample complexity \order{1/\epsilon^2} in terms of $\epsilon$. Similar analysis leads to a sample complexity of \order{1/\epsilon^4} for each, the general convex case and the non-convex case. 

Our results allow for an increasing batch size without requiring knowledge of $N$ beforehand, making \Cref{alg:ubsr-minimization} an any-time algorithm.
%\todos{Say what sample complexity means and justify $1/\epsilon^2$ sample complexity bound.}

%An asymptotic convergence rate of $\mathcal{O}(1/n)$ was derived earlier by \cite{zhaolin2016ubsrest} for the scalar UBSR optimization case, but their result required a batch size $m\geq n$ for each iteration, thus requiring $n$ to be known beforehand. Our result not only establishes a non-asymptotic bound of the same order, but also allows for an increasing batch size that does require knowing $n$ beforehand, making \Cref{alg:ubsr-minimization} an any-time algorithm.

%A result similar to \Cref{theorem:sgd_convergence} can be obtained for the constant batch size, i.e., $m_1^k=m_2^k=m, \,\forall k$. The proof follows by completely parallel arguments to those employed in the proof of \Cref{theorem:sgd_convergence}, and we omit the details. The error bounds of the order \order{1/m} are obtained in this constant batch size case. %The second column of \Cref{table:complexity-analysis} covers this case for $m=n^p$.

\subsection{UBSR minimization in the offline setting}
\label{section:offline}
In this section, we describe how the UBSR gradient estimator given in \cref{eq:J-definition}, and \Cref{alg:ubsr-minimization} can be modified to cater to an offline setting, where a fixed set (dataset) of samples from $\xi$ are available and $F$ is presumed to be known in closed-form. Suppose we obtain $M$ i.i.d. samples of $\xi$. 
For notational convenience, we split these into two sets, namely $\{\hat{\xi}_1,\hat\xi_2,\ldots,\hat\xi_{M_1}\}$  and $\{{\xi}_1,\xi_2,\ldots,\xi_{M_2}\}$, where $M_1,M_2 \in \mathcal{N}$ such that $M_1+M_2=M$.  Then, our proposed gradient estimator $Q^{M_1,M_2}:\Theta \to \Rel^d$ is defined by
\begin{align}\label{eq:Q-definition}
    Q^{M_1,M_2}\left(\theta\right) \triangleq J^{M_1,M_2}\left(\left(F(\theta,\hat{\xi}_j\right)_{j=1}^{M_1},\left(F(\theta,{\xi}_j\right)_{j=1}^{M_2},\left(\nabla F(\theta,{\xi}_j\right)_{j=1}^{M_2}\right),
\end{align}
where $J^{M_1,M_2}$ is defined in \cref{eq:J-definition}. 

\begin{figure}[!t]
\centering
\begin{algorithm}[H]
\caption{UBSR-SG-OFFLINE}\label{alg:ubsr-minimization-offline}
\begin{algorithmic}[1] 
    \REQUIRE $\theta_0 \in \Theta$, samples $\{\hat{\xi}_k\}_{k = 1}^{M_1} \sim \xi, \{\xi_k\}_{k=1}^{M_2} \sim \xi$,  step sizes $\{\alpha_k\}_{k \ge 1}$.
    \FOR{$k= 1, 2, \ldots, n$}
        \STATE $\mathrm{\hat{Z}}^k_{\theta_{k-1}} \leftarrow \left[F(\theta_{k-1},\hat{\xi}_1),F(\theta_{k-1},\hat{\xi}_2),\ldots,F(\theta_{k-1},\hat{\xi}_{M_1}) \right]$.
        \STATE compute $\srt^{M_1}(\mathrm{\hat{Z}}^k_{\theta_{k-1}})$ using \cref{eq:sr-m-definition}.
        \STATE $\langle \mathrm{{Z}}^k_{\theta_{k-1}}, \!\mathrm{{V}}^k_{\theta_{k-1}}\rangle\! \leftarrow\! \left[\left(F\!\left(\theta_{k-1},{\xi}_1\right),\!\nabla F\!\left(\theta_{k-1},{\xi}_1\right)\right),\dots,\left(F\!\left(\theta_{k-1},{\xi}_1\right),\!\nabla F\!\left(\theta_{k-1},{\xi}_{M_2}\right)\right) \right]$.
        \STATE compute $\hat{Q}^k = J^{M_1,M_2}(\mathrm{\hat{Z}}^k_{\theta_{k-1}},\mathrm{Z}^k_{\theta_{k-1}}, \mathrm{V}^k_{\theta_{k-1}})$ using \cref{eq:J-definition} and $\srt^{m_1^k}$.
        \STATE update $\theta_k \leftarrow \Pi_\Theta\left(\theta_{k-1} - \alpha_k \hat{Q}_{k}\right)$.
    \ENDFOR
    \RETURN $\theta_n$
\end{algorithmic}
\end{algorithm}
\end{figure}

We now present error bounds on the gradient estimator $Q^{M_1,M_2}\left(\theta\right)$.

\begin{lemmma}\label{lemma:ubsr-grad-est-bound-offline}
    Suppose \cref{assumption:l-b1-increasing,as:variance-bound-only-l,assumption:F-bounded-higher-moment,as:variance-bound-l-F,assumption:loss-fn-smooth} and the assumptions of \Cref{theorem:ubsr-gradient} hold. Then, for every $\theta \in \Theta$ and every $M_1,M_2 \in \N$ satisfying $M_1+M_2=M$, the gradient estimator $Q^{M_1,M_2}$ defined in \cref{eq:Q-definition} satisfies 
    \begin{align*}
        \Exp \left[ \norm{Q^{M_1,M_2}\left(\theta\right) - \nabla \srth{}}_2 \right] &\leq \frac{D_3}{\sqrt{\max\{M_1,M_2\}}}, 
        \textrm{ and } \\
        \Exp \left [\norm{Q^{M_1,M_2}\left(\theta\right) - \nabla \srth{}}_2^2 \right] &\leq \frac{D_4}{\max\{M_1,M_2\}}.
    \end{align*}
    The constants $D_3$ and $D_4$ are independent of $M_1$ and $M_2$, and contain a factor of $\log(d)$, where $d$ is the dimensionality of the parameter space $\Theta$.
\end{lemmma}
The proof for the above lemma follows via application of \Cref{lemma:ubsr-grad-est-bound} by noting that the $M_1$ samples are independent from the $M_2$ samples. The constants $D_3$ and $D_4$ are similar to the constants $D_1$ and $D_2$ given in \Cref{lemma:ubsr-grad-est-bound}, and we avoid a separate derivation. Using the UBSR gradient estimator given in \cref{eq:Q-definition}, we propose \Cref{alg:ubsr-minimization-offline} to solve the offline UBSR minimization problem. With some trivial modifications, it is easy to see that the non-asymptotic convergence results for \Cref{alg:ubsr-minimization} stated in \Cref{theorem:sgd_convergence,theorem:sgd_convergence_convex,theorem:sgd_convergence_non_convex} can be easily extended to cover \Cref{alg:ubsr-minimization-offline}, and we omit separate proofs. 
As an aside, we remark that \Cref{alg:ubsr-minimization-offline} is of interest to risk-sensitive decision-making problems in machine learning and portfolio optimization applications in finance. For instance, in a portfolio optimization application, $\xi$ denotes a random returns vector and the objective is given by $F(\theta,\xi)=\theta^T\xi$.

% \subsection{OCE Risk Optimization}
% \label{sec:oce-opt}
% \input{inputs/oce_opt}

% \section{Proofs}
% \label{sec:proofs}
% \input{inputs/proofs-ubsr}

\section{Simulation Experiments}
\label{sec:experiments}

In this section, we demonstrate the performance of our algorithms for estimating and optimizing the UBSR measure across three experiments. In the first experiment, we estimate Value-at-Risk (VaR) using the UBSR-SB \cref{alg:saa_bisect} algorithm. In the second experiment, we estimate the entropic risk of a r.v. $X$, and solve a risk-sensitive optimization problem with entropic risk as the risk criterion. In the third experiment, we solve a portfolio optimization problem using historical data from three popular equity markets and compare the performance of several instances of the UBSR measure against popular benchmarks such as the Sharpe ratio and the equal-weighted portfolio. These experiments provide empirical support for the theoretical guarantees of the UBSR estimation and optimization algorithms that are proposed in this paper. 

\begin{figure}[ht]
    \centering
    \includegraphics[width=\textwidth]{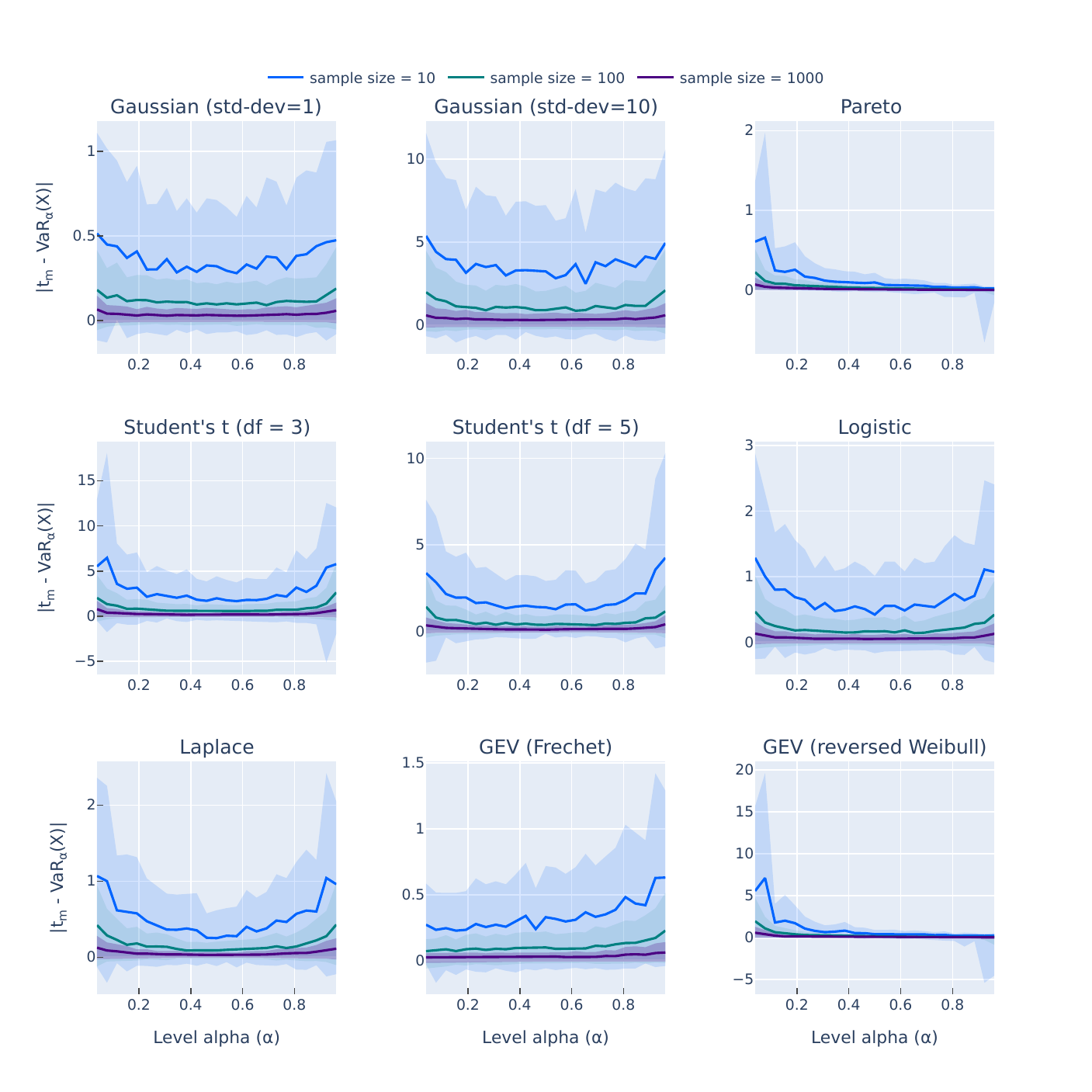}
    \caption{The figure shows the error distribution of the $m$-sample estimate $t_m$, given by algorithm UBSR-SB for different choices of $m$. The algorithm's performance is evaluated across different distributions, and the error statistics for each distribution are plotted separately. Each plot shows the error distribution across 25 uniformly chosen values of $\alpha$ in $(0,1)$ and three sample sizes: $ m=10, 100, 1000$. For each choice of $\alpha$ and $m$, we simulate the experiment $N=1000$ times and plot the error and its spread (standard error) by averaging across the $N$ simulations}
    \label{fig:ubsr-sb-distribution-VaR}
\end{figure}

\subsection{Estimation of Value-at-Risk (VaR)}
Given a r.v. $X$, we employ Algorithm \ref{alg:saa_bisect} for estimating $\textrm{VaR}_\alpha(X)$. We choose UBSR parameters $l(x) = \mathbf{1}_{\{x>0\}}, \forall x \in \Rel$ and $\lambda=\alpha$, and therefore $\ubsr$ coincides with $\textrm{VaR}_\alpha(X)$. In this experiment, we estimate $\textrm{VaR}_\alpha(X)$ for a variety of distributions of $X$, $25$ different values for level $\alpha$, and three choices for sample size $m$. We show separate plots for each choice of distribution of $X$. In each such plot, we show the mean and standard deviation of the errors for each combination of $\alpha$ and $m$. Each mean and std is obtained by averaging across $1000$ runs of the experiment. We plot the estimate of $\textrm{VaR}_\alpha(X)$ obtained by Algorithm \ref{alg:saa_bisect} in \Cref{fig:ubsr-sb-distribution-VaR}. From the figure above, one can conclude that the mean and variance of the estimation error converge to zero as $m$ increases.

\subsection{Estimation and Optimization of Entropic Risk.}
% Entropic risk is a special case of UBSR. 
% We consider the problem of estimation of entropic risk with case where the underlying distribution is Gaussian. 
% %Under a Gaussian assumption, a closed form expression for entropic risk is available. This makes entropic risk an ideal candidate to 
% We test the non-asymptotic performance of the estimator of entropic risk given by Algorithm \ref{alg:saa_bisect} (UBSR-SB).
% %(\ \cite{giesecke-risk-large-losses,dunkel_efficient_2007,dunkel2010stochastic,zhaolin2016ubsrest})
% Further, we also investigate the performance of Algorithm \ref{alg:ubsr-minimization} (UBSR-SG) for the problem of minimization of entropic risk.

 % From the expression for the entropic risk for the Gaussian case, we infer that the problem of entropic risk minimization is equivalent to mean-variance optimization. The minima of this mean-variance optimization problem is readily available, and therefore, non-asymptotic convergence of the iterates given by the UBSR-SG algorithm can be investigated.
 
\paragraph{Entropic risk estimation.}
For our experiment, we assume $X \sim \mathcal{N}(\mu, \sigma^2)$ with $\mu=-1$ and $\sigma^2 = 4$. Under this assumption, the value of the entropic risk measure $\rho_e(X)$ is given by
\begin{equation}\label{eq:entropic-risk}
    \rho_e(X) = \frac{1}{\beta} \log\left(\Exp\left[e^{-\beta X}\right] \right) = -\mu + \frac{\beta \sigma^2}{2},
\end{equation}
where $\beta>0$. We set $\beta=0.5$ in our experiments.
 
% \begin{figure}[ht]
%   %\begin{subfigure}{.5\textwidth}
%   \centering
%     \includegraphics[width=0.7\linewidth]{figures/entropic-mae-mse-ubsr.png}
%     \caption{Performance of UBSR-SB algorithm for estimation of entropic risk of a univariate Gaussian r.v. $X \sim \mathcal{N}(\mu,\sigma^2)$ using $m$ samples.}\label{figure:entropic-risk-ubsr-estimation}
%   %\end{subfigure}%
%   % \begin{subfigure}{.5\textwidth}
%   % \centering
%   %   \includegraphics[width=1\linewidth]{figures/entropic-error-pdf-ubsr.png}
%   %   \caption{Histogram of the estimation error}\label{figure:entropic-risk-ubsr-distribution}
%   % \end{subfigure}
%   % \caption{Performance of UBSR-SB algorithm for estimation of entropic risk of a univariate Gaussian r.v. $X \sim \mathcal{N}(\mu,\sigma^2)$ with $\mu=-1$ and $\sigma^2 = 4$. Given a sample size $m$, Algorithm \ref{alg:saa_bisect} (UBSR-SB) is run with $\delta = 1/\sqrt{m}$ and $m$ i.i.d. samples from $X$ to obtain the estimator $t_m$. For each choice of $m$, we repeat the simulation $N=1000$ times and compute the error mean and its spread (standard error) by averaging across the $N$ simulations.} %As the error vanishes rapidly with $m$, we plot \\$\sqrt{m}$ on the X-axis instead of $m$.
% \end{figure}

\begin{figure}[ht]
  \begin{tabular}{ccc}
  \begin{subfigure}{.5\textwidth}
  \centering
    \centering
    \includegraphics[width=\textwidth]{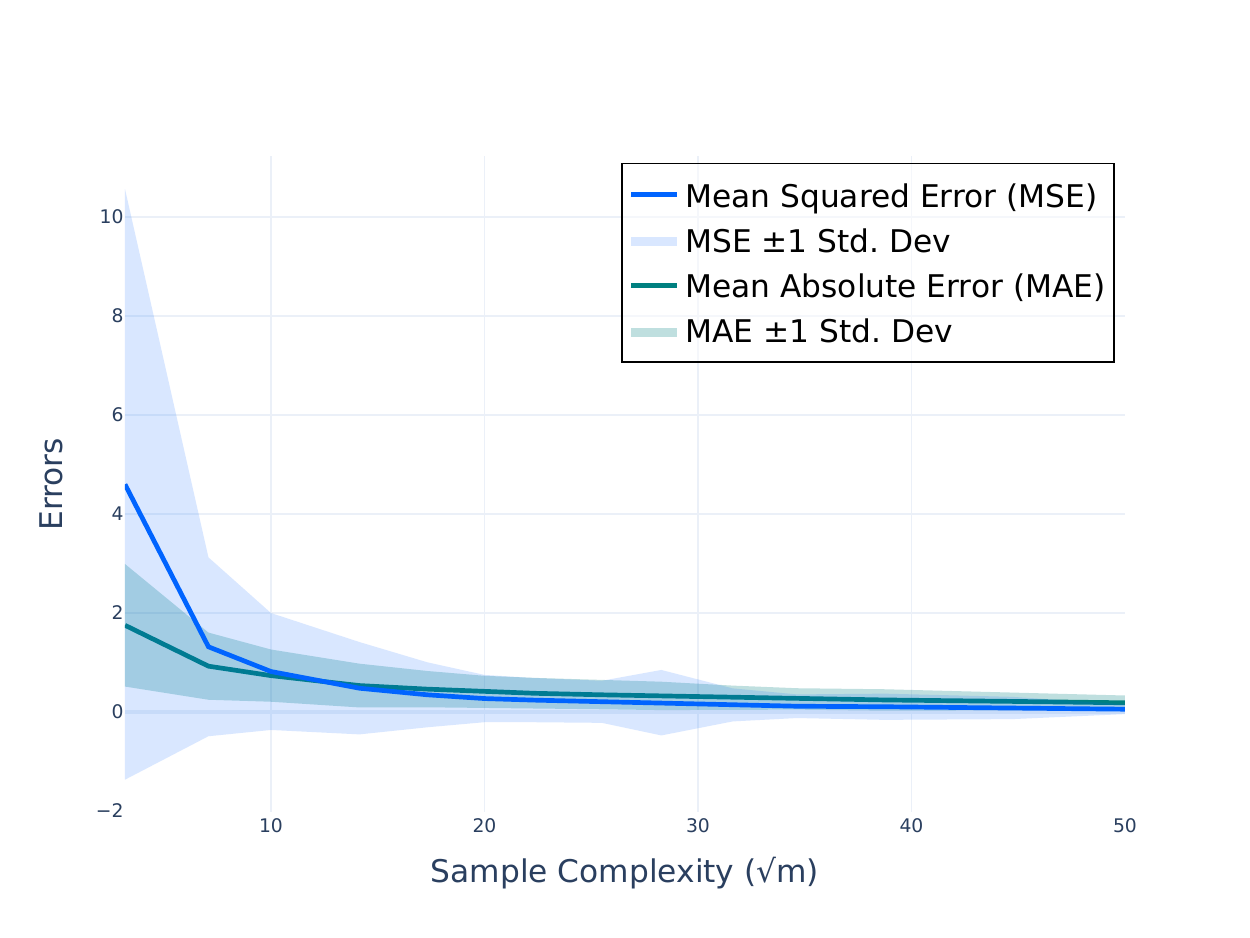}
    \caption{Performance of UBSR-SB algorithm for \\estimation of entropic risk of a univariate \\ Gaussian r.v. $X \sim \mathcal{N}(\mu,\sigma^2)$ using $m$ samples}\label{figure:entropic-risk-ubsr-estimation}
  \end{subfigure}
  \begin{subfigure}{.5\textwidth}
  \centering
    \includegraphics[width=\textwidth]{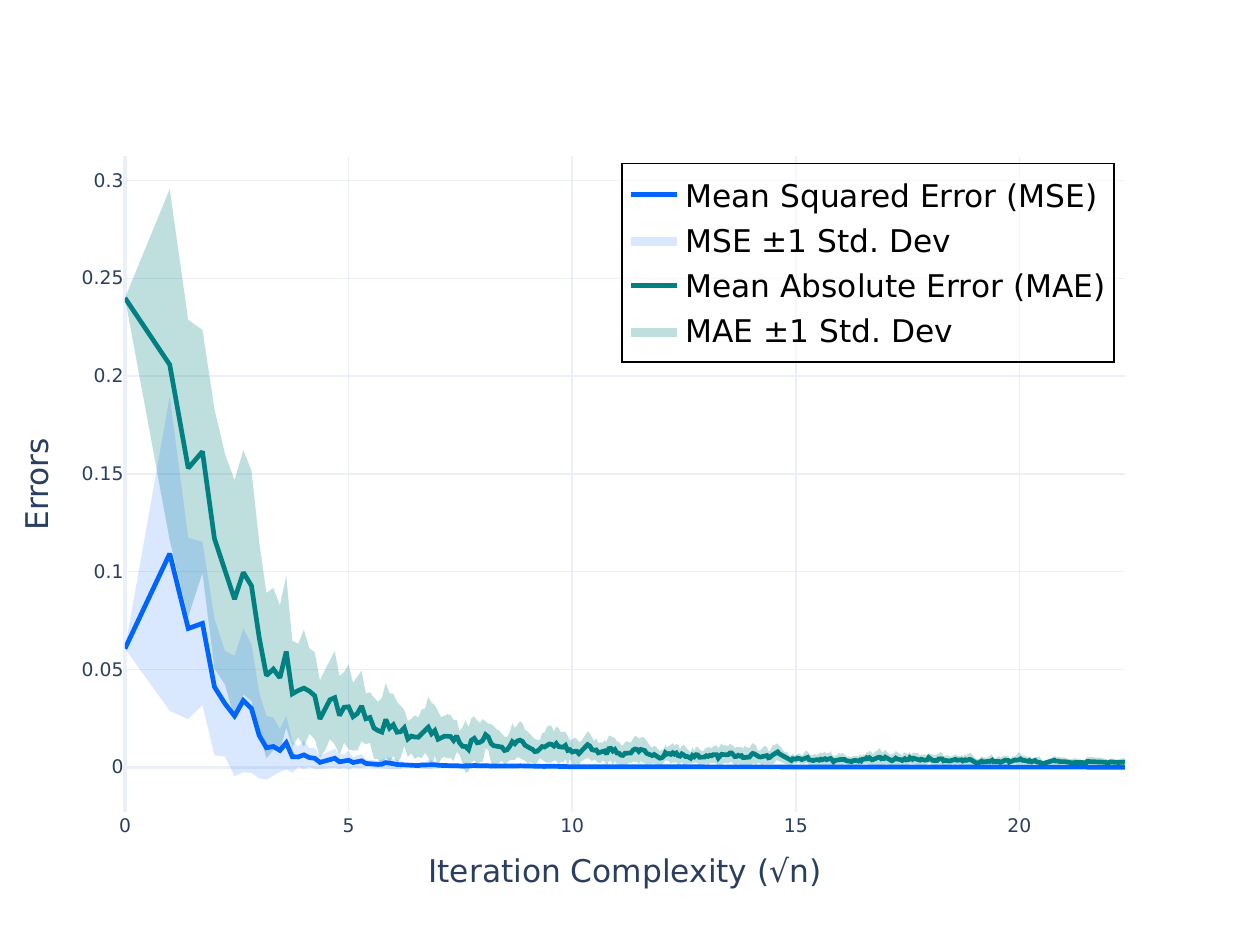}
    \caption{Performance of Algorithm \ref{alg:ubsr-minimization} \\(UBSR-SG) for minimizing entropic \\risk.}\label{figure:entropic-risk-ubsr-optimization}
  \end{subfigure}
  \end{tabular}
  \caption{Performance of our UBSR estimation and optimization algorithms for the case of entropic risk}
\end{figure}

For the choice of $l(x)=e^{\beta x}$ and $\lambda = 1$, $\ubsr$ coincides with the entropic risk in (\ref{eq:entropic-risk}). We choose $X \sim \mathcal{N}(\mu,\sigma^2)$ with $\mu=-1$ and $\sigma^2 = 4$. We employ Algorithm \ref{alg:saa_bisect} to estimate $\ubsr$ using $m$ samples of $X$. \Cref{figure:entropic-risk-ubsr-estimation} shows the plots of the MAE and MSE values of the estimation error vs $m$, computed across $1000$ simulations. The error bands around the MAE and MSE plots denote the standard deviation of the MAE and MSE errors respectively, for varying choices of sample size $m$. 

% \begin{figure}[ht]
%   \begin{subfigure}{.5\textwidth}
%   \centering
%     \includegraphics[width=1\linewidth]{figures/entropic-mae-mse-oce.png}
%     \caption{Estimation error as a function of number of samples}\label{figure:entropic-risk-oce-estimation}
%   \end{subfigure}%
%   \begin{subfigure}{.5\textwidth}
%   \centering
%     \includegraphics[width=1\linewidth]{figures/entropic-error-pdf-oce.png}
%     \caption{Histogram of the estimation error}\label{figure:entropic-risk-oce-distribution}
%   \end{subfigure}
%   \caption{The figure shows performance of OCE-SAA algorithm for estimation of entropic risk of a univariate Gaussian r.v. $X \sim \mathcal{N}(\mu,\sigma^2)$ with $\mu=-1$ and $\sigma^2 = 4$. Given a sample size $m$, Algorithm \ref{alg:oce_estimation} (OCE-SAA) is run with $\delta = 1/\sqrt{m}, \epsilon = 1$ and $m$ i.i.d. samples from $X$ to obtain the estimator $s_m$. For each choice of $m$, we repeat the simulation $N=1000$ times and compute the error mean and its spread (standard error) by averaging across the $N$ simulations.}
% \end{figure}

% For the choice of $u(x) = \frac{e^{\beta x} - 1}{\beta}$, $\oce$ coincides with the entropic risk in (\ref{eq:entropic-risk}). We employ Algorithm \ref{alg:oce_estimation} (OCE-SAA) to estimate $\oce$ using $m$ samples of $X$. The associated MAE and MSE bounds on the estimation error, for varying choices of sample size $m$ are given in \Cref{figure:entropic-risk-oce-estimation}. For $m=10,100$ and $1000$, we plot the histogram of the estimation error in \Cref{figure:entropic-risk-oce-distribution}. 

From the MAE and MSE error plots in \Cref{figure:entropic-risk-ubsr-estimation}, we note that our proposed estimators converge rapidly. Therefore we choose to plots these errors versus $\sqrt{m}$ instead of $m$, in order to make the error decrease discernible. %From the error distributions in the plots of \Cref{figure:entropic-risk-ubsr-distribution}, we conclude that the estimators $t_m$ and $s_m$ are asymptotically normal. 

\paragraph{Entropic risk minimization.}

We consider a portfolio optimization application with entropic risk as the objective.
In particular, we consider a $d$-dimensional random vector $\xi$, which follows a multivariate normal distribution with mean $\mu$ and covariance matrix $\Sigma$, which is positive-definite. The decision space $\Theta$ is a $d$-dimensional unit simplex. Given $\theta \in \Theta$, we are interested in optimizing the quantity $\theta^T\xi$ over $\Theta$. The quantity $\theta^T\xi$ is uncertain, owning to the randomness of $\xi$. Decision makers prefer a risk-sensitive criterion that caters to their preferences, instead of optimizing the mean $\theta^T \mu$, which is the classical risk-neutral approach.  

% \textbf{Mean-variance optimization.} A popular optimization criterion is the mean-variance objective, defined as follows. Let $\beta>0$, then the \textit{mean-variance optimization} problem is posed as
% \begin{equation}\label{eq:mean-variance-optimization}
%     \text{find }\; \theta^* \triangleq \argmin_{\theta \in \Theta} \left[-\theta^T \mu + \frac{\beta}{2} \theta^T \Sigma \theta\right]. 
% \end{equation}

For our experiments, we choose entropic risk as the optimization criterion. Consider the objective function defined as $F(\theta,\xi) \triangleq \theta^T\xi$. Since $\xi \sim \mathcal{N}(\mu,\Sigma)$, we have $F(\theta,\xi) \sim \mathcal{N}(\theta^T\mu, \theta^T\Sigma\mu), \forall \theta \in \Theta$. Using $\rho_e$ defined in (\ref{eq:entropic-risk}), we replace $X$ in (\ref{eq:entropic-risk}) with $F(\theta,\xi)$, and redefine the entropic risk as a function of $\theta$. Formally, we have $\rho_E:\Theta \to \Rel$, where $\rho_E(\theta) = \rho_e(F(\theta,\xi)), \forall \theta \in \Theta$. Then, by (\ref{eq:entropic-risk}), it follows that for every $\theta \in \Theta$,
\begin{equation}\label{eq:entropic-risk-theta}
    \rho_E(\theta) = -\theta^T\mu + \frac{\beta \theta^T\Sigma\theta}{2}. 
\end{equation}
From (\ref{eq:entropic-risk-theta}), it is easy to see that $\theta^*$ can be obtained using any standard numerical solver for a known $\mu$ and $\Sigma$.

\paragraph{Experiment setup.} 
In our setup, we set $d=5$. Using an arbitrary vector $\mu \in \Rel^d$ and arbitrary, positive-definite matrix $\Sigma \in \Rel^d \times \Rel^d$, we define $\xi \sim \mathcal{N}(\mu,\Sigma)$. The choices for $\mu,\Sigma$ are governed by a distribution underlying the \texttt{make\_spd\_matrix} function of \texttt{scikit-learn} python package. We employ the convex optimization solver in the \texttt{pyportfolioopt} Python package to obtain $\theta^*$.  

We test the performance of the \Cref{alg:ubsr-minimization} (UBSR-SG). We define $\srth{} \triangleq \sr{l}{\lambda}{F(\theta,\xi)}$ and choose $l(x) = e^{\beta x}$ and $\lambda = 1$. In this case, $\srt$ coincides with $\rho_E$ and therefore, $\theta^*$ is also the minima of $\srt(\cdot)$. We choose $\beta = 0.5$ and run the algorithm for 500 epochs. At every epoch $k$, we draw $2k$ samples from $\xi$ to construct the UBSR gradient estimator. We carry out the SG update with step size $1/\sqrt{k}$, where we project the iterate back to the $d$-dimensional simplex. To construct the gradient estimator at each epoch $k$, we form an estimator of the UBSR by running Algorithm \ref{alg:saa_bisect} (UBSR-SB) with the first $k$ samples and parameter $\delta = 1/\sqrt{k}$. \Cref{figure:entropic-risk-ubsr-optimization} shows the error plots for convergence of the iterates of \Cref{alg:ubsr-minimization} (UBSR-SG) to optima, i.e., $\theta_k \to \theta^*$.

% \begin{figure}[ht]
%   \centering
%     \includegraphics[width=0.7\linewidth]{figures/entropic-optimization-ubsr.png}
%     \caption{Performance of Algorithm \ref{alg:ubsr-minimization} (UBSR-SG) for minimizing entropic risk.}\label{figure:entropic-risk-ubsr-optimization}
%   % \caption{Performance of UBSR-SB algorithm for estimation of entropic risk of a univariate Gaussian r.v. $X = \mathcal{N}(\mu,\sigma^2)$ with $\mu=-1$ and $\sigma^2 = 4$. Given a sample size $m$, Algorithm \ref{alg:ubsr-minimization} (UBSR-SB) is run with $\delta = 1/\sqrt{m}, \epsilon = 1$ and $m$ i.i.d. samples from $X$ to obtain the estimator $s_m$. For each choice of $m$, we repeat the simulation $N=1000$ times and compute the error mean and its spread (standard error) by averaging across the $N$ simulations.}
% \end{figure}
\begin{figure*}[htbp]
  \centering

  % --- Subfigure 1 ---
  \begin{minipage}{0.48\linewidth}
    \centering
    \includegraphics[width=\linewidth]{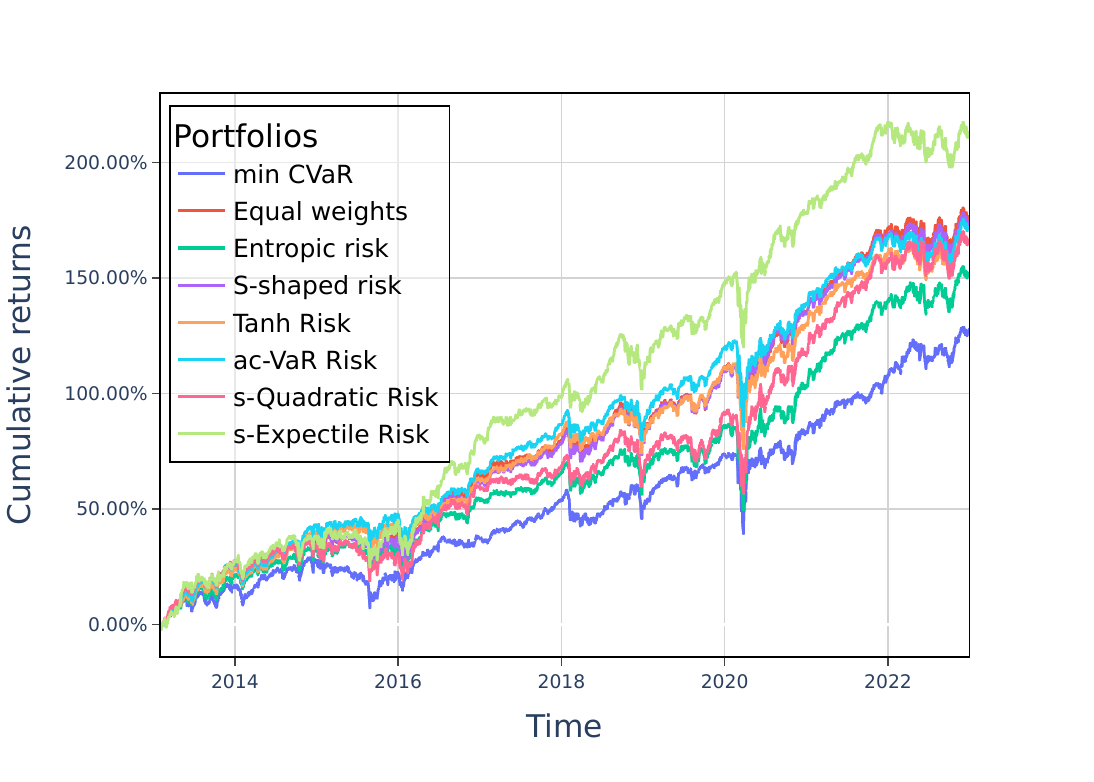}
    \caption{Cumulative portfolio returns for the S\&P market given by \Cref{alg:ubsr-minimization-offline}}
    \label{figure:ubsr-portfolio-returns-sp500}
  \end{minipage}%
\hfill
  % --- Subfigure 2 ---
    \begin{minipage}{0.48\linewidth}
    \centering
    \includegraphics[width=\linewidth]{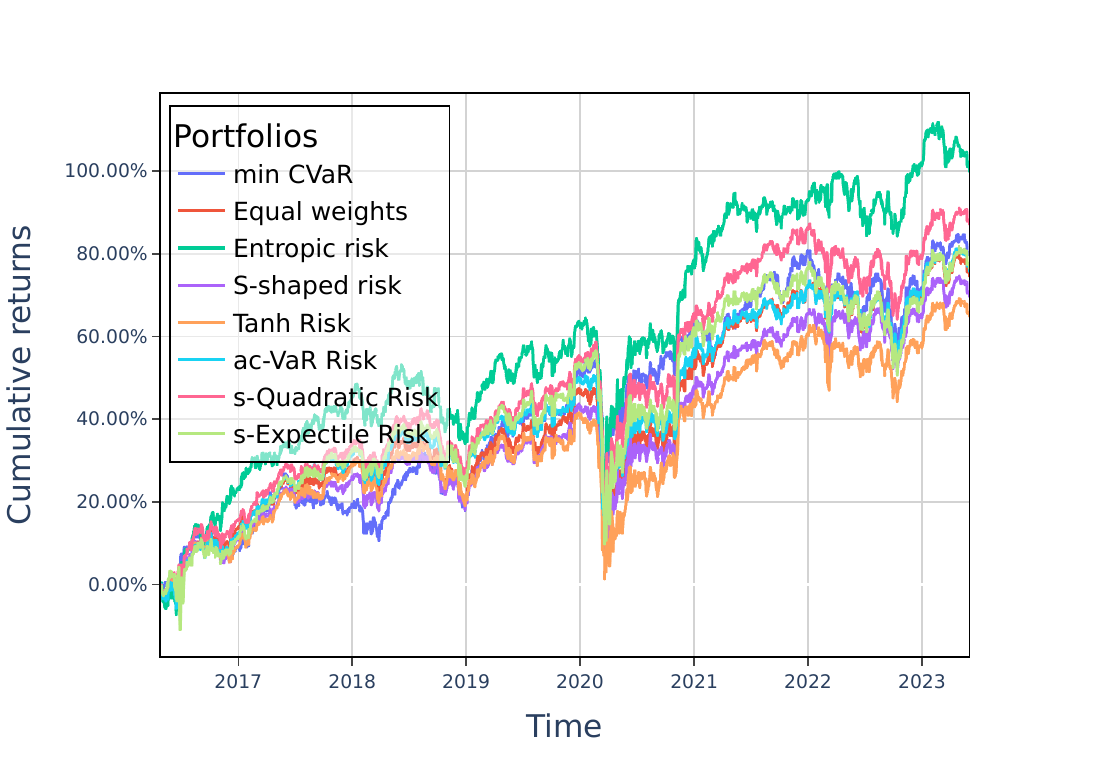}
    \caption{Cumulative portfolio returns for the FTSE market given by \Cref{alg:ubsr-minimization-offline} }
    \label{figure:ubsr-portfolio-returns-third}
  \end{minipage}

  \vfill

  % --- Subfigure 3 ---
  
    \begin{minipage}{0.48\linewidth}
    \centering
    \includegraphics[width=\linewidth]{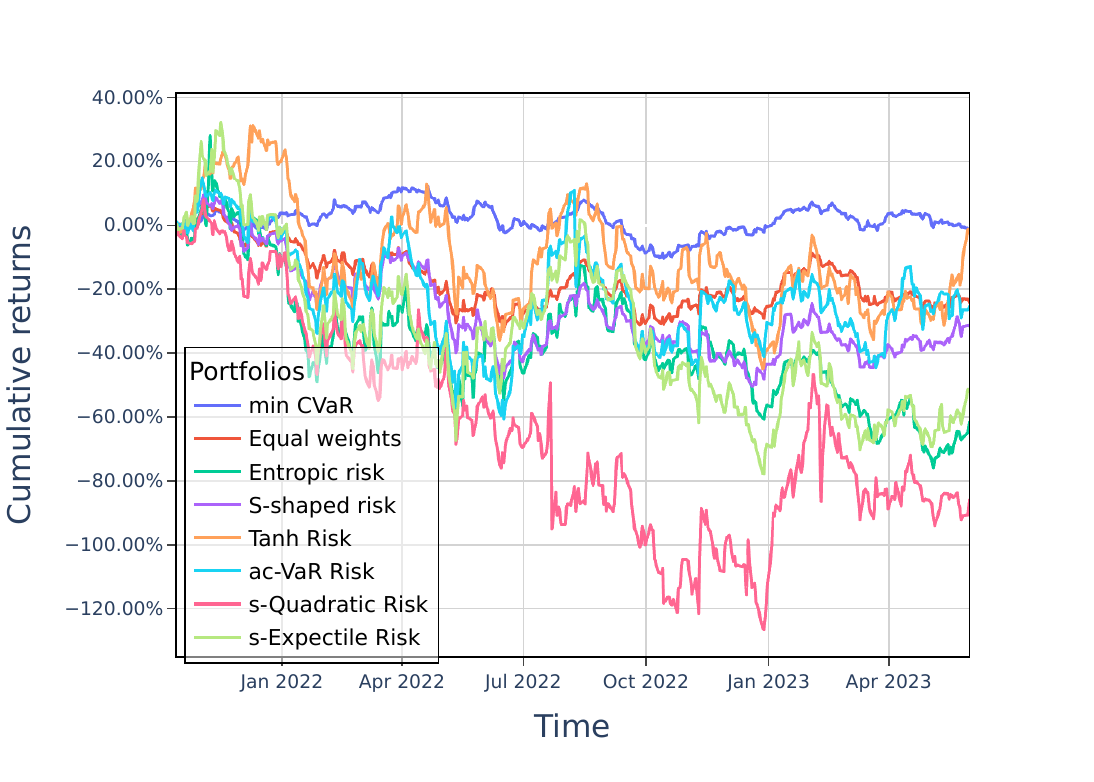}
    \caption{Cumulative portfolio returns for the NASDAQ market given by \Cref{alg:ubsr-minimization-offline} }
    \label{figure:ubsr-portfolio-returns-nasdaq}
  \end{minipage}

  % --- Main Caption ---
  \vspace{0.3cm}
  \caption{Performance of  \Cref{alg:ubsr-minimization-offline} (UBSR-SG OFFLINE) for a variety of UBSR risk measures in a portfolio optimization application}
  \label{figure:port-opt-ubsr}
\end{figure*}
% \begin{figure*}
%   \begin{tabular}{ccc}
%   \begin{subfigure}{.5\textwidth}
%   \centering
%     \includegraphics[width=1\linewidth, height=0.2\textheight]{figures/portfolio_returns_sp500.pdf}
%     \caption{Cumulative portfolio returns for \\S\&P given by UBSR-SG algorithm.}\label{figure:ubsr-portfolio-returns-sp500}
%   \end{subfigure}
%   &
%   \begin{subfigure}{.5\textwidth}
%   \centering
%     \includegraphics[width=1\linewidth, height=0.2\textheight]{figures/portfolio_returns_ftse.pdf}
%     \caption{Cumulative portfolio returns for \\ NASDAQ given by UBSR-SG algorithm.}\label{figure:ubsr-portfolio-returns-nasdaq}
%   \end{subfigure}
%   \end{tabular}
%   \caption{Performance of UBSR-SG algorithm for a variety of UBSR risk measures in a portfolio optimization application. The risk measures differ in the choice of loss functions and threshold $\lambda$, see \Cref{tab:ubsr_instances} for the details.}\label{figure:port-opt-ubsr}
% \end{figure*}
\subsection{Portfolio Optimization}
\label{sec:portfolio-expts}
Suppose we have a set of $d$ assets in a financial market. Let the random vector $\mathbf{\xi} \in \Rel^d$ denote asset-wise market returns. Given an asset allocation $\theta \in \Theta$, the random variable $F(\theta, \mathbf{\xi})\triangleq \xi^T\theta$ denotes portfolio returns. for a portfolio optimization (PO) setting. Our PO implementation is based on the \texttt{skfolio} python library. We test the performance of the UBSR-SG algorithm across three stock market datasets: 'Standard and Poor's 500 (S\&P500)', 'Financial Times Stock Exchange (FTSE100)', and 'Nasdaq'. The 'S\&P 500' dataset is composed of the daily prices of 20 assets from the 'S\&P 500' composition starting from 1990-01-02 up to 2022-12-28. The 'FTSE100' dataset contains the daily prices of 64 assets from the 'FTSE100' index, from 2000-01-04 to 2026-05-26. The 'Nasdaq' dataset contains the daily prices of 1455 assets from the 'Nasdaq' index, from 2018-01-02 to 2026-05-26.

We test the performance of the \Cref{alg:ubsr-minimization-offline} (UBSR-SG-OFFLINE) on three different choices for stock market data: 'Standard and Poor's 500 (S\&P500)', 'Financial Times Stock Exchange (FTSE100)', and 'Nasdaq'. For each of the above datasets, we run \Cref{alg:ubsr-minimization-offline} for a variety of UBSR loss functions and obtain the corresponding UBSR-optimal portfolios, including the popular choices like entropic risk, quadratic risk and the expectile risk. A comparison of the out-of-sample performance of the portfolios generated by these risk measures against two popular benchmarks: the equal-weighted portfolio and the minimum CVaR portfolio, is given in \Cref{figure:port-opt-ubsr}. These figures show that the portfolios generated by the UBSR-SG are either comparable to or outperform the benchmarks. Apart from the aforementioned popular special cases of UBSR, we also propose new instances of the UBSR by employing new choices for the loss function $l$. In \Cref{tab:ubsr_instances}, we provide the choice of loss function $l$ and risk threshold $\lambda$ that we use to obtain the portfolios given in \Cref{figure:port-opt-ubsr}. 

 % In particular, the Algorithm \ref{alg:ubsr-minimization} (UBSR-SG) is run for 10000 epochs, and at every epoch $k$, we draw $2k$ samples from the available stock market data to construct the UBSR gradient estimator. The $2k$ samples are obtained after infusing a zero-mean Gaussian noise that is proportional to the variance of the data. We carry out the SG update with step size $1/\sqrt{k}$, where we project the iterate back to the $d$-dimensional simplex. To construct the gradient estimator at each epoch $k$, we form an estimator of the UBSR by running Algorithm \ref{alg:saa_bisect} (UBSR-SB) with the first $k$ samples and parameter $\delta = 1/\sqrt{k}$.

\begin{table}
    \centering
    \caption{Choice of UBSR parameters for portfolio optimization.}\label{tab:ubsr_instances}
    \begin{tabular}{llr}
        \toprule
        Portfolio name  & Loss function $l$ & Risk Threshold $\lambda$\\
        \midrule
        Entropic risk & $l(x) = e^{\beta x}, \beta = 0.4$ & $1.$\\
        \hline
        S-shaped risk & $l(x) = \frac{2 x^2 tan^{-1}(x)}{\pi}$ & $0.$\\
        \hline
        Tanh risk & $l(x) = x^2 tanh(x)$ & $0.$\\
        \hline
        ac-VaR risk & $l(x) = \frac{tan^{-1}(\beta x)}{\pi}+0.5$ & $0.05$\\
        \hline
        s-Quadratic risk & $l(x) = \begin{cases}
                                    \alpha \; x \; ln(1+e^{x}) &\quad \text{if} x\geq 0\\
                                    \alpha \;ln(2) \;x &\quad \text{o.w.} 
                                    \end{cases}$ & $0.$\\
        \hline
        s-Expectile risk & $l(x) = \begin{cases}
                                    x \; (1+ \alpha \;tan^{-1}(x)) &\quad \text{if} x\geq 0\\
                                    (1-\alpha)x + \alpha \; tan^{-1}(x)  &\quad \text{o.w.}
                                    \end{cases}$ & $0.$\\
    \bottomrule
    \end{tabular}
\end{table}

% \begin{figure}[!ht]
%   \begin{subfigure}{.55\textwidth}
%   \centering
%     \includegraphics[width=1\linewidth]{figures/ubsr-composition-ftse.png}
%     \caption{Portfolio compositions of FTSE assets \\ given by the UBSR-SG algorithm.}\label{figure:ubsr-portfolio-composition-ftse}
%   \end{subfigure}%
%   \begin{subfigure}{.45\textwidth}
%   \centering
%     \includegraphics[width=1\linewidth]{figures/ubsr-returns-ftse.png}
%     \caption{Cumulative portfolio returns for FTSE \\ portfolios given by UBSR-SG algorithm.}\label{figure:ubsr-portfolio-returns-ftse}
%   \end{subfigure}
%   \caption{The figure shows the performance of UBSR-SG algorithm for a variety of UBSR risk measures in a portfolio optimization application sourced from the FTSE stock market data. The risk measures differ in the choice of loss functions and threshold $\lambda$. For each such choice, Algorithm \ref{alg:ubsr-minimization} (UBSR-SG) is run for 10000 epochs, and at every epoch $k$, we draw $2k$ samples from the available stock market data to construct the UBSR gradient estimator. The $2k$ samples are obtained after infusing a zero-mean Gaussian noise that is proportional to the variance of the data. We carry out the SG update with step size $1/\sqrt{k}$, where we project the iterate back to the $d$-dimensional simplex. To construct the gradient estimator at each epoch $k$, we form an estimator of the UBSR by running Algorithm \ref{alg:saa_bisect} (UBSR-SB) with the first $k$ samples and parameter $\delta = 1/\sqrt{k}$.}
% \end{figure}

\section{Conclusions}
\label{section:conclusions}
We laid the foundations for UBSR estimation and optimization for unbounded random variables.
We proposed and analyzed algorithms for UBSR estimation with provable non-asymptotic guarantees.
Next, for a parameterized class of distributions, we derived an expression for the gradient of UBSR. This expression led to a gradient estimation scheme using i.i.d. samples, which was subsequently used in stochastic gradient algorithms for UBSR optimization. We provided non-asymptotic error bounds that quantify the convergence of our algorithms to global optima under a strong convexity assumption and a general convexity assumption on the UBSR objective. For the non-convex UBSR objective, we establish first-order convergence under a smoothness assumption. 
Our contributions are appealing in financial applications, such as portfolio optimization and risk-sensitive online learning. We verified the former empirically, using real-world stock market datasets from the financial domain. We also conducted experiments to validate our estimation and optimization schemes using synthetic data for the entropic risk, a special case of UBSR.
 
 As future work, it would be interesting to develop Newton-based methods for UBSR optimization using the UBSR gradient expression that we have derived.

% \section*{Acknowledgment}

% The preferred spelling of the word ``acknowledgment'' in American English is 
% without an ``e'' after the ``g.'' Use the singular heading even if you have 
% many acknowledgments. Avoid expressions such as ``One of us (S.B.A.) would 
% like to thank $\ldots$ .'' Instead, write ``F. A. Author thanks $\ldots$ .'' In most 
% cases, sponsor and financial support acknowledgments are placed in the 
% unnumbered footnote on the first page, not here.

\begin{appendices}

\section{Minor technical results}
\label{supp:minor-results}

\subsection{Taking norm inside summation}
\begin{lemmma}\label{lemma-normed-sum}
    Given vectors $\{a_i\}_{i=1}^m$ in any normed space $\norm{\cdot}$, then for any $n\ge 1$ following holds:
    \begin{equation*}
         \norm{\sum_m a_i}^n \leq  m^{n-1} \sum_m \norm{a_i}^n.
    \end{equation*}
\end{lemmma}

\begin{proof}
    \begin{align*}
        \norm{\sum_m a_i}^n \leq \left(\sum_m \norm{a_i}\right)^n 
        = m^n \left(\sum_m \frac{\norm{a_i}}{m}\right)^n 
        \leq m^n \frac{\sum_m \norm{a_i}^n}{m} 
        = m^{n-1} \sum_m \norm{a_i}^n
    \end{align*}
    Here, the first inequality is Minkowski's inequality, while the second inequality follows from Jensen's inequality applied to the function $x \mapsto x^n$, which is convex for $x \geq 0, n \ge 1$.
\end{proof}

\subsection{Taking norm inside expectation of product of two R.V.s.}
\begin{lemmma}\label{lemma:uvw}
    Let $p \in [1,\infty)$. Suppose $U$ is a real-valued random variable whose absolute value is bounded above by $w>0$ and let $\mathbf{V}$ be an $n$-dimensional random vector whose $p^{th}$ moment exists and is finite. Then, 
    \begin{equation*}
        \norm{\Exp[U\mathbf{V}]}_p \leq w \norm{\mathbf{V}}_{L_p}
    \end{equation*}
\end{lemmma}
\begin{proof}We have
    \begin{align*}
    \norm{\Exp[U\mathbf{V}]}_p = \Bigg( \sum_{i=1}^n \Big| \Exp\left[UV_i\right]  \Big|^p \Bigg)^\frac{1}{p} 
    \leq \Bigg( \sum_{i=1}^n \Exp\left[ | UV_i |\right]^p  \Bigg)^\frac{1}{p} &\leq w \Bigg( \sum_{i=1}^n \Exp\left[ | V_i|\right]^p \Bigg)^\frac{1}{p} \\
    &= w \Bigg( \Exp\left[ \sum_{i=1}^n | V_i|^p \right]\Bigg)^\frac{1}{p}
    = w \norm{\mathbf{V}}_{L_p},
\end{align*}
where the interchange of summation and expectation in the second equality follows because the $p^{th}$ moment of the random vector is finite.
\end{proof}
\begin{lemmma}\label{lemma:uv-2-norm}
    Suppose the random variable $U$ has bounded $2^{nd}$ moment: $\norm{U}_{\Leb{2}} \leq M < \infty$ and let $\mathbf{V}$ be an $n$-dimensional random vector with finite $2^{nd}$ moment. Then,
    \begin{align*}
        \norm{\Exp\left[U\mathbf{V}\right]}_2 \leq M \norm{\mathbf{V}}_{\Leb{2}}.
    \end{align*}
\end{lemmma}
\begin{proof}
    \begin{align*}
        &\norm{\Exp[U\mathbf{V}]}_2 \\
        &= \left( \sum_{i=1}^n \left| \Exp\left[UV_i\right]  \right|^2 \right)^\frac{1}{2} \leq \left( \sum_{i=1}^n \Exp\left[U^2\right]\Exp\left[V_i^2\right]  \right)^\frac{1}{2} \leq M \left( \Exp\left[\sum_{i=1}^n \left[V_i^2\right]\right]\right)^\frac{1}{2}
    = M \norm{\mathbf{V}}_{\Leb{2}},
    \end{align*}
    where the first inequality is the Cauchy-Schwartz inequality. The second inequality follows from the moment bound on $U$, and the interchange of expectation and summation for the second inequality follows because $\mathbf{V}$ has finite $2^{nd}$ moment.  
\end{proof}
\subsection{Bounding Variance of a Lipschitz Function}
\begin{lemmma}\label{lemma:lip-function-variance-bound}
    Let $f:\Rel \to \Rel$ be S-Lipschitz, and let $X$ be a real-valued random variable, with variance $\sigma^2$. Then, variance of $f(X)$ is bounded above by $S^2\sigma^2$. 
\end{lemmma}
\begin{proof}Let $Y$ be an identical copy of $X$. Then we have
    \begin{align*}
        \textrm{Var}\left(f(X)\right) &= \Exp\left[f(X)^2\right] - \left(\Exp\left[f(X)\right]\right)^2 \\
        &= \Exp\left[f(X)^2\right] - 2\left(\Exp\left[f(X)\right]\right)^2 + \Exp\left[f(Y)^2\right] - \Exp\left[f(Y)^2\right] + \left(\Exp\left[f(X)\right]\right)^2 \\
        &= \Exp_X\left[\Exp_Y\left[\left|f(X) - f(Y)\right|^2 \right]\right] - \textrm{Var}\left(f(Y)\right),
    \end{align*}
    where the last equality follows by change of variables. Rearranging the above equality and using the Lipschitz assumption on $f$, we have
    \begin{align*}
        2\textrm{Var}\left(f(X)\right) \leq S^2 \Exp_X\left[\Exp_Y\left[(X -Y)^2 \right]\right] &= S^2 \Exp_X\left[X^2 - 2X \Exp[Y] + \Exp\left[Y^2\right]\right] \\
        &= S^2 \left[\text{Var}(X)+\text{Var}(Y)\right].
    \end{align*}
     Since $X$ and $Y$ are identical, the claim of the lemma follows.
\end{proof}

\end{appendices}

% \appendix
% \section{Proofs}
% \label{sec:proofs}
% \input{inputs/proofs-ubsr}

\bibliographystyle{plainnat}
\bibliography{references}

\end{document}